\documentclass[format=acmsmall, review=false]{acmart}
\usepackage{acm-ec-20}
\usepackage{booktabs} % For formal tables
\usepackage[ruled]{algorithm2e} % For algorithms
\usepackage{algorithmic}

\SetAlFnt{\small}
\SetAlCapFnt{\small}
\SetAlCapNameFnt{\small}
\SetAlCapHSkip{0pt}
\IncMargin{-\parindent}

\usepackage{graphicx}
\usepackage{xcolor}
\usepackage{amsmath}
\usepackage{booktabs}
\usepackage{dsfont}
\usepackage{enumerate}
\usepackage{multirow}
\usepackage{nicematrix}
\usepackage{amsthm}
\usepackage{amsfonts}
\usepackage{pifont}
\usepackage{mathrsfs}
\usepackage{enumitem}
\usepackage{nicefrac}
\usepackage{tikz,pgfplots}
\usepackage{natbib}
\usepackage{thm-restate}
\usepackage{mathtools}
\usepackage{physics}
\usepackage[capitalize]{cleveref}
\usepackage{bbm}

\setcitestyle{authoryear}
\allowdisplaybreaks

\newtheorem{theorem}{Theorem}[section]
\newtheorem{corollary}[theorem]{Corollary}
\newtheorem{proposition}[theorem]{Proposition}
\newtheorem{lemma}[theorem]{Lemma}

\theoremstyle{definition}
\newtheorem{definition}[theorem]{Definition}
\newtheorem{example}[theorem]{Example}

\usetikzlibrary{positioning}
\usetikzlibrary{arrows.meta}
\tikzset{>={Latex[width=1.5mm,length=1.5mm]}}
\pgfplotsset{compat=1.18}

\DeclareMathOperator*{\argmax}{arg\,max}
\DeclareMathOperator*{\argmin}{arg\,min}

\newcommand{\Z}{\mathbb{Z}}
\newcommand{\Zplus}{\Z_{>0}}
\newcommand{\Zpluszero}{\Z_{\geq 0}}
\newcommand{\R}{\mathbb{R}}
\newcommand{\Rplus}{\R_{>0}}
\newcommand{\Rpluszero}{\R_{\geq 0}}
\newcommand{\inlinefrac}[2]{#1/#2}

\renewcommand{\emptyset}{\varnothing}

\newcommand{\condnorealgeneral}{1}
\newcommand{\condnorealnormtwo}{2}
\newcommand{\condnointidengood}{3}
\newcommand{\condnointbinary}{4}
\newcommand{\condnointtwovalue}{5}
\newcommand{\condnointgeneral}{6}

\newcommand{\condrealgeneral}{Condition~\condnorealgeneral}
\newcommand{\condrealnormtwo}{Condition~\condnorealnormtwo}
\newcommand{\condintidengood}{Condition~\condnointidengood}
\newcommand{\condintbinary}{Condition~\condnointbinary}
\newcommand{\condinttwovalue}{Condition~\condnointtwovalue}
\newcommand{\condintgeneral}{Condition~\condnointgeneral}

\title[On the Fairness of Additive Welfarist Rules]{On the Fairness of Additive Welfarist Rules}

\author{Karen Frilya Celine}
\affiliation{%
  \institution{National University of Singapore, Singapore}}

\author{Warut Suksompong}
\affiliation{%
  \institution{National University of Singapore, Singapore}}

\author{Sheung Man Yuen}
\affiliation{%
  \institution{National University of Singapore, Singapore}}

\begin{abstract}
Allocating indivisible goods is a ubiquitous task in fair division.
We study \emph{additive welfarist rules}, an important class of rules which choose an allocation that maximizes the sum of some function of the agents' utilities.
Prior work has shown that the maximum Nash welfare (MNW) rule is the unique additive welfarist rule that guarantees envy-freeness up to one good (EF1).
We strengthen this result by showing that MNW remains the only additive welfarist rule that ensures EF1 for identical-good instances, two-value instances, as well as normalized instances with three or more agents.
On the other hand, if the agents' utilities are integers, we demonstrate that several other rules offer the EF1 guarantee, and provide characterizations of these rules for various classes of instances.
\end{abstract}

\begin{document}

\maketitle

% Paper body

\section{Introduction}
\label{sec:intro}

The distribution of limited resources to interested agents is a fundamental problem in society.
Fairness often plays a key role in the distribution process---whether it be dividing inheritance between family members, allocating the national budget across competing sectors, or sharing credit and responsibility amongst participants in collaborative efforts.
The complexity of this problem has led to the research area of \emph{fair division}, which develops methods and algorithms to ensure that all agents feel fairly treated and has been studied in both economics and computer science \citep{BramsTa96,RobertsonWe98,Moulin03,Moulin19,Thomson16}.

What does it mean for an allocation of the resources to be ``fair''?
The answer depends on the context and the specific fairness benchmark applied.
Among the plethora of fairness criteria that have been proposed, one of the most prominent is \emph{envy-freeness}, which states that no agent should prefer another agent's share to her own \citep{Foley67,Varian74}.
In the ubiquitous setting of allocating \emph{indivisible goods}---such as houses, cars, jewelry, and artwork---it is sometimes infeasible to attain envy-freeness.
For example, if a single valuable good is to be allocated between two or more agents, only one of the agents can receive the good, thereby incurring envy from the remaining agents.
In light of this, envy-freeness is often relaxed to \emph{envy-freeness up to one good (EF1)}.
The EF1 criterion allows an agent to envy another agent provided that the removal of some good in the latter agent's bundle would eliminate the former agent's envy.
It is known that an EF1 allocation always exists, and that such an allocation can be computed in polynomial time \citep{LiptonMaMo04,Budish11,CaragiannisKuMo19}.

In addition to fairness, another important property of allocations is \emph{efficiency}.
A well-known efficiency notion is \emph{Pareto optimality (PO)}, which stipulates that no other allocation makes some agent better off and no agent worse off.
In a far-reaching work, \citet{CaragiannisKuMo19} proved that a \emph{maximum Nash welfare (MNW)} allocation, which maximizes the product of the agents' utilities across all possible allocations,%
\footnote{This is equivalent to maximizing the sum of the logarithm of the agents' utilities.}
satisfies both EF1 and PO, thereby offering fairness and efficiency simultaneously.
In spite of this, the MNW rule does have limitations.
For example, it avoids giving an agent zero utility at all costs---an allocation that gives most agents a tiny positive utility is preferred to another allocation that gives one agent zero utility and every other agent a large utility, as illustrated in the following example.

\begin{example}
\label{ex:mnw}
Consider an instance with $n \geq 4$ agents and $n$~goods $g_1,\dots,g_n$ such that the utilities of the goods are as follows.
\begin{itemize}
    \item $u_1(g_1) = n$.
    \item For $i \in \{2, \ldots, n\}$, $u_i(g_{i-1}) = n-1$ and $u_i(g_i) = 1$.
    \item $u_i(g) = 0$ for all other pairs $(i, g)$.
\end{itemize}
The only allocation $\mathcal{A}$ that gives positive utility to every agent is the one that assigns good $g_i$ to agent $i$ for all $i \in \{1, \ldots, n\}$---this is also the MNW allocation.
The sum of the agents' utilities in this allocation is $n + (n-1) = 2n-1$.
On the other hand, consider the allocation $\mathcal{B}$ that gives good $g_1$ to agent $1$, $g_i$ to agent $i+1$ for $i \in \{2, \ldots, n-1\}$, and $g_n$ to agent $n$.
The sum of the agents' utilities in this allocation is $n^2-2n+3$, which is much larger than that in $\mathcal{A}$.
Not only is $\mathcal{B}$ also EF1 and PO, but it also ensures that one agent receives the same utility, all remaining agents except one receive much higher utility, while the exceptional agent receives only marginally lower utility.
\end{example}

The MNW rule belongs to the class of \emph{additive welfarist rules}---rules that choose an allocation that maximizes the sum of some function of the agents' utilities for their bundles \cite[p.~67]{Moulin03}.
Additive welfarist rules have the advantage that they satisfy PO by definition.
\citet{Suksompong23} showed that MNW is the only additive welfarist rule%
\footnote{\citet{YuenSu23} extended this result by showing that MNW is the only (not necessarily additive) welfarist rule that guarantees EF1 for all instances.} that guarantees EF1 for all instances.
Nevertheless, there exist other additive welfarist rules that guarantee EF1 for restricted classes of instances.
For example, \citet[Appendix C]{MontanariScSu25} proved that the \emph{maximum harmonic welfare (MHW)} rule ensures EF1 for the class of \emph{integer-valued} instances.
Note that this is an important class of instances in practice---for example, the popular fair division website Spliddit \citep{GoldmanPr14} only allows each user to specify integer values for goods summing up to $1000$.
In \Cref{ex:mnw}, allocation $\mathcal{B}$ is the MHW allocation, and we saw earlier that $\mathcal{B}$ has certain advantages compared to the MNW allocation $\mathcal{A}$.
Beyond integer-valued instances, \citet{EckartPsVe24} showed that certain additive welfarist rules based on \emph{$p$-mean} functions, besides the MNW rule, also yield EF1 for the class of \emph{normalized} instances with two agents---normalized instances have the property that the utility for the entire set of goods is the same for every agent.
By the remark above, all instances on Spliddit are normalized.

In this paper, we study the necessary and sufficient conditions for additive welfarist rules to guarantee EF1 for different classes of instances.
These classes include integer-valued instances, identical-good instances, two-value instances, normalized instances, as well as some of their combinations.
We also establish relationships between the conditions for different classes of instances.

\subsection{Our Results}

In our model, we assume that each agent has an additive utility function over a set of indivisible goods.
We consider additive welfarist rules, each of which is defined by a function $f$; the rule chooses an allocation that maximizes the sum of the function $f$ applied to the agents' utilities for their own bundles.
In order for the additive welfarist rules to ensure that agents receive as much utility as possible, we assume that $f$ is strictly increasing.%
\footnote{
  If $f$ is not strictly increasing, then the allocation chosen by the additive welfarist rule may not even be PO.
  In fact, we show in Appendix~\ref{ap:welfare_function_strictly_increasing} that even if $f$ is \emph{non-decreasing} but not \emph{strictly increasing}, then there exists an identical-good normalized instance such that the additive welfarist rule with function $f$ does \emph{not} always choose an EF1 allocation.
  This means that it is desirable for $f$ to be strictly increasing.
}
For our characterizations, we only consider instances that are \emph{positive-admitting}, i.e., instances in which there exists an allocation that gives positive utility to every agent.%
\footnote{
  Even the MNW rule does not always return EF1 for instances that are not positive-admitting, unless there are additional tie-breaking mechanisms \citep{CaragiannisKuMo19}.
}
Our model is described formally in \Cref{sec:prelim}.

As one of the highlights of our paper, we extend the result of \citet{Suksompong23} that the MNW rule is the unique additive welfarist rule that guarantees EF1 for all (positive-admitting) instances.
In particular, we show that the same result still holds even if we restrict the class of instances to only identical-good instances, where each agent values every good equally but the value may differ between agents.
This solidifies the advantage of MNW over other additive welfarist rules and indicates that the unique fairness of MNW stems from its \emph{scale-invariance},\footnote{That is, scaling the utility function of an agent by a constant factor does not change the outcome of the MNW rule.} since for normalized identical-good instances, an additive welfarist rule with any strictly concave function $f$ guarantees EF1.
The proof of this result is technically involved, requiring a careful mathematical analysis of the growth rate of functions.
The theorem statement and its proof can be found in \Cref{subsec:real_idengood}.

To provide a more complete picture, we complement this result by studying other classes of instances.
In \Cref{subsec:real_twovalue,subsec:real_normalized}, we show that MNW remains the unique additive welfarist rule that ensures EF1 for the classes of two-value instances as well as normalized instances with three or more agents.
Our result on normalized instances with three or more agents extends a result of \citet[Theorem 6]{EckartPsVe24}, which only handles additive welfarist rules that are $p$-mean rules.
For normalized instances with two agents, however, this characterization ceases to hold, and we provide a necessary condition and a sufficient condition for additive welfarist rules to guarantee EF1.

In \Cref{sec:integer}, we consider integer-valued instances.
For such instances, not only the MNW rule but also the MHW rule guarantees EF1 \citep{MontanariScSu25}.
We shall attempt to characterize rules with this property.
We show that for identical-good, binary, and two-value instances (which are also integer-valued), the respective functions defining the additive welfarist rules are characterized by Conditions {\condnointidengood}, {\condnointbinary}, and {\condnointtwovalue} respectively (see \Cref{def:conditions}).
This allows much larger classes of functions beyond the MNW rule or the MHW rule.
We provide examples and non-examples of $p$-mean rules as well as additive welfarist rules defined by modified logarithmic functions and modified harmonic functions that guarantee EF1 for these classes of instances.
For the larger class of all integer-valued instances, we give a necessary condition and a sufficient condition for such additive welfarist rules.

In general, our characterizations are proven using the following outline:
When a certain instance admits a unique EF1 allocation, an additive welfarist rule guaranteeing EF1 needs to return this allocation.
In particular, moving some good to a different bundle in such a way that the allocation is no longer EF1 must decrease the social welfare.
That is, the increase in welfare caused by adding the good to the new bundle must be strictly less than the decrease in welfare due to the removal of the good from the original bundle.
This imposes a restriction on the growth rate of a function defining an additive welfarist rule guaranteeing EF1.

For all of our characterization results (Theorems~\ref{thm:real_idengood}, 
\ref{thm:real_twovalue}, \ref{thm:real_normalized_three}, \ref{thm:integer_idengood}, \ref{thm:integer_binary}, and \ref{thm:integer_twovalue}), we show that if a function defining an additive welfarist rule guarantees EF1 for that particular class of instance for \emph{some} number of agents $n$, then it also guarantees EF1 for that class for \emph{every} number of agents $n$.

Our results are summarized in Table~\ref{tab:results}.
The relationships between the different conditions are summarized in \Cref{fig:implications}, and the relationships between additive welfarist rules that guarantee EF1 for different classes of instances are illustrated in \Cref{fig:venn_diagram}.

\renewcommand{\arraystretch}{1.2}

\begin{table*}[t]
\centering
\resizebox{\textwidth}{!}{
\begin{NiceTabular}{|c||c|c|}
\hline 
    & \textbf{Real-valued}
    & \textbf{Integer-valued} \\ \hline \hline
\textbf{Identical-good}
    & Logarithmic (Th.~\ref{thm:real_idengood})
    & {\condintidengood} (Th.~\ref{thm:integer_idengood}) \\ \hline
\textbf{Binary}
    & {\condintbinary} (Th.~\ref{thm:integer_binary})
    & {\condintbinary} (Th.~\ref{thm:integer_binary}) \\ \hline
\textbf{Two-value}
    & Logarithmic (Th.~\ref{thm:real_twovalue})
    & {\condinttwovalue} (Th.~\ref{thm:integer_twovalue}) \\ \hline
\textbf{Normalized}
    & $n \geq 3$: Logarithmic (Th.~\ref{thm:real_normalized_three})
    & ? \\ \hline
\textbf{No other restrictions}
    & Logarithmic \citep{CaragiannisKuMo19,Suksompong23}
    & ? \\
\hline
\end{NiceTabular}
}
\vspace{2mm}
\caption{
  Conditions on the functions defining the additive welfarist rules that guarantee EF1 allocations for different classes of instances (see \Cref{def:conditions}).
  The logarithmic function corresponds to the MNW rule.
  The symbol ``?'' means that the precise condition remains unknown.
  For normalized instances with two agents, concaveness is a necessary condition (\Cref{prop:real_normalized_two_necessary}) and {\condrealnormtwo} is a sufficient condition (\Cref{prop:real_normalized_two_sufficient}).
  For integer-valued instances with no other restrictions, {\condintgeneral}a is a necessary condition (\Cref{prop:integer_general_necessary}) and {\condintgeneral}b is a sufficient condition (\Cref{prop:integer_general_sufficient}).  
} 
\label{tab:results}
\end{table*}

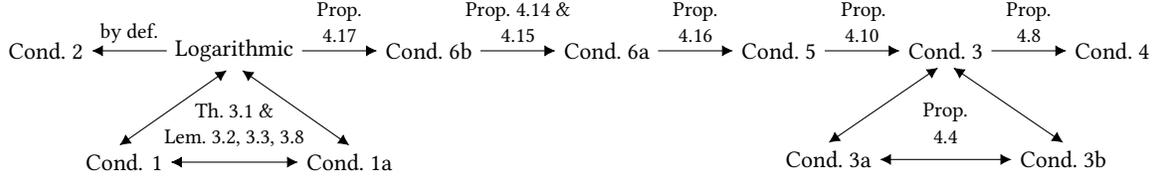
\begin{figure*}[t]
\centering
\begin{tikzpicture}
\fontsize{7.5}{9}\selectfont
\node (logarithmic) {Logarithmic};
\node[text width={width("Cond.~{\condnorealgeneral}a")}, align=right] (cond1) [below left=1cm and -0.05cm of logarithmic] {Cond.~{\condnorealgeneral}};
\node (cond1a) [below right=1cm and -0.05cm of logarithmic] {Cond.~{\condnorealgeneral}a};
\node (cond2) [left=of logarithmic] {Cond.~{\condnorealnormtwo}};
\node (cond6b) [right=of logarithmic] {Cond.~{\condnointgeneral}b};
\node (cond6a) [right=of cond6b] {Cond.~{\condnointgeneral}a};
\node (cond5) [right=of cond6a] {Cond.~\condnointtwovalue};
\node (cond3) [right=of cond5] {Cond.~{\condnointidengood}};
\node (cond3a) [below left=1cm and 0.27cm of cond3] {Cond.~{\condnointidengood}a};
\node (cond3b) [below right=1cm and 0.27cm of cond3] {Cond.~{\condnointidengood}b};
\node (cond4) [right=of cond3] {Cond.~{\condnointbinary}};

\draw[->] (logarithmic.west) -- node [midway, above, font=\scriptsize, align=center] {by def.} (cond2.east);
\draw[<->] (logarithmic.south)++(-0.1cm,0) -- (cond1.north);
\draw[<->] (logarithmic.south)++(0.1cm,0) -- (cond1a.north);
\draw[<->] (cond1.east) -- node [midway, above, font=\scriptsize, align=center, yshift=0.1cm] {Th.~\ref{thm:real_idengood} \& \\  Lem.~\ref{lem:real_idengood_cond1}, \ref{lem:real_idengood_constant}, \ref{lem:real_idengood_log}} (cond1a.west);
\draw[->] (logarithmic.east) -- node [midway, above, font=\scriptsize, align=center] {Prop.\\\ref{prop:integer_general_log}} (cond6b.west);
\draw[->] (cond6b.east) -- node [midway, above, font=\scriptsize, align=center] {Prop.~\ref{prop:integer_general_necessary} \& \\\ref{prop:integer_general_sufficient}} (cond6a.west);
\draw[->] (cond6a.east) -- node [midway, above, font=\scriptsize, align=center] {Prop.\\\ref{prop:integer_general_implies_twovalue}}(cond5.west);
\draw[->] (cond5.east) -- node [midway, above, font=\scriptsize, text width=1cm, align=center] {Prop.\\\ref{prop:integer_twovalue_sufficient}} (cond3.west);
\draw[<->] (cond3.south)++(-0.1cm,0) -- (cond3a.north);
\draw[<->] (cond3.south)++(0.1cm,0) -- (cond3b.north);
\draw[<->] (cond3a.east) -- node [midway, above, font=\scriptsize, align=center, yshift=0.1cm] {Prop.\\\ref{prop:condition_three_eqv}} (cond3b.west);
\draw[->] (cond3.east) -- node [midway, above, font=\scriptsize, align=center] {Prop.\\\ref{prop:integer_idengood_implies_binary}} (cond4.west);
\end{tikzpicture}
\caption{Relationships between the different conditions (see \Cref{def:conditions}).}
\label{fig:implications}
\end{figure*}

\section{Preliminaries}
\label{sec:prelim}

Let $N = \{1, \ldots, n\}$ be a set of $n \geq 2$ agents and $M = \{g_1, \ldots, g_m\}$ be a set of goods.
Each agent $i \in N$ has a \emph{utility function} $u_i : 2^M \to \Rpluszero$ such that $u_i(S)$ is $i$'s utility for a subset $S$ of goods; we write $u_i(g)$ instead of $u_i(\{g\})$ for a single good $g \in M$.
The utility function is \emph{additive}, i.e., for each $S \subseteq M$, $u_i(S) = \sum_{g \in S} u_i(g)$; in particular, $u_i(\emptyset) = 0$.
An \emph{instance} consists of $N$, $M$, and $(u_i)_{i \in N}$.

An allocation $\mathcal{A} = (A_1, \ldots, A_n)$ is an ordered partition of $M$ into $n$ bundles $A_1, \ldots, A_n$ such that $A_i$ is allocated to agent $i \in N$.
An allocation is \emph{envy-free up to one good (EF1)} if for all $i, j \in N$ with $A_j \neq \emptyset$, there exists a good $g \in A_j$ such that $u_i(A_i) \geq u_i(A_j \setminus \{g\})$.

We say that an instance is
\begin{itemize}
    \item \emph{integer-valued} if $u_i(g) \in \Zpluszero$ for all $i \in N$ and $g \in M$;
    \item \emph{identical-good} if for each $i \in N$, there exists $a_i \in \Rplus$ such that $u_i(g) = a_i$ for all $g \in M$;
    \item \emph{binary} if $u_i(g) \in \{0, 1\}$ for all $i \in N$ and $g \in M$;
    \item \emph{two-value} if there exist distinct $a_1, a_2 \in \Rpluszero$ such that $u_i(g) \in \{a_1, a_2\}$ for all $i \in N$ and $g \in M$;\footnote{Two-value instances have received attention in fair division, particularly in relation to the MNW rule \citep{AmanatidisBiFi21,AkramiChHo22}.}
    \item \emph{normalized} if $u_i(M) = u_j(M)$ for all $i, j \in N$;
    \item \emph{positive-admitting} if there exists an allocation $(A_1, \ldots, A_n)$ such that $u_i(A_i) > 0$ for all $i \in N$.
\end{itemize}

Let an instance be given.
An \emph{additive welfarist rule} is defined by a strictly increasing function $f : \Rpluszero \to \R \cup \{-\infty\}$ such that the rule chooses an allocation $\mathcal{A} = (A_1, \ldots, A_n)$ that maximizes $\sum_{i \in N} f(u_i(A_i))$; if there are multiple such allocations, the rule chooses one arbitrarily.
Since $f$ is strictly increasing, $f(x) = -\infty$ implies that $x = 0$, which means that $f(x) \in \R$ for all $x \in \Rplus$.

A function $f$ is \emph{strictly concave} if for any distinct $x, y \in \Rpluszero$ and $\alpha \in (0, 1)$, we have $f(\alpha x + (1-\alpha)y) > \alpha f(x) + (1-\alpha)f(y)$.
We use $\log$ to denote the natural logarithm, and define $\log 0$ as $-\infty$.

\begin{definition}
\label{def:conditions}
Let $f: \Rpluszero \to \R \cup \{-\infty\}$ be a strictly increasing function.
For each $k \in \Zpluszero$, define $\Delta_{f, k} : \Rplus \to \R \cup \{\infty\}$ such that $\Delta_{f, k}(x) = f((k+1)x) - f(kx)$.
The function $f$ is said to satisfy
\begin{itemize}
    \item \emph{\condrealgeneral} if $\Delta_{f, k}(b) > \Delta_{f, k+1}(a)$ for all $k \in \Zpluszero$ and $a, b \in \Rplus$;
    \item \emph{{\condrealgeneral}a} if $\Delta_{f, k}$ is a constant function (on domain $\Rplus$) for every $k \in \Zplus$;
    \vspace{2mm}
    \item \emph{\condrealnormtwo} if $f(a) + f(b) < f(c) + f(d)$ for all $a, b, c, d \in \Rpluszero$ such that $\min\{a, b\} \leq \min\{c, d\}$ and $ab < cd$;
    \vspace{2mm}
    \item \emph{\condintidengood} if $\Delta_{f, k}(b) > \Delta_{f, k+1}(a)$ for all $k \in \Zpluszero$ and $a, b \in \Zplus$;
    \item \emph{{\condintidengood}a} if $\Delta_{f, \ell}(b) > \Delta_{f, k}(a)$ for all $k, \ell \in \Zpluszero$ and $a, b \in \Zplus$ such that $\ell < k$;
    \item \emph{{\condintidengood}b} if $\Delta_{f, k}(1) > \Delta_{f, k+1}(a) > \Delta_{f, k+2}(1)$ for all $k \in \Zpluszero$ and $a \in \Zplus$;
    \vspace{2mm}
    \item \emph{\condintbinary} if $\Delta_{f, k}(1) > \Delta_{f, k+1}(1)$ for all $k \in \Zpluszero$;
    \vspace{2mm}
    \item \emph{\condinttwovalue} if $f((\ell+1)b+ra) - f(\ell b+ra) > \Delta_{f, k+1}(a) > \Delta_{f, k+2}(1)$ for all $a, b \in \Zplus$ and $k, \ell, r \in \Zpluszero$ such that $a \geq b$ and $(k+1)b > \ell b + ra$;
    \vspace{2mm}
    \item \emph{{\condintgeneral}a} if $f((k+1)b-1) - f(kb-1) > \Delta_{f, k}(a)$ for all $k, a, b \in \Zplus$;
    \item \emph{{\condintgeneral}b} if $f(y+b) - f(y) > f(x+a) - f(x)$ for all $a, b \in \Zplus$ and $x, y \in \Zpluszero$ such that $x/a \geq (y+1)/b$.
\end{itemize}
\end{definition}

Note that for $k \in \Zpluszero$ and $x \in \Rplus$, we have $f((k+1)x) > f(kx)$ since $f$ is strictly increasing, so $\Delta_{f, k}$ is well-defined and positive.
Moreover, we have $\Delta_{f, k}(x) = \infty$ if and only if $k = 0$ and $f(0) = -\infty$.

Before proceeding further, we provide some intuition behind these conditions. 
Firstly, {\condrealgeneral}a is a characterization of the logarithmic function $f(x) = \alpha \log x + \beta$ where $\alpha \in \Rplus$ and  $\beta \in \R$: clearly, any logarithmic function satisfies {\condrealgeneral}a, while for the converse, see \Cref{lem:real_idengood_log}.
Most other conditions are related to {\condrealgeneral}a in that they impose some restrictions on $\Delta_{f, k}$.
For example, Conditions~{\condnointidengood}b and {\condnointtwovalue} allow $\Delta_{f, k}$ to take on any values as long as they lie within specific upper and lower bounds, whereas {\condintbinary} is equivalent to being concave on the integer domain.
In other words, these conditions are satisfied by functions with properties similar to $\log$, with varying degrees of similarity depending on the conditions.
{\condrealnormtwo} is slightly different from the rest in that it references not only the logarithmic function (since $ab < cd$ is equivalent to $\log a + \log b < \log c + \log d$) but also the $\min$ function.

\subsection{Examples of Additive Welfarist Rules}
\label{subsec:examples}

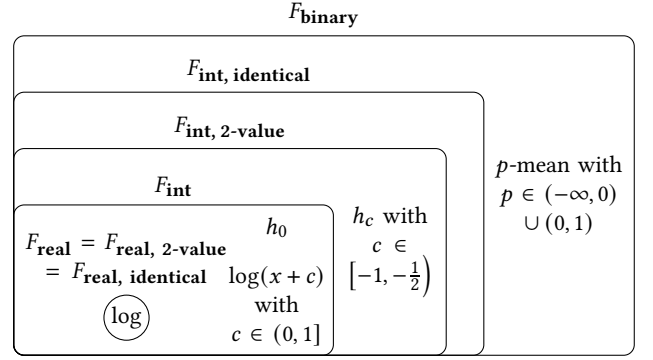
\begin{figure}[t]
\fontsize{9}{12}\selectfont
\centering
\begin{tikzpicture}
\draw (1.5, 0.5) circle (0.3) node {$\log$};
\draw (1.5, 0.8) node [text=black,above,text width=10em,text centered,font=\small] {$F_{\textbf{real}} = F_{\textbf{real, 2-value}}$\\$=F_{\textbf{real, identical}}$};

\draw[rounded corners] (0, 0) rectangle (5.25, 2);
\draw (2.625, 2) node [text=black,above,font=\small] {$F_{\textbf{int}}$};
\draw (4, 0.65) node [text=black,text width=10em,text centered] {$\log(x+c)$ with \\ $c \in (0, 1]$};
\draw (4, 1.7) node [text=black,text width=5em,text centered] {$h_0$};

\draw[rounded corners] (0, 0) rectangle (7.75, 2.75);
\draw (3.875, 2.75) node [text=black,above,font=\small] {$F_{\textbf{int, 2-value}}$};
\draw (6.5, 1.375) node [text=black,text width=10em,text centered] {$h_c$ with \\ $c \in \left[-1, -\frac{1}{2}\right)$};

\draw[rounded corners] (0, 0) rectangle (8, 3.5);
\draw (4, 3.5) node [text=black,above,font=\small] {$F_{\textbf{int, identical}}$};

\draw[rounded corners] (0, 0) rectangle (11.5, 4.25);
\draw (5.75, 4.25) node [text=black,above,font=\small] {$F_{\textbf{binary}}$};
\draw (9.75, 2.125) node [text=black,text width=10em,text centered] {$p$-mean with \\ $p \in (-\infty, 0) \cup\, (0,1)$};
\end{tikzpicture}
\caption{Venn diagram showing the relationships between additive welfarist rules that guarantee EF1 for different classes of instances (see \Cref{tab:results} for the full names of these classes).
Each region, labeled by $F$ with a subscript, represents the set of functions defining the additive welfarist rules that guarantee EF1 for the class of instances corresponding to the subscript.
Some examples of these functions are given in the respective regions.}
\label{fig:venn_diagram}
\end{figure}

The \emph{maximum Nash welfare (MNW) rule} is the additive welfarist rule defined by the function $\log$.
We consider variations of this rule by defining, for each $c \in \Rpluszero$, the \emph{modified logarithmic function} $\lambda_c : \Rpluszero \to \R \cup \{-\infty\}$ such that $\lambda_c(x) = \log(x+c)$.

The \emph{maximum harmonic welfare (MHW) rule} is the additive welfarist rule defined by the function $h_0(x) = \sum_{t=1}^x 1/t$.
Note that this definition only makes sense when $x$ is a non-negative integer; hence, we extend its definition to the non-negative real domain using the function $h_0(x) = \int_{0}^{1} \frac{1-t^{x}}{1-t} \, \dd t$ \citep{Hintze19}.
We also consider variations of this rule.
We define, for each $c \geq -1$, the \emph{modified harmonic function} $h_c : \Rpluszero \to \R \cup \{-\infty\}$ such that%
\footnote{This family of functions is defined for $c \in [-1, 0]$ in the domain of non-negative integers in previous work \citep{MontanariScSu25}.}
\begin{align*}
    h_{-1}(x) = \int_{0}^{1}{\frac{1-t^{x-1}}{1-t}}\, \dd t \hspace{2mm} \text{ and } \hspace{2mm} h_c(x) = \int_{0}^{1} \frac{t^c-t^{x+c}}{1-t} \, \dd t
\end{align*}
for $c > -1$.
When the domain of $h_c$ is restricted to the set of non-negative integers, we can rewrite%
\footnote{We prove this equivalence in Appendix~\ref{ap:harmonic_extension}.}
the values as
\begin{align*}
    h_{-1}(x) =
    \begin{cases}
        \sum_{t=1}^{x-1} \frac{1}{t} & \text{if $x \geq 1$}; \\
        -\infty                    & \text{if $x = 0$}
    \end{cases}
    \hspace{2mm} \text{ and } \hspace{2mm}
    h_c(x) = \sum_{t=1}^x \frac{1}{t+c}
\end{align*}
for $c > -1$ and all $x \in \Zpluszero$.

An \emph{MNW allocation} (resp.~\emph{MHW allocation}) is an allocation chosen by the MNW rule (resp.~MHW rule).

For $p \in \R$, the \emph{(generalized) $p$-mean rule} is an additive welfarist rule defined by the function $\varphi_p$, where
\begin{align*}
    \varphi_p(x) = \begin{cases}
        x^p & \text{if $p > 0$;} \\
        \log x & \text{if $p = 0$;} \\
        -x^p & \text{if $p < 0$.}
    \end{cases}
\end{align*}
Hence, the MNW rule is also the $0$-mean rule.

\section{Real-Valued Instances}
\label{sec:real}

We start with the general case where the utility of a good can be any (non-negative) real number.
As mentioned earlier, the only additive welfarist rule that guarantees EF1 for  all such instances is the MNW rule \citep{Suksompong23}.
In this section, we consider the subclasses of \emph{identical-good} instances (\Cref{subsec:real_idengood}), \emph{two-value} instances (\Cref{subsec:real_twovalue}), and \emph{normalized} instances (\Cref{subsec:real_normalized}).
We demonstrate that the characterization of MNW continues to hold for the first two classes, as well as for the third class when there are three or more agents.
To this end, we show that the only additive welfarist rule that ensures EF1 for these classes of instances is defined by the logarithmic function, which is equivalent to Conditions {\condnorealgeneral} and {\condnorealgeneral}a.

\subsection{Identical-Good Instances}
\label{subsec:real_idengood}

Recall from \Cref{sec:prelim} that identical-good instances have the property that each agent assigns the same utility to every good, although the utilities may be different for different agents.
This is a simple and natural subclass corresponding to the scenario where the utility of each agent depends only on the number of goods received, but the scaling may differ among agents.
We state the characterization for such instances.

\begin{theorem}
\label{thm:real_idengood}
Let $n \geq 2$ be given, and let $f: \Rpluszero \to \R \cup \{-\infty\}$ be a strictly increasing function continuous on $\Rplus$.
Then, the following statements are equivalent:
\begin{enumerate}[label=(\alph*)]
    \item For every positive-admitting identical-good instance with $n$ agents, every allocation chosen by the additive welfarist rule with $f$ is EF1.
    \item There exist constants $\alpha \in \Rplus$ and $\beta \in \R$ such that $f(x) = \alpha \log x + \beta$ for all $x \in \Rpluszero$.
\end{enumerate}
\end{theorem}

Note that the MNW rule is equivalent to the additive welfarist rule with a function that satisfies \Cref{thm:real_idengood}(b); therefore, the theorem essentially says that the only additive welfarist rule that guarantees EF1 for identical-good instances is the MNW rule.

The implication (b) $\Rightarrow$ (a) follows from the result of \citet{CaragiannisKuMo19} that MNW guarantees EF1.
We shall prove the implication (a) $\Rightarrow$ (b) via a series of lemmas.
We begin by showing that \emph{{\condrealgeneral}} is a necessary condition for a function $f$ that satisfies the statement in \Cref{thm:real_idengood}(a).

\begin{lemma}
\label{lem:real_idengood_cond1}
Let $n \geq 2$ be given, and let $f: \Rpluszero \to \R \cup \{-\infty\}$ be a strictly increasing function such that the statement in \Cref{thm:real_idengood}(a) holds.
Then, $f$ satisfies {\condrealgeneral}.
\end{lemma}

\begin{proof}
Let $k \in \Zpluszero$ and $a, b \in \Rplus$ be given.
Consider an instance with $n$ agents and $(k+1)n$ goods such that the utility of each good is $a$ for agent $1$ and $b$ for the remaining agents.
This instance is an identical-good instance; moreover, it is positive-admitting since the allocation where every agent receives $k+1$ goods gives positive utility to every agent.

Let $\mathcal{A} = (A_1, \ldots, A_n)$ be the allocation such that every agent receives $k+1$ goods.
We have $\sum_{i \in N} f(u_i(A_i)) = f((k+1)a) + (n-1)f((k+1)b)$.
Note that $\mathcal{A}$ is the only EF1 allocation; indeed, if some agent does not receive exactly $k+1$ goods, then some agent $i$ receives fewer than $k+1$ goods while another agent $j$ receives more than $k+1$ goods, so agent $i$ envies agent $j$ by at least two goods, thereby violating EF1.
Let $\mathcal{B} = (B_1, \ldots, B_n)$ be the allocation such that agent $1$ receives $k+2$ goods, agent $2$ receives $k$ goods, and every other agent receives $k+1$ goods.
We have $\sum_{i \in N} f(u_i(B_i)) = f((k+2)a) + f(kb) + (n-2)f((k+1)b)$.
Since $\mathcal{A}$ is the only EF1 allocation, the additive welfarist rule with $f$ chooses (only) $\mathcal{A}$, and it holds that $\sum_{i \in N} f(u_i(A_i)) > \sum_{i \in N} f(u_i(B_i))$.
Rearranging the terms, we get 
\begin{align*}
f((k+1)b) - f(kb) > f((k+2)a) - f((k+1)a),
\end{align*}
or equivalently, $\Delta_{f,k}(b) > \Delta_{f,k+1}(a)$.
Since $k, a, b$ were arbitrarily chosen, $\Delta_{f,k}(b) > \Delta_{f,k+1}(a)$ holds for all $k \in \Zpluszero$ and $a, b \in \Rplus$.
Therefore, $f$ satisfies {\condrealgeneral}.
\end{proof}

Next, we show that {\condrealgeneral} is sufficient for the function $f$ to satisfy \Cref{thm:real_idengood}(b).
To this end, we prove in \Cref{lem:real_idengood_constant,lem:real_idengood_log} that {\condrealgeneral} implies {\condrealgeneral}a, which in turn implies that $f$ satisfies the statement in \Cref{thm:real_idengood}(b).

\begin{restatable}{lemma}{lemrealidengoodconstant}
\label{lem:real_idengood_constant}
Let $f: \Rpluszero \to \R \cup \{-\infty\}$ be a strictly increasing function that satisfies {\condrealgeneral}.
Then, $f$ satisfies {\condrealgeneral}a.
\end{restatable}

We shall prove \Cref{lem:real_idengood_constant} via a series of lemmas.
Let $\inf \Delta_{f, k} = \inf \left\{ \Delta_{f, k}(x) : x \in \Rplus \right\}$ and $\sup \Delta_{f, k} = \sup \left\{ \Delta_{f, k}(x) : x \in \Rplus \right\}$.

\begin{lemma}
\label{lem:condition1_p2k}
Let $k \in \Zplus$ and $x \in \Rplus$, and let $f: \Rpluszero \to \R \cup \{-\infty\}$ be a strictly increasing function. 
Then, we have $\Delta_{f, k}(x) = \Delta_{f, 2k}(x/2) + \Delta_{f, 2k+1}(x/2)$.
\end{lemma}

\begin{proof}
We have
\begin{align*}
    \Delta_{f, k}(x) &= f((k+1)x) - f(kx) \\
    &= f\left((2k+2)\frac{x}{2}\right) - f\left((2k)\frac{x}{2}\right) \\
    &= f\left((2k+1)\frac{x}{2}\right) - f\left((2k)\frac{x}{2}\right) + f\left((2k+2)\frac{x}{2}\right) - f\left((2k+1)\frac{x}{2}\right) \\
    &= \Delta_{f, 2k}\left(\frac{x}{2}\right) + \Delta_{f, 2k+1}\left(\frac{x}{2}\right). \qedhere
\end{align*}
\end{proof}

\begin{lemma}
\label{lem:condition1_supinfpk}
Let $k \in \Zplus$, and let $f: \Rpluszero \to \R \cup \{-\infty\}$ be a strictly increasing function that satisfies {\condrealgeneral}.
Then, $\sup \Delta_{f, k}$ and $\inf \Delta_{f, k}$ are both non-negative and finite.
Moreover, $\inf \Delta_{f, k} \geq \sup \Delta_{f, k+1}$.
\end{lemma}

\begin{proof}
First, we show non-negativity.
For $x \in \Rplus$, it holds that $(k+1)x > kx$.
Since $f$ is strictly increasing, we have $f((k+1)x) > f(kx)$, which implies that $\Delta_{f, k}(x) > 0$.
Therefore, both $\inf \Delta_{f, k}$ and $\sup \Delta_{f, k}$ are non-negative.

Next, we show finiteness.
Let $x \in \Rplus$.
For $k \geq 2$, since $f$ satisfies {\condrealgeneral}, we have $\Delta_{f, k}(x) < \Delta_{f, k-1}(1) < \cdots < \Delta_{f, 1}(1) = f(2) - f(1)$, which is finite, so $\sup \Delta_{f, k}$ is finite.
For $k = 1$, we have $\Delta_{f, 1}(x) = \Delta_{f, 2}(x/2) + \Delta_{f, 3}(x/2)$ by \Cref{lem:condition1_p2k}, and the latter is upper-bounded by $\sup \Delta_{f, 2} + \sup \Delta_{f, 3}$, which is finite, so $\sup \Delta_{f, 1}$ is also finite.
Therefore, $\sup \Delta_{f, k}$ is finite for all $k \in \Zplus$.
Since $\inf \Delta_{f, k} \leq \sup \Delta_{f, k}$, the finiteness of $\inf \Delta_{f, k}$ follows.

Finally, since $f$ satisfies {\condrealgeneral}, we get $\Delta_{f, k}(b) > \Delta_{f, k+1}(a)$ for all $a, b \in \Rplus$.
This implies that $\inf \Delta_{f, k} \geq \sup \Delta_{f, k+1}$.
\end{proof}

For $k \in \Zplus$, define $d_{f, k} = \sup \Delta_{f, k} - \inf \Delta_{f, k}$.
Note that $d_{f, k}$ is well-defined and finite by \Cref{lem:condition1_supinfpk}.

\begin{lemma}
\label{lem:condition1_dk}
Let $k \in \Zplus$, and let $f: \Rpluszero \to \R \cup \{-\infty\}$ be a strictly increasing function that satisfies {\condrealgeneral}. 
Then, $0 \leq d_{f, k} \leq d_{f, 2k} + d_{f, 2k+1}$.
\end{lemma}

\begin{proof}
Let $k \in \Zplus$, and note that
$d_{f, k} \geq 0$ since $\sup \Delta_{f, k} \geq \inf \Delta_{f, k}$.
By the definitions of supremum and infimum, there exist sequences of positive real numbers $(a_t)_{t \in \Zplus}$ and $(b_t)_{t \in \Zplus}$ such that
\begin{align*}
    \lim_{t \to \infty} \Delta_{f, k}(a_t) = \sup \Delta_{f, k} \hspace{2mm} \text{ and } \hspace{2mm} \lim_{t \to \infty} \Delta_{f, k}(b_t) = \inf \Delta_{f, k}.
\end{align*}
Then,
\begin{align*}
    d_{f, k} &= \sup \Delta_{f, k} - \inf \Delta_{f, k} \\
    &= \lim_{t \to \infty} \Delta_{f, k}(a_t) - \lim_{t \to \infty} \Delta_{f, k}(b_t) \\
    &= \lim_{t \to \infty} \left( \Delta_{f, k}(a_t) - \Delta_{f, k}(b_t) \right) \\
    &= \lim_{t \to \infty} \left( \Delta_{f, 2k}\left(\frac{a_t}{2}\right) + \Delta_{f, 2k+1}\left(\frac{a_t}{2}\right) - \Delta_{f, 2k}\left(\frac{b_t}{2}\right) - \Delta_{f, 2k+1}\left(\frac{b_t}{2}\right) \right) \tag{by \Cref{lem:condition1_p2k}} \\
    &\leq \lim_{t \to \infty} \left( \sup \Delta_{f, 2k} + \sup \Delta_{f, 2k+1} - \inf \Delta_{f, 2k} - \inf \Delta_{f, 2k+1} \right) \tag{since each term is finite by \Cref{lem:condition1_supinfpk}} \\
    &= \sup \Delta_{f, 2k} - \inf \Delta_{f, 2k} + \sup \Delta_{f, 2k+1} - \inf \Delta_{f, 2k+1} \\
    &= d_{f, 2k} + d_{f, 2k+1}. \tag*{\qedhere}
\end{align*}
\end{proof}

\begin{lemma}
\label{lem:condition1_sum}
Let $r \in \Zpluszero$ and $\ell \in \Zplus$, and let $f: \Rpluszero \to \R \cup \{-\infty\}$ be a strictly increasing function that satisfies {\condrealgeneral}. 
Then,
\begin{align*}
    \ell \cdot \sum_{t=2^r}^{2^{r+1}-1} d_{f, t} \leq \sum_{t=2^r}^{2^{r+\ell}-1} d_{f, t}.
\end{align*}
\end{lemma}

\begin{proof}
We prove this statement by induction on $\ell$.
The base case $\ell = 1$ holds since the expressions on both sides coincide.
For the inductive step, assume that
\begin{align*}
    \ell \cdot \sum_{t=2^r}^{2^{r+1}-1} d_{f, t} \leq \sum_{t=2^r}^{2^{r+\ell}-1} d_{f, t}
\end{align*}
holds for some $\ell \in \Zplus$; we shall prove the statement for $\ell+1$.
We have
\begin{align*}
    (\ell+1) \cdot \sum_{t=2^r}^{2^{r+1}-1} d_{f, t} &\leq \sum_{t=2^r}^{2^{r+1}-1} d_{f, t} + \sum_{t=2^r}^{2^{r+\ell}-1} d_{f, t} \tag{by the inductive hypothesis} \\
    &\leq \sum_{t=2^r}^{2^{r+1}-1} d_{f, t} + \sum_{t=2^r}^{2^{r+\ell}-1} \left( d_{f, 2t} + d_{f, 2t+1} \right) \tag{by \Cref{lem:condition1_dk}} \\
    &= \sum_{t=2^r}^{2^{r+1}-1} d_{f, t} + \sum_{t=2^{r+1}}^{2^{r+\ell+1}-1} d_{f, t} \\
    &= \sum_{t=2^r}^{2^{r+\ell+1}-1} d_{f, t},
\end{align*}
completing the induction.
\end{proof}

We are now ready to prove \Cref{lem:real_idengood_constant}.

\begin{proof}[Proof of \Cref{lem:real_idengood_constant}]
Let $k \in \Zplus$ be given.
It suffices to show that $d_{f, k} = 0$.
Note that $d_{f, k}$ is non-negative by \Cref{lem:condition1_dk}.

Let $\ell \in \Zplus$.
Define $r = \lfloor \log_2 k \rfloor$ and $s = k - 2^r$, so that $k = 2^r + s$ for $r, s \in \Zpluszero$ with $s < 2^r$.
Then, we have
\begin{align*}
    \ell d_{f, k} &\leq \ell (d_{f, 2^r} + \cdots + d_{f, k} + \cdots + d_{f, 2^r+(2^r-1)}) \tag{since \protect{$d_{f, t} \geq 0$} by \Cref{lem:condition1_dk}} \\
    &= \ell \cdot \sum_{t=2^r}^{2^{r+1}-1} d_{f, t} \\
    &\leq \sum_{t=2^r}^{2^{r+\ell}-1} d_{f, t} \tag{by \Cref{lem:condition1_sum}} \\
    &\leq \sum_{t=1}^{2^{r+\ell}-1} d_{f, t} \tag{since \protect{$d_{f, t} \geq 0$} by \Cref{lem:condition1_dk}} \\
    &= \sum_{t=1}^{2^{r+\ell}-1} \left( \sup \Delta_{f, t} - \inf \Delta_{f, t} \right) \tag{by definition} \\
    &\leq \sum_{t=1}^{2^{r+\ell}-1} \left( \sup \Delta_{f, t} - \sup \Delta_{f, t+1} \right) \tag{by \Cref{lem:condition1_supinfpk}} \\
    &= \sup \Delta_{f, 1} - \sup \Delta_{f, 2^{r+\ell}} \tag{by telescoping sum} \\
    &\leq \sup \Delta_{f, 1}. \tag{since \protect{$\sup \Delta_{f, 2^{r+\ell}} \geq 0$} by \Cref{lem:condition1_supinfpk}}
\end{align*}
This gives $0 \leq d_{f, k} \leq (\sup \Delta_{f, 1})/\ell$.
Since $\ell \in \Zplus$ is arbitrary and $\sup \Delta_{f, 1}$ is a finite non-negative constant by \Cref{lem:condition1_supinfpk}, we conclude that $d_{f, k} = 0$.
\end{proof}

Next, we state another key lemma.
Combined with \Cref{lem:real_idengood_constant}, this lemma shows that if $f$ satisfies {\condrealgeneral}, then it satisfies the statement in \Cref{thm:real_idengood}(b).

\begin{restatable}{lemma}{lemrealidengoodlog}
\label{lem:real_idengood_log}
Let $f: \Rpluszero \to \R \cup \{-\infty\}$ be a strictly increasing function continuous on $\Rplus$ that satisfies {\condrealgeneral}a.
Then, there exist constants $\alpha \in \Rplus$ and $\beta \in \R$ such that $f(x) = \alpha \log x + \beta$ for all $x \in \Rpluszero$.
\end{restatable}

We remark that \citet[Lemma 2]{Suksompong23} proved a result similar to \Cref{lem:real_idengood_log} but made the stronger assumption that the function $f$ is \emph{differentiable}.
Our result only assumes \emph{continuity}.
Again, we establish \Cref{lem:real_idengood_log} via a series of intermediate results.
For each $t \in \Rplus$, define $q_{f, t} : \Rplus \to \R$ such that $q_{f, t}(x) = f(tx) - f(x)$.

\begin{lemma}
\label{lem:const_rational}
Let $t > 1$ be a rational number, and let $f: \Rpluszero \to \R \cup \{-\infty\}$ be a strictly increasing function that satisfies {\condrealgeneral}a.
Then, $q_{f, t}$ is a constant function.
\end{lemma}

\begin{proof}
Let $x, y \in \Rplus$, and let $t = a/b$ with $a, b \in \Zplus$.
Since $t > 1$, we have $a > b$.
Then,
\begin{align*}
    q_{f, t}(x) &= f(tx) - f(x) \\
    &= \sum_{k=b}^{a-1} f\left( \frac{k+1}{b} x \right) - f\left( \frac{k}{b} x \right) \tag{by telescoping sum} \\
    &= \sum_{k=b}^{a-1} \Delta_{f, k}(x/b) \\
    &= \sum_{k=b}^{a-1} \Delta_{f, k}(y/b) \tag{since \protect{$\Delta_{f, k}$} is constant by {\condrealgeneral}a} \\
    &= \sum_{k=b}^{a-1} f\left( \frac{k+1}{b} y \right) - f\left( \frac{k}{b} y \right)  \\
    &= f(ty) - f(y) \tag{by telescoping sum} \\
    &= q_{f, t}(y).
\end{align*}
It follows that $q_{f, t}$ is a constant function.
\end{proof}

\begin{lemma}
\label{lem:const_irrational}
Let $t > 1$ be an irrational number, and let $f: \Rpluszero \to \R \cup \{-\infty\}$ be a strictly increasing function continuous on $\Rplus$ that satisfies {\condrealgeneral}a.
Then, $q_{f, t}$ is a constant function.
\end{lemma}

\begin{proof}
Let $x, y \in \Rplus$, and let $\epsilon > 0$.
Since $f$ is continuous at $tx$ and $ty$, there exists $\delta > 0$ such that for all $x' \in (tx-\delta, tx+\delta)$ we have $|f(tx) - f(x')| < \epsilon/2$, and for all $y' \in (ty-\delta, ty+\delta)$ we have $|f(ty) - f(y')| < \epsilon/2$.
Moreover, there exists a rational number $t'$ such that $t < t' < t + \min \{\delta/x, \delta/y\}$.
Then, we have $t'x \in (tx-\delta, tx+\delta)$ and $t'y \in (ty-\delta, ty+\delta)$, and hence $|f(tx) - f(t'x)| < \epsilon/2$ and $|f(ty) - f(t'y)| < \epsilon/2$.
Thus,
\begin{align*}
    \left|q_{f, t}(x) - q_{f, t}(y)\right| &= |(f(tx) - f(x)) - (f(ty) - f(y))| \\
    &= |(f(tx) - f(t'x) + f(t'x) - f(x)) - (f(ty) - f(t'y) + f(t'y) - f(y))| \\
    &\leq |(f(t'x) - f(x)) - (f(t'y) - f(y))| + |f(tx) - f(t'x)| + |f(ty) - f(t'y)| \tag{by triangle inequality} \\
    &< \left|q_{f, t'}(x) - q_{f, t'}(y)\right| + \epsilon/2 + \epsilon/2 \\
    &= 0 + \epsilon \tag{by \Cref{lem:const_rational}} \\
    &= \epsilon.
\end{align*}
Since $\epsilon > 0$ is arbitrary, we have $q_{f, t}(x) = q_{f, t}(y)$.
It follows that $q_{f, t}$ is a constant function.
\end{proof}

\begin{lemma}
\label{lem:const_q}
Let $t$ be a positive real number, and let $f: \Rpluszero \to \R \cup \{-\infty\}$ be a strictly increasing function continuous on $\Rplus$ that satisfies {\condrealgeneral}a.
Then, $q_{f, t}$ is a constant function.
\end{lemma}

\begin{proof}
For $t > 1$, this follows directly from \Cref{lem:const_rational,lem:const_irrational}.
For $t = 1$, this is trivial.
Finally, for $0 < t < 1$, this follows from the fact that $q_{f, 1/t}$ is constant.
\end{proof}

Define $q_f : \Rplus \to \R$ such that $q_f(t) = q_{f, t}(1)$ for all $t \in \Rplus$.
Since $q_{f, t}$ is a constant function for each $t \in \Rplus$ by \Cref{lem:const_q}, we have $q_f(t) = f(t) - f(1) = f(tx) - f(x)$ for all $x \in \Rplus$.

\begin{lemma}
\label{lem:qf_b}
Let $b$ be a positive integer, and let $f: \Rpluszero \to \R \cup \{-\infty\}$ be a strictly increasing function continuous on $\Rplus$ that satisfies {\condrealgeneral}a.
Then, $q_f\left(2^{1/b}\right) = \frac{1}{b} q_f(2)$.
\end{lemma}

\begin{proof}
We have
\begin{align*}
    q_f(2) &= f(2) - f(1) \\
    &= \sum_{k=0}^{b-1} f\left( 2^{(k+1)/b} \right) - f\left( 2^{k/b} \right) \tag{by telescoping sum} \\
    &= \sum_{k=0}^{b-1} f\left( 2^{1/b} \cdot 2^{k/b} \right) - f\left( 1 \cdot 2^{k/b} \right) \\
    &= \sum_{k=0}^{b-1} q_f \left( 2^{1/b} \right) \tag{by \Cref{lem:const_q}} \\
    &= b \cdot q_f \left( 2^{1/b} \right),
\end{align*}
which implies the desired result.
\end{proof}

\begin{lemma}
\label{lem:qf_ab}
Let $t$ be a rational number, and let $f: \Rpluszero \to \R \cup \{-\infty\}$ be a strictly increasing function continuous on $\Rplus$ that satisfies {\condrealgeneral}a.
Then, $q_f\left(2^t\right) = t q_f(2)$.
\end{lemma}

\begin{proof}
The result is trivial for $t = 0$, since $q_f\left(2^0\right) = q_f(1) = f(1) - f(1) = 0$.

We first prove the result for positive $t$.
Let $a, b \in \Zplus$ such that $t = a/b$.
Then, we have
\begin{align*}
    q_f \left( 2^{a/b} \right) &= f \left( 2^{a/b} \right) - f(1) \\
    &= \sum_{k=0}^{a-1} f\left( 2^{(k+1)/b} \right) - f\left( 2^{k/b} \right) \tag{by telescoping sum} \\
    &= \sum_{k=0}^{a-1} f\left( 2^{1/b} \cdot 2^{k/b} \right) - f\left( 1 \cdot 2^{k/b} \right) \\
    &= \sum_{k=0}^{a-1} q_f \left( 2^{1/b} \right) \tag{by \Cref{lem:const_q}} \\
    &= a \cdot q_f \left( 2^{1/b} \right) \\
    &= \frac{a}{b} q_f(2) \tag{by \Cref{lem:qf_b}} \\
    &= t q_f(2),
\end{align*}
as desired.

Finally, we prove the result for negative $t$.
Let $a$ be a negative integer and $b \in \Zplus$ such that $t = a/b$.
Then, we have
\begin{align*}
    q_f \left( 2^{a/b} \right) &= f \left( 2^{a/b} \cdot 2^{-a/b} \right) - f \left( 1 \cdot 2^{-a/b} \right) \\
    &= -\left( f \left(2^{-a/b} \right) - f(1) \right) \\
    &= -q_f \left(2^{-a/b} \right) \\
    &= -\left( -\frac{a}{b} q_f(2) \right) \tag{by the first part of this proof} \\
    &= t q_f(2),
\end{align*}
as desired.
Hence, $q_f\left(2^t\right) = t q_f(2)$ for all rational $t$.
\end{proof}

We are now ready to prove \Cref{lem:real_idengood_log}.

\begin{proof}[Proof of \Cref{lem:real_idengood_log}]
Let $\alpha = (f(2)-f(1))/\log 2 \in \Rplus$ and $\beta = f(1) \in \R$, and define $f_0 : \Rpluszero \to \R \cup \{-\infty\}$ such that $f_0(x) = \alpha \log x + \beta$ for all $x \in \Rpluszero$.

We show that $f(x) = f_0(x)$ for all $x = 2^t$ where $t$ is rational.
Let $t$ be a rational number.
By \Cref{lem:qf_ab}, $q_f \left(2^t\right) = t q_f(2)$, which implies that $f\left(2^t\right)-f(1) = t(f(2)-f(1))$.
Since $f(1) = \beta$ and $f(2)-f(1) = \alpha \log 2$, we have $f\left(2^t\right) = t \alpha \log 2 + \beta = \alpha \log \left(2^t\right) + \beta = f_0 \left(2^t\right)$, so $f(x) = f_0(x)$ for all $x = 2^t$ where $t$ is rational.

Now, since $f$ and $f_0$ are continuous functions that coincide on the set $\{2^t : t \text{ is rational} \}$, which is dense in $\Rplus$, we must have $f(x) = f_0(x)$ for all $x \in \Rplus$.
For the case $x=0$, since $f$ is strictly increasing, we must have $f(0) < f(x) = \alpha \log x + \beta$ for all $x \in \Rplus$.
The only possibility for $f(0)$ is $f(0) = -\infty = f_0(0)$.
This completes the proof that $f(x) = f_0(x) = \alpha \log x + \beta$ for all $x \in \Rpluszero$.
\end{proof}

Finally, we are ready to prove \Cref{thm:real_idengood}.

\begin{proof}[Proof of \Cref{thm:real_idengood}]
The implication (a) $\Rightarrow$ (b) follows from Lemmas~\ref{lem:real_idengood_cond1}, \ref{lem:real_idengood_constant}, and \ref{lem:real_idengood_log}, while the implication (b) $\Rightarrow$ (a) follows from the result of \citet{CaragiannisKuMo19} that MNW guarantees EF1.
\end{proof}

\Cref{thm:real_idengood} and Lemmas~\ref{lem:real_idengood_cond1}, \ref{lem:real_idengood_constant}, and \ref{lem:real_idengood_log} imply that a continuous function having the logarithmic property is equivalent to the function satisfying Conditions {\condnorealgeneral} and {\condnorealgeneral}a (separately).
We view this as an interesting mathematical property in itself.

\subsection{Two-Value Instances}
\label{subsec:real_twovalue}

We now consider two-value instances, which are instances with only two possible utilities for the goods.
Note that the classes of identical-good instances and two-value instances do not contain each other.
Indeed, an identical-good instance with $n \geq 3$ agents may be \emph{$n$-value} rather than \emph{two-value} since each of the $n$ agents could assign a unique utility to the goods,%
\footnote{However, for two agents, an identical-good instance is necessarily two-value.}
whereas a two-value instance allows an agent to have varying utilities for different goods.

Despite this incomparability, we show that the same characterization for identical-good instances also holds for two-value instances.

\begin{theorem}
\label{thm:real_twovalue}
Let $n \geq 2$ be given, and let $f: \Rpluszero \to \R \cup \{-\infty\}$ be a strictly increasing function continuous on $\Rplus$.
Then, the following statements are equivalent:
\begin{enumerate}[label=(\alph*)]
    \item For every positive-admitting two-value instance with $n$ agents, every allocation chosen by the additive welfarist rule with $f$ is EF1.
    \item There exist constants $\alpha \in \Rplus$ and $\beta \in \R$ such that $f(x) = \alpha \log x + \beta$ for all $x \in \Rpluszero$.
\end{enumerate}
\end{theorem}

\begin{proof}
(a) $\Rightarrow$ (b): We claim that $f$ satisfies {\condrealgeneral}.
The proof of this claim is identical to that of \Cref{lem:real_idengood_cond1}, since the construction in the proof of \Cref{lem:real_idengood_cond1} uses an instance that is also two-value.
Then, by \Cref{lem:real_idengood_constant,lem:real_idengood_log}, statement (b) holds.

(b) $\Rightarrow$ (a) follows from the result of \citet{CaragiannisKuMo19} that MNW guarantees EF1.
\end{proof}

\subsection{Normalized Instances}
\label{subsec:real_normalized}

Next, we consider normalized instances, where agents assign the same utility to the entire set of goods $M$.
We first consider the case of three or more agents.
We show that just like for the cases of identical-good instances and two-value instances, the only additive welfarist rule that yields EF1 for this case is the MNW rule.
This extends Theorem~6 of \citet{EckartPsVe24}, which only handles $p$-mean rules (as opposed to all additive welfarist rules).

\begin{theorem}
\label{thm:real_normalized_three}
Let $n \geq 3$ be given, and let $f: \Rpluszero \to \R \cup \{-\infty\}$ be a strictly increasing function continuous on $\Rplus$.
Then, the following statements are equivalent:
\begin{enumerate}[label=(\alph*)]
    \item For every positive-admitting normalized instance with $n$ agents, every allocation chosen by the additive welfarist rule with $f$ is EF1.
    \item There exist constants $\alpha \in \Rplus$ and $\beta \in \R$ such that $f(x) = \alpha \log x + \beta$ for all $x \in \Rpluszero$.
\end{enumerate}
\end{theorem}

To prove \Cref{thm:real_normalized_three}, we cannot use the construction in the proof of \Cref{lem:real_idengood_cond1} since the instance in that construction is not necessarily normalized.
Instead, we show that the function defining the additive welfarist rule must satisfy \emph{{\condrealgeneral}a}.
To this end, we augment the construction of \citet{Suksompong23} by adding a highly valuable good and an extra agent who only values that good, so that the instance becomes normalized.

\begin{lemma}
\label{lem:real_normalized_implies_constant}
Let $n \geq 3$ be given, and let $f: \Rpluszero \to \R \cup \{-\infty\}$ be a strictly increasing function continuous on $\Rplus$ such that the statement in \Cref{thm:real_normalized_three}(a) holds.
Then, $f$ satisfies {\condrealgeneral}a.
\end{lemma}

\begin{proof}
Assume, for the sake of contradiction, that $f$ does not satisfy {\condrealgeneral}a.
Hence, there exist $k \in \Zplus$ and $a, b \in \Rplus$ such that $\Delta_{f,k}(a) > \Delta_{f,k}(b)$, that is, $ f((k+1)a) - f(ka) > f((k+1)b) - f(kb)$.
By the continuity of $f$, there exists $\epsilon \in (0, b)$ such that
\begin{align}
\label{eq:real_normalized_pk}
    f((k+1)a) - f(ka) > f((k+1)b-\epsilon) - f(kb-\epsilon).
\end{align}
Consider an instance with $n$ agents and $m = k(n-1)+2$ goods.
For ease of notation, let $c = \max\{kna, knb\}$, $N' = N \setminus \{1, n\}$, $g' = g_{m-1}$, $g'' = g_m$, and $M' = M \setminus \{g', g''\}$.
The utilities of the goods are as follows.
\begin{itemize}
    \item $u_1(g') = b - \epsilon$, $u_1(g'') = c - k(n-1)b - (b-\epsilon)$, and $u_1(g) = b$ for $g \in M'$.
    \item For $i \in N'$, $u_i(g') = 0$, $u_i(g'') = c - k(n-1)a$, and $u_i(g) = a$ for $g \in M'$.
    \item $u_n(g'') = c$, and $u_n(g) = u_n(g') = 0$ for $g \in M'$.
\end{itemize}
Note that $|N'| = n-2 \geq 1$, $|M'| = k(n-1) \geq |N'|$, the utility of each good is non-negative, and $u_i(M) = c$ for all $i \in N$.
Therefore, this is a valid normalized instance; moreover, it is positive-admitting since the allocation where agent $1$ receives $g'$, agent $n$ receives $g''$, and agent $i \in N'$ each receives at least one good from $M'$ gives positive utility to every agent.

Let $\mathcal{A} = (A_1, \ldots, A_n)$ be an allocation chosen by the additive welfarist rule with $f$.
We first show that $g' \in A_1$ and $A_n \subseteq \{g''\}$.
If $g' \in A_j$ for some $j \in N \setminus \{1\}$, then giving $g'$ to agent $1$ increases $u_1(A_1)$ and does not change $u_j(A_j)$; the value of $\sum_{i \in N} f(u_i(A_i))$ increases since $f$ is strictly increasing, which contradicts the assumption that $\mathcal{A}$ is chosen by the additive welfarist rule with $f$.
Likewise, if $A_n \setminus \{g''\} \neq \emptyset$, then giving $A_n \setminus \{g''\}$ to agent $1$ results in the same contradiction.

Next, we show that $A_n = \{g''\}$.
Suppose on the contrary that $g'' \in A_j$ for some agent $j \in N \setminus \{n\}$.
Let $\mathcal{A}' = (A'_1, \ldots, A'_n)$ be the allocation such that $g''$ is given instead to agent $n$, i.e., $A'_j = A_j \setminus \{g''\}$, $A'_n = \{g''\}$, and $A'_i = A_i$ for all $i \in N \setminus \{j, n\}$.
Then,
\begin{align*}
    \sum_{i \in N} f(u_i(A'_i)) &= f(u_j(A'_j)) + f(u_n(A'_n)) + \sum_{i \in N \setminus \{j, n\}} f(u_i(A'_i)) \\
    &\geq f(0) + f(c) + \sum_{i \in N \setminus \{j, n\}} f(u_i(A'_i)) \\
    &= f(c) + f(0) + \sum_{i \in N \setminus \{j, n\}} f(u_i(A_i)) \\
    &\geq f(u_j(A_j)) + f(u_n(A_n)) + \sum_{i \in N \setminus \{j, n\}} f(u_i(A_i))
    = \sum_{i \in N} f(u_i(A_i)),
\end{align*}
where the inequalities hold because $u_n(A'_n) = c$ and $u_n(A_n) = 0$.
Now, at least one of the two inequalities must be strict because $f$ is strictly increasing and we have $u_j(A'_j) > 0$ or $u_j(A_j) < c$.
This shows that $\sum_{i \in N} f(u_i(A'_i)) > \sum_{i \in N} f(u_i(A_i))$.
It follows that $\mathcal{A}$ is not chosen by the welfarist rule, a contradiction.
Hence, $A_n = \{g''\}$.

Recall that $\mathcal{A}$ is EF1 since it is chosen by the additive welfarist rule with $f$.
We show that every agent $i \in N \setminus \{n\}$ receives exactly $k$ goods from~$M'$.
If agent $1$ receives at most $k-1$ goods from~$M'$, then some agent $j \in N'$ receives at least $k+1$ goods from~$M'$.
Then, $u_1(A_1) \leq kb-\epsilon$ and $u_1(A_j \setminus \{g\}) \geq kb > u_1(A_1)$ for all $g \in A_j$, which shows that $\mathcal{A}$ is not EF1 for agent $1$.
Therefore, agent~$1$ receives at least $k$ goods from $M'$.
Likewise, if some agent $i \in N'$ receives at most $k-1$ goods from $M'$, then some agent $j \in N \setminus \{i, n\}$ receives at least $k+1$ goods from $M'$.
Then, $u_i(A_i) \leq (k-1)a$ and $u_i(A_j \setminus \{g\}) \geq ka > u_i(A_i)$ for all $g \in A_j$, which shows that $\mathcal{A}$ is not EF1 for agent $i$.
Therefore, every agent $i \in N'$ receives at least $k$ goods from $M'$ as well.
The only way for every agent in $N \setminus \{n\}$ to receive at least $k$ goods from $M'$ is when every agent in $N \setminus \{n\}$ receives \emph{exactly} $k$ goods from $M'$.
Then, we have
\begin{align*}
    \sum_{i \in N} f(u_i(A_i)) &= f(u_1(A_1)) + \sum_{i \in N'} f(u_i(A_i)) + f(u_n(A_n)) \\
    &= f((k+1)b-\epsilon) + (n-2)f(ka) + f(c).
\end{align*}

Let $\mathcal{B} = (B_1, \ldots, B_n)$ be the allocation such that agent $1$ receives $g'$ and $k-1$ goods from $M'$, agent $2$ receives $k+1$ goods from $M'$, each agent $i \in N' \setminus \{2\}$ receives $k$ goods from $M'$, and agent~$n$ receives $g''$.
We have
\begin{align*}
    \sum_{i \in N} f(u_i(B_i)) &= f(u_1(B_1)) + \sum_{i \in N'} f(u_i(B_i)) + f(u_n(B_n)) \\
    &= f(kb-\epsilon) + f((k+1)a) + (n-3)f(ka) + f(c).
\end{align*}
Since $\mathcal{A}$ is chosen, we have $\sum_{i \in N} f(u_i(B_i)) \leq \sum_{i \in N} f(u_i(A_i))$.
Rearranging the terms, we get
\begin{align*}
    f((k+1)a) - f(ka) \leq f((k+1)b-\epsilon) - f(kb-\epsilon),
\end{align*}
contradicting \eqref{eq:real_normalized_pk}.
Therefore, $f$ satisfies {\condrealgeneral}a.
\end{proof}

With this lemma in hand, we are ready to prove \Cref{thm:real_normalized_three}.

\begin{proof}[Proof of \Cref{thm:real_normalized_three}]
The implication (a) $\Rightarrow$ (b) follows from Lemmas \ref{lem:real_normalized_implies_constant} and \ref{lem:real_idengood_log}, while the implication (b) $\Rightarrow$ (a) follows from the result of \citet{CaragiannisKuMo19} that MNW guarantees EF1.
\end{proof}

We now address the case of two agents.
\citet{EckartPsVe24} showed that in this case, the $p$-mean rule guarantees EF1 for all $p \leq 0$, which implies that there is a larger class of additive welfarist rules that guarantee EF1.
Accordingly, our characterization in \Cref{thm:real_normalized_three} does not work for two agents.
The problem of finding a characterization for this case turns out to be rather challenging.
We instead provide a necessary condition and a sufficient condition for the functions defining the additive welfarist rules that guarantee EF1 for normalized instances with two agents.

\begin{restatable}{proposition}{proprealnormalizedtwonecessary}
\label{prop:real_normalized_two_necessary}
Let $f: \Rpluszero \to \R \cup \{-\infty\}$ be a strictly increasing function continuous on $\Rplus$ such that for every positive-admitting normalized instance with two agents, every allocation chosen by the additive welfarist rule with $f$ is EF1.
Then, $f$ is strictly concave.
\end{restatable}

\begin{proof}
Assume, for the sake of contradiction, that $f$ is not strictly concave.
We claim that there exist $c \in \Rplus$ and $\delta \in (0, c)$ such that for all $\delta' \in (0, \delta]$, it holds that $f(c) \leq 0.5 f(c-\delta') + 0.5 f(c+\delta')$.
We first prove this claim.
By definition, there exist $a, b \in \Rpluszero$ with $a < b$ and $\tau \in (0, 1)$ such that $f(\tau a + (1-\tau)b) \leq \tau f(a) + (1-\tau)f(b)$.
Let $F(x) = f(x) - \frac{f(b)-f(a)}{b-a}(x-a)$.
Note that $F$ is continuous on $[a, b]$ and $F(a) = F(b) = f(a)$.
By the extreme value theorem, $F$ attains a minimum in the interval $[a, b]$ at some $c \in [a, b]$, i.e., $F(x) \geq F(c)$ for all $x \in [a, b]$.
If $c = a$ or $c = b$, then $c' = \tau a + (1-\tau)b \in (a, b)$ is also a minimum, since
\begin{align*}
    F(c')
    &= f(c') - \frac{f(b)-f(a)}{b-a}(c' - a) \\
    &\leq \tau f(a) + (1-\tau)f(b) - \frac{f(b)-f(a)}{b-a}(c' - a) \\
    &= \tau f(a) + (1-\tau)f(b) - \frac{f(b)-f(a)}{b-a}(1-\tau)(b-a) \\
    &= \tau f(a) + (1-\tau)f(b) - (1-\tau)(f(b)-f(a)) \\
    &= f(a) \  (= F(a) = F(b)).
\end{align*}
Therefore, we may assume without loss of generality that $F$ attains a minimum at some $c \in (a, b)$.
Let $\delta = \min \{c-a, b-c\} > 0$.
Then, for all $x \in [c-\delta, c+\delta]$, we have $F(x) \geq F(c)$.
This means that for all $\delta' \in (0, \delta]$, it holds that $F(c) \leq 0.5F(c-\delta') + 0.5F(c + \delta')$.
Thus, for all $\delta' \in (0, \delta]$, we have
\begin{align*}
    f(c) &= F(c) + \frac{f(b)-f(a)}{b-a}(c-a) \\
    &\leq 0.5F(c-\delta') + 0.5F(c + \delta') + \frac{f(b)-f(a)}{b-a}(c-a) \\
    &= 0.5F(c-\delta') + 0.5F(c + \delta') + 0.5 \cdot \frac{f(b)-f(a)}{b-a}((c - \delta' - a) + (c + \delta' -a)) \\
    &= 0.5\left( F(c-\delta') + \frac{f(b)-f(a)}{b-a}(c - \delta' - a) \right) + 0.5\left( F(c+\delta') + \frac{f(b)-f(a)}{b-a}(c + \delta' - a) \right) \\
    &= 0.5f(c-\delta') + 0.5f(c+\delta'),
\end{align*}
proving the claim.

Now, let $c \in \Rplus$ and $\delta \in (0, c)$ be given such that for all $\delta' \in (0, \delta]$, we have $f(c) \leq 0.5 f(c-\delta') + 0.5 f(c+\delta')$.
Consider an instance with two agents and $2 \lceil c/\delta \rceil$ goods such that each good has utility $\delta' = c/\lceil c/\delta \rceil$ to each agent.
This instance is a normalized instance; moreover, it is positive-admitting since the allocation where every agent receives $\lceil c/\delta \rceil$ goods gives positive utility to every agent.

Let $\mathcal{A} = (A_1, A_2)$ be the allocation where each agent receives $\lceil c/\delta \rceil$ goods.
Then, we have $\sum_{i \in N} f(u_i(A_i)) = 2f(c)$.
Note that $\mathcal{A}$ is the only EF1 allocation; indeed, if some agent receives fewer than $\lceil c/\delta \rceil$ goods, then this agent is not EF1 towards the other agent who receives more than $\lceil c/\delta \rceil$ goods.
Let $\mathcal{B} = (B_1, B_2)$ be the allocation such that agent $1$ receives $\lceil c/\delta \rceil - 1$ goods and agent~$2$ receives $\lceil c/\delta \rceil + 1$ goods.
We have $\sum_{i \in N} f(u_i(B_i)) = f(c-\delta') + f(c+\delta')$.
Since $\mathcal{A}$ is the only EF1 allocation, the additive welfarist rule with $f$ chooses (only) $\mathcal{A}$, and we have $\sum_{i \in N} f(u_i(A_i)) > \sum_{i \in N} f(u_i(B_i))$.
Simplifying the inequality, we get $f(c) > 0.5 f(c-\delta') + 0.5 f(c+\delta')$.
However, since $\delta' = c/\lceil c/\delta \rceil \leq c/(c/\delta) = \delta$, this contradicts the assumption on $f$.
Therefore, $f$ is strictly concave.
\end{proof}

We now turn to the sufficient condition.
Note that this is a generalization of the condition used by \citet[Lemma~3]{EckartPsVe24}.

\begin{restatable}{proposition}{proprealnormalizedtwosufficient}
\label{prop:real_normalized_two_sufficient}
Let $f: \Rpluszero \to \R \cup \{-\infty\}$ be a strictly increasing function that satisfies {\condrealnormtwo}.
Then, for every positive-admitting normalized instance with two agents, every allocation chosen by the additive welfarist rule with $f$ is EF1.
\end{restatable}

\begin{proof}
Let a positive-admitting normalized instance with two agents be given, and let $U = u_1(M) = u_2(M)$.
Let $\mathcal{A} = (A_1, A_2)$ be an allocation chosen by the additive welfarist rule with~$f$.
Assume, for the sake of contradiction, that $\mathcal{A}$ is not EF1.
Without loss of generality, we may assume that agent $1$ is not EF1 towards agent $2$.
Then, we have $u_1(A_1) < u_1(A_2 \setminus \{g\}) \leq u_1(A_2)$ for all $g \in A_2$, and thus $u_1(A_1) < U/2$.
Moreover, if $u_2(A_2) < U/2$, then exchanging the two agents' bundles increases both agents' utilities of their own bundles, thereby increasing the value of $\sum_{i=1}^2 f(u_i(A_i))$ and contradicting the assumption that $\mathcal{A}$ is chosen by the additive welfarist rule with $f$.
Therefore, $u_2(A_2) \geq U/2$.
Let $a = u_1(A_1)$, $b = u_2(A_2)$; note that $a < U/2 \leq b$, and hence $\min \{a, b\} = a$.

Since $u_1(A_2) > u_1(A_1) \geq 0$, there exists a good in $A_2$ with positive utility to agent $1$.
Let $g' = \argmin_{g \in A_2,\, u_1(g) > 0} u_2(g) / u_1(g)$.
Let $B_1 = A_1 \cup \{g'\}$, $B_2 = A_2 \setminus \{g'\}$, $c_1 = u_1(B_1)$, $d_1 = u_2(B_2)$, $c_2 = u_1(B_2)$, and $d_2 = u_2(B_1)$.

We claim that $ab < c_1d_1$.
By the choice of $g'$, we have
\begin{align*}
    \frac{u_2(g')}{u_1(g')} 
    \leq \frac{\sum_{g \in A_2,\, u_1(g) > 0} u_2(g)}{\sum_{g \in A_2,\, u_1(g) > 0} u_1(g)} 
    \leq \frac{\sum_{g \in A_2} u_2(g)}{\sum_{g \in A_2} u_1(g)} 
    = \frac{u_2(A_2)}{u_1(A_2)},
\end{align*}
which gives $(b-d_1)/(c_1-a) \leq b/u_1(A_2)$.
On the other hand, since agent $1$ is not EF1 towards agent $2$ in $\mathcal{A}$, we have $u_1(A_1 \cup \{g'\}) < u_1(A_2)$, or $c_1 < u_1(A_2)$.
Multiplying these two inequalities together and noting that the expressions on both sides of the inequalities are non-negative, we get $c_1(b-d_1)/(c_1-a) < b$.
Simple algebraic manipulation gives us the desired inequality $ab < c_1d_1$.

We consider the following two cases.

\underline{Case 1:} $a \leq d_1$.
Note that $c_1 - a = u_1(g') > 0$, and so $a < c_1$.
Thus, $a \leq \min \{c_1, d_1\}$ in this case.

Now, since $\min \{a, b\} = a \leq \min \{c_1, d_1\}$ and $ab < c_1d_1$, {\condrealnormtwo} implies that $f(a) + f(b) < f(c_1) + f(d_1)$.
On the other hand, since $\mathcal{A}$ is chosen by the additive welfarist rule with $f$, we have $\sum_{i=1}^2 f(u_i(A_i)) \geq \sum_{i=1}^2 f(u_i(B_i))$, where the right-hand side of the inequality corresponds to the allocation $(B_1, B_2)$.
Substituting in the respective values, we get $f(a) + f(b) \geq f(c_1) + f(d_1)$, a contradiction.

\underline{Case 2:} $a > d_1$.
We claim that $a \leq \min \{c_2, d_2\}$.
Since agent $1$ is not EF1 towards agent $2$ in $\mathcal{A}$, we have $a < c_2$.
Moreover, it holds that $a+d_2 > d_1+d_2 = U$, and $a < U/2$ implies that $a < U/2 < d_2$.
Therefore, $a \leq \min\{c_2, d_2\}$.

We show that $ab < c_2d_2$.
To this end, observe that
\begin{align*}
    c_2d_2 &= (U-c_1)(U-d_1) \\
    &= U(U-c_1-d_1) + c_1d_1 
    = U(c_2-d_1) + c_1d_1 
    > U(a-d_1) + c_1d_1 
    > c_1d_1,
\end{align*}
where the first inequality follows from $a < c_2$ and the second inequality from the assumption $a > d_1$.
Therefore, $ab < c_1d_1 < c_2d_2$.

Now, since $\min \{a, b\} = a \leq \min \{c_2, d_2\}$ and $ab < c_2d_2$, {\condrealnormtwo} implies that $f(a) + f(b) < f(c_2) + f(d_2)$.
On the other hand, since $\mathcal{A}$ is chosen by the additive welfarist rule with $f$, we have $\sum_{i=1}^2 f(u_i(A_i)) \geq \sum_{i=1}^2 f(u_i(B_{3-i}))$, where the right-hand side of the inequality corresponds to the allocation $(B_2, B_1)$.
Substituting in the respective values, we get $f(a) + f(b) \geq f(c_2) + f(d_2)$, a contradiction.

Since we have reached a contradiction in both cases, we conclude that $\mathcal{A}$ is EF1.
\end{proof}

We show in Appendix~\ref{ap:real_normalized_two} that beyond the $p$-mean rules for $p \leq 0$, there are (infinitely many) other additive welfarist rules that guarantee EF1 for normalized instances with two agents.
In particular, any linear combination of $\varphi_p$---the function defining the $p$-mean rule---satisfies {\condrealnormtwo}, so the additive welfarist rule with such a function also guarantees EF1 for this class of instances.
On the other hand, the additive welfarist rule with the modified logarithmic function $\lambda_c$ where $c > 0$ and the additive welfarist rule with the modified harmonic function $h_c$ where $c > -1$ do not guarantee EF1 for these instances.

\section{Integer-Valued Instances}
\label{sec:integer}

In this section, we turn our attention to \emph{integer-valued} instances, where the utility of each agent for each good must be a (non-negative) integer.
For these instances, the MNW rule is no longer the unique additive welfarist rule that guarantees EF1: the MHW rule exhibits the same property \citep{MontanariScSu25}.
We shall explore the conditions for rules to satisfy this property, and provide several examples of such rules.

\subsection{Identical-Good Instances}
\label{subsec:integer_idengood}

We begin with the class of (integer-valued) identical-good instances.
We show that the functions defining the additive welfarist rules that guarantee EF1 for such instances are precisely those that satisfy \emph{{\condintidengood}}.
Perhaps unsurprisingly, {\condintidengood} is similar to {\condrealgeneral} (which is the condition corresponding to \emph{real-valued} identical-good instances), with the difference being that the values of $a$ and $b$ in the conditions are positive real numbers in {\condrealgeneral} and positive integers in {\condintidengood}.

\begin{theorem}
\label{thm:integer_idengood}
Let $n \geq 2$ be given, and let $f: \Rpluszero \to \R \cup \{-\infty\}$ be a strictly increasing function.
Then, the following statements are equivalent:
\begin{enumerate}[label=(\alph*)]
    \item For every positive-admitting integer-valued identical-good instance with $n$ agents, every allocation chosen by the additive welfarist rule with $f$ is EF1.
    \item $f$ satisfies {\condintidengood}.
\end{enumerate}
\end{theorem}

Due to the similarity between Conditions {\condnorealgeneral} and {\condnointidengood}, the proof for (a) $\Rightarrow$ (b) in \Cref{thm:integer_idengood} follows similarly to that in the real-valued case (\Cref{lem:real_idengood_cond1}).
To prove (b) $\Rightarrow$ (a), we show that {\condintidengood} implies {\condintidengood}a, which in turn implies (a).
Note that the result by \citet{CaragiannisKuMo19} does not apply in this case since there exist functions $f$ satisfying {\condintidengood} that are not logarithmic, so the additive welfarist rule with~$f$ is not necessarily equivalent to the MNW rule.

\begin{lemma}
\label{lem:integer_idengood_cond3_3a}
Let $f: \Rpluszero \to \R \cup \{-\infty\}$ be a strictly increasing function that satisfies {\condintidengood}.
Then, $f$ satisfies {\condintidengood}a.
\end{lemma}

\begin{proof}
Let $k, \ell \in \Zpluszero$ and $a, b \in \Zplus$ be given such that $\ell < k$.
By applying {\condintidengood} repeatedly, we have $\Delta_{f, \ell}(b) > \Delta_{f, \ell+1}(a) > \cdots > \Delta_{f, k}(a)$, so $f$ indeed satisfies {\condintidengood}a.
\end{proof}

\begin{lemma}
\label{lem:integer_idengood_cond3a_ef1}
Let $n \geq 2$ be given, and let $f: \Rpluszero \to \R \cup \{-\infty\}$ be a strictly increasing function that satisfies {\condintidengood}a.
Then, the statement in \Cref{thm:integer_idengood}(a) holds.
\end{lemma}

\begin{proof}
Let a positive-admitting integer-valued identical-good instance be given, and let $\mathcal{A} = (A_1, \ldots, A_n)$ be an allocation chosen by the additive welfarist rule with $f$.
Assume, for the sake of contradiction, that $\mathcal{A}$ is not EF1.
Then, there exist two agents $i, j \in N$ such that agent $i$ envies agent $j$ by more than one good.
Hence, $0 \leq |A_i| \leq |A_j| - 2$.
Moreover, since the instance is integer-valued and identical-good, we have for each $g\in M$ that $u_i(g) = b$ and $u_j(g) = a$ for some $a, b \in \Zplus$.
Let $k = |A_j| - 1$ and $\ell = |A_i|$; note that $k, \ell \in \Zpluszero$ with $\ell < k$.
We have $\sum_{r \in N} f(u_r(A_r)) = f(\ell b) + f((k+1)a) + \sum_{r \in N \setminus \{i, j\}} f(u_r(A_r))$.

Let $\mathcal{B} = (B_1, \ldots, B_n)$ be the same allocation as $\mathcal{A}$ except that one good is transferred from agent~$j$'s bundle to agent $i$'s bundle.
Then, $\sum_{r \in N} f(u_r(B_r)) = f((\ell+1)b) + f(ka) + \sum_{r \in N \setminus \{i, j\}} f(u_r(A_r))$.
Since $\mathcal{A}$ is chosen by the additive welfarist rule with $f$, we have $\sum_{r \in N} f(u_r(B_r)) \leq \sum_{r \in N} f(u_r(A_r))$.
Rearranging the terms, we get 
\begin{align*}
f((\ell+1)b) - f(\ell b) \leq f((k+1)a) - f(ka),
\end{align*}
or equivalently, $\Delta_{f,\ell}(b) \leq \Delta_{f,k}(a)$, 
contradicting {\condintidengood}a.
Therefore, $\mathcal{A}$ is EF1.
\end{proof}

\begin{proof}[Proof of \Cref{thm:integer_idengood}]
The implication (a) $\Rightarrow$ (b) follows similarly as \Cref{lem:real_idengood_cond1}, while the implication (b) $\Rightarrow$ (a) follows from \Cref{lem:integer_idengood_cond3_3a,lem:integer_idengood_cond3a_ef1}.
\end{proof}

Since Conditions {\condnorealgeneral} and {\condnorealgeneral}a are equivalent (see \Cref{thm:real_idengood} and Lemmas~\ref{lem:real_idengood_cond1}, \ref{lem:real_idengood_constant}, and \ref{lem:real_idengood_log}), one may be tempted to think that there is a condition analogous to {\condrealgeneral}a that {\condintidengood} is equivalent to, perhaps ``$\Delta_{f, k}$ is a constant function (on domain $\Zplus$) for each $k \in \Zplus$''.
However, this condition is too strong, and is not implied by {\condintidengood}.\footnote{Indeed, the MHW rule---which corresponds to the function $h_0$---guarantees EF1 for integer-valued instances~\citep{MontanariScSu25}, but $\Delta_{h_0, 1}$ is not a constant function as $\Delta_{h_0, 1}(1) = h_0(2) - h_0(1) = 1/2 < 7/12 = h_0(4) - h_0(2) = \Delta_{h_0, 1}(2)$.}
Instead, we show that {\condintidengood} is equivalent to {\condintidengood}b.

\begin{restatable}{proposition}{propconditionthreeeqv}
\label{prop:condition_three_eqv}
Let $f: \Rpluszero \to \R \cup \{-\infty\}$ be a strictly increasing function.
Then, the following statements are equivalent:
\begin{enumerate}[label=(\alph*)]
    \item $f$ satisfies {\condintidengood}.
    \item $f$ satisfies {\condintidengood}a.
    \item $f$ satisfies {\condintidengood}b.
\end{enumerate}
\end{restatable}

\begin{proof}
(a) $\Rightarrow$ (b): This follows from \Cref{lem:integer_idengood_cond3_3a}.

(b) $\Rightarrow$ (c): This follows directly by definition.

(c) $\Rightarrow$ (a): Let $f$ be a function that satisfies {\condintidengood}b.
Let $k \in \Zpluszero$ and $a, b \in \Zplus$ be given as in {\condintidengood}.
If $b = 1$, then the left inequality of {\condintidengood}b implies that $f$ satisfies {\condintidengood}.
If $b > 1$ and $a = 1$, then for $k = 0$ we have $\Delta_{f, k}(b) = f(b) - f(0) > f(2) - f(1) = \Delta_{f, k+1}(a)$ where the inequality holds due to $f$ being strictly increasing, and for $k > 0$ the right inequality of {\condintidengood}b implies that $\Delta_{f, k}(b) > \Delta_{f, k+1}(a)$; hence $f$ satisfies {\condintidengood}.
Therefore, it suffices to prove that $f$ satisfies {\condintidengood} for $a, b \geq 2$.

Let $K_1 = (k+1)a-1$, $K_2 = (k+1)b+1$, $K_3 = (k+2)a-1$, and $K_4 = kb+1$.
By algebraic manipulation, we have $K_1 K_2 - K_3 K_4 = ab - a - b$, which is non-negative for $a, b \geq 2$.
This means that $K_1 K_2 \geq K_3 K_4$.
Hence, we have
\begin{align*}
    \Delta_{f, k}(b) &= \sum_{t=K_4-1}^{K_2-2} \Delta_{f, t}(1) \tag{by telescoping sum} \\
    &> \sum_{t=K_4-1}^{K_2-2} \Delta_{f, t+1} (K_3) \tag{since \protect{$f$} satisfies {\condintidengood}b} \\
    &= f(K_2 K_3) - f(K_4 K_3) \tag{by telescoping sum} \\
    &\geq f(K_3 K_2) - f(K_1 K_2) \tag{since \protect{$f$} is strictly increasing} \\
    &= \sum_{t=K_1}^{K_3-1} \Delta_{f, t} (K_2) \tag{by telescoping sum} \\
    &> \sum_{t=K_1}^{K_3-1} \Delta_{f, t+1} (1) \tag{since \protect{$f$} satisfies {\condintidengood}b} \\
    &= \Delta_{f, k+1}(a), \tag{by telescoping sum}
\end{align*}
which shows that $f$ satisfies {\condintidengood}.
\end{proof}

We next give examples of additive welfarist rules that guarantee EF1 for all integer-valued identical-good instances.
We study two families of functions: modified logarithmic functions $\lambda_c$ and modified harmonic functions $h_c$.
By leveraging \Cref{prop:condition_three_eqv} and \Cref{thm:integer_idengood}, we can determine which of these additive welfarist rules ensure EF1 for integer-valued identical-good instances.

\begin{restatable}{proposition}{propintegeridengoodlog}
\label{prop:integer_idengood_log}
The function $\lambda_c$ satisfies {\condintidengood}b if and only if $0 \leq c \leq 1$.
\end{restatable}

\begin{proof}
($\Leftarrow$) We show that $\lambda_c$ satisfies {\condintidengood}b for $0 \leq c \leq 1$.
Let $k \in \Zpluszero$ and $a \in \Zplus$ be given as in {\condintidengood}b.
We first prove the left inequality.
To this end, we use the following chain of equivalences:
\begin{alignat*}{2}
    &&\Delta_{\lambda_c, k}(1) &> \Delta_{\lambda_c, k+1}(a) \\
    &\Leftrightarrow\quad &\lambda_c(k+1) - \lambda_c(k) &> \lambda_c((k+2)a) - \lambda_c((k+1)a) \\
    &\Leftrightarrow &\log(k+1+c) - \log(k+c) &> \log((k+2)a+c) - \log((k+1)a+c) \tag{by definition of $\lambda_c$} \\
    &\Leftrightarrow &\log(k+1+c) + \log((k+1)a+c) &> \log(k+c) + \log((k+2)a+c) \\
    &\Leftrightarrow &(k+1+c)((k+1)a+c) &> (k+c)((k+2)a+c) \\
    &\Leftrightarrow &(k+1+c)((k+1)a+c) - (k+c)((k+2)a+c) &> 0.
\end{alignat*}
To simplify the expression in the last line, we observe that 
\begin{align*}
    (k+1+c)((k+1)a+c) - (k+c)((k+2)a+c) &= ((k+1)a+c) - (k+c)a \\
    &= a+c-ac \\
    &= a(1-c) + c.
\end{align*}
We have $a(1-c) + c > 0$ since $1-c, c \geq 0$, at least one of $1-c$ and $c$ is non-zero, and $a > 0$.
This implies that $(k+1+c)((k+1)a+c) > (k+c)((k+2)a+c)$, and equivalently, $\Delta_{\lambda_c, k}(1) > \Delta_{\lambda_c, k+1}(a)$, proving the left inequality of {\condintidengood}b.

We prove the right inequality of {\condintidengood}b by a similar method.
In this case, we have the following chain of equivalences:
\begin{alignat*}{2}
    &&\Delta_{\lambda_c, k+1}(a) &> \Delta_{\lambda_c, k+2}(1) \\
    &\Leftrightarrow\quad &\lambda_c((k+2)a) - \lambda_c((k+1)a) &> \lambda_c(k+3) - \lambda_c(k+2) \\
    &\Leftrightarrow &\log((k+2)a+c) - \log((k+1)a+c) &> \log(k+3+c) - \log(k+2+c) \tag{by definition of $\lambda_c$} \\
    &\Leftrightarrow &\log((k+2)a+c) + \log(k+2+c) &> \log((k+1)a+c) + \log(k+3+c) \\
    &\Leftrightarrow &((k+2)a+c)(k+2+c) &> ((k+1)a+c)(k+3+c) \\
    &\Leftrightarrow &((k+2)a+c)(k+2+c) - ((k+1)a+c)(k+3+c) &> 0.
\end{alignat*}
To simplify the expression in the last line, we observe that 
\begin{align*}
    ((k+2)a+c)(k+2+c) - ((k+1)a+c)(k+3+c) &= a(k+2+c) - ((k+1)a+c) \\
    &= a+ac-c \\
    &= a + c(a-1).
\end{align*}
We have $a + c(a-1) > 0$ since $a > 0$ and $c, a-1 \geq 0$.
This implies that $((k+2)a+c)(k+2+c) > ((k+1)a+c)(k+3+c)$, and equivalently, $\Delta_{\lambda_c, k+1}(a) > \Delta_{\lambda_c, k+2}(1)$, proving the right inequality of {\condintidengood}b.

($\Rightarrow$) We now show that $\lambda_c$ does not satisfy {\condintidengood}b for $c > 1$.
Take $a \in \Zplus$ such that $a \geq c/(c-1)$; note that such $a$ always exists when $c > 1$.
We claim that $\Delta_{\lambda_c, 0}(1) \leq \Delta_{\lambda_c, 1}(a)$, thus violating the left inequality of {\condintidengood}b.
We note the following chain of equivalences:
\begin{alignat*}{2}
    &&\Delta_{\lambda_c, 0}(1) &\leq \Delta_{\lambda_c, 1}(a) \\
    &\Leftrightarrow\quad &\lambda_c(1) - \lambda_c(0) &\leq \lambda_c(2a) - \lambda_c(a) \\
    &\Leftrightarrow &\log(1+c) - \log c &\leq \log(2a+c) - \log(a+c) \tag{by definition of $\lambda_c$} \\
    &\Leftrightarrow &\log(1+c) + \log(a+c) &\leq \log c + \log(2a+c) \\
    &\Leftrightarrow &(1+c)(a+c) &\leq c(2a+c).
\end{alignat*}
We can verify that the last line is true, since 
\begin{align*}
    (1+c)(a+c) - c(2a+c) &= a+c-ac \\
    &= c - a(c-1) \\
    &\leq c - c \tag{by the choice of $a \geq c/(c-1)$} \\
    &= 0.
\end{align*}
Hence, we have $\Delta_{\lambda_c, 0}(1) \leq \Delta_{\lambda_c, 1}(a)$, and so $\lambda_c$ does not satisfy {\condintidengood}b.
\end{proof}

\begin{restatable}{proposition}{propintegeridengoodharmonic}
\label{prop:integer_idengood_harmonic}
The function $h_c$ satisfies {\condintidengood}b if and only if $-1 \leq c \leq \frac{1}{\log 2} - 1$ ($\approx 0.443$).
\end{restatable}

The proof of \Cref{prop:integer_idengood_harmonic} involves complex calculations and is deferred to Appendix~\ref{ap:integer_idengood_harmonic}.

\Cref{prop:integer_idengood_log,prop:integer_idengood_harmonic} show that beyond the MNW rule and the MHW rule, many other additive welfarist rules guarantee EF1 for integer-valued identical-good instances.
On the other hand, we observe that all such functions must be similar to the logarithmic function in that the marginal increase $f(x+1) - f(x)$ grows as $\Theta(1/x)$.
As a result of this, all $p$-mean rules with $p \neq 0$ do not guarantee EF1 for integer-valued identical-good instances.
We discuss these results formally in Appendix~\ref{ap:integer_idengood}.

\subsection{Binary Instances}
\label{subsec:integer_binary}

Next, we consider binary instances, where each good is worth either $0$ or $1$ to each agent.
We show that the functions defining the additive welfarist rules that guarantee EF1 for all binary instances are those that satisfy \emph{{\condintbinary}}, which is equivalent to the function being strictly concave in the domain of non-negative integers.

\begin{restatable}{theorem}{thmintegerbinary}
\label{thm:integer_binary}
Let $n \geq 2$ be given, and let $f: \Rpluszero \to \R \cup \{-\infty\}$ be a strictly increasing function.
Then, the following statements are equivalent:
\begin{enumerate}[label=(\alph*)]
    \item For every positive-admitting binary instance with $n$ agents, every allocation chosen by the additive welfarist rule with $f$ is EF1.
    \item $f$ satisfies {\condintbinary}.
\end{enumerate}
\end{restatable}

\begin{proof}
(a) $\Rightarrow$ (b): Let $k \in \Zpluszero$.
Consider an instance with $n$ agents and $2k+n$ goods, and let $M' = \{g_1, \ldots, g_{2k+2}\}$.
The utilities of the goods are as follows.
\begin{itemize}
    \item For $i \in \{1, 2\}$, $u_i(g) = 1$ for $g \in M'$, and $u_i(g) = 0$ otherwise.
    \item For $i \in \{3, \ldots, n\}$, $u_i(g_{2k+i}) = 1$, and $u_i(g) = 0$ for all other $g \in M$.
\end{itemize}
This instance is a binary instance; moreover, it is positive-admitting since any allocation where each agent $i \in N$ receives $g_{2k+i}$ gives positive utility to every agent.

Let $\mathcal{A} = (A_1, \ldots, A_n)$ be an allocation chosen by the additive welfarist rule with $f$.
Each good must be allocated to an agent who has utility $1$ for the good; otherwise, transferring such a good to an agent who has utility $1$ for the good increases the value of $\sum_{i \in N} f(u_i(A_i))$, contradicting the assumption that $\mathcal{A}$ is chosen by the additive welfarist rule with $f$.
Accordingly, the goods in $M'$ are allocated to agent $1$ and $2$, and for each $i \in \{3, \ldots, n\}$, $g_{2k+i}$ is allocated to agent $i$.

Recall that $\mathcal{A}$ is EF1 since it is chosen by the additive welfarist rule with $f$.
We claim that agent~$1$ and $2$ each receives exactly $k+1$ goods from $M'$.
If not, then one of them receives at most $k$ goods from $M'$, and will envy the other agent (who receives at least $k+2$ goods from $M'$) by more than one good, making the allocation not EF1 and contradicting our assumption.
Then, we have $\sum_{i \in N} f(u_i(A_i)) = 2f(k+1) + (n-2)f(1)$.

Let $\mathcal{B} = (B_1, \ldots, B_n)$ be the allocation such that agent $1$ receives $k$ goods from $M'$, agent~$2$ receives $k+2$ goods from $M'$, and for each $i \in \{3, \ldots, n\}$, agent $i$ receives $g_{2k+i}$.
We have $\sum_{i \in N} f(u_i(B_i)) = f(k) + f(k+2) + (n-2)f(1)$.
Note that $\mathcal{B}$ is not EF1, and cannot be chosen by the additive welfarist rule with $f$.
Therefore, we have $\sum_{i \in N} f(u_i(A_i)) > \sum_{i \in N} f(u_i(B_i))$.
Rearranging the terms, we get 
\begin{align*}
f(k+1) - f(k) > f(k+2) - f(k+1), 
\end{align*}
or equivalently, $\Delta_{f, k}(1) > \Delta_{f, k+1}(1)$.
Since $k \in \Zpluszero$ was arbitrarily chosen, $\Delta_{f, k}(1) > \Delta_{f, k+1}(1)$ holds for all $k \in \Zpluszero$.
Therefore, $f$ satisfies {\condintbinary}.

(b) $\Rightarrow$ (a): Let a positive-admitting binary instance with $n$ agents be given, and let $\mathcal{A} = (A_1, \ldots, A_n)$ be an allocation chosen by the additive welfarist rule with $f$.
Assume, for the sake of contradiction, that $\mathcal{A}$ is not EF1.
Then, there exist $i, j \in N$ such that agent~$i$ envies agent $j$ by more than one good, i.e., $u_i(A_i) < u_i(A_j \setminus \{g\})$ for all $g \in A_j$.
Note that we must have $u_i(g) \leq u_j(g)$ for all $g \in A_j$.
Indeed, otherwise we have $u_i(g) = 1$ and $u_j(g) = 0$, and transferring $g$ from agent $j$'s bundle to agent $i$'s bundle increases $f(u_i(A_i))$ and does not decrease $f(u_j(A_j))$, thereby increasing $\sum_{k \in N} f(u_k(A_k))$ and contradicting the assumption that $\mathcal{A}$ is chosen by the additive welfarist rule with $f$.
Moreover, since $u_i(A_j) > 0$, there exists $g' \in A_j$ such that $u_i(g') = 1$ (and hence $u_j(g') = 1$).

Let $\mathcal{B} = (B_1, \ldots, B_n)$ be the same allocation as $\mathcal{A}$ except that $g'$ is transferred from agent $j$'s bundle to agent $i$'s bundle, i.e., $B_i = A_i \cup \{g'\}$, $B_j = A_j \setminus \{g'\}$, and $B_k = A_k$ for all $k \in N \setminus \{i, j\}$.
Note that $u_i(A_i) < u_i(A_j \setminus \{g'\}) = u_i(B_j)$, and that $u_i(B_j) \leq u_j(B_j)$ since $u_i(g) \leq u_j(g)$ for all $g \in B_j$.
Therefore, $u_i(A_i) < u_j(B_j)$.
By {\condintbinary}, it holds that $\Delta_{f, u_i(A_i)}(1) > \cdots > \Delta_{f, u_j(B_j)}(1)$.
On the other hand, we have $\sum_{k \in N} f(u_k(B_k)) = f(u_i(B_i)) + f(u_j(B_j)) + \sum_{k \in N \setminus \{i, j\}} f(u_k(A_k))$.
Since $\mathcal{A}$ is chosen by the additive welfarist rule with $f$, we have $\sum_{k \in N} f(u_k(B_k)) \leq \sum_{k \in N} f(u_k(A_k))$.
Rearranging the terms, we get
\begin{align*}
    f(u_i(B_i)) - f(u_i(A_i)) \leq f(u_j(A_j)) - f(u_j(B_j)),
\end{align*}
or equivalently, $\Delta_{f, u_i(A_i)}(1) \leq \Delta_{f, u_j(B_j)}(1)$, a contradiction.
Therefore, $\mathcal{A}$ is EF1.
\end{proof}

{\condintbinary} is weaker than {\condintidengood}.
Indeed, if $f$ satisfies {\condintidengood}, then it satisfies {\condintbinary} by definition.
However, the converse is not true: the function $f(x) = \sqrt{x}$ satisfies {\condintbinary}, but it does not satisfy {\condintidengood} since $\Delta_{f, 0}(1) = 1$ is not greater than $\Delta_{f, 1}(6) \approx 1.01$.
This means that each additive welfarist rule that guarantees EF1 for every positive-admitting integer-valued identical-good instance also guarantees EF1 for every positive-admitting binary instance, but the converse does not hold.

\begin{proposition}
\label{prop:integer_idengood_implies_binary}
Let $f: \Rpluszero \to \R \cup \{-\infty\}$ be a strictly increasing function.
If $f$ satisfies {\condintidengood}, then it satisfies {\condintbinary}.
\end{proposition}

We mentioned earlier that {\condintbinary} is equivalent to the function being strictly concave in the domain of non-negative integers. On the other hand, the function $\varphi_p$ defining the $p$-mean rule is strictly concave if and only if $p < 1$.
We prove that, indeed, $\varphi_p$ satisfies {\condintbinary} exactly when $p < 1$.

\begin{proposition}
\label{prop:integer_binary_pmean}
The function $\varphi_p$ satisfies {\condintbinary} if and only if $p < 1$.
\end{proposition}

\begin{proof}
($\Rightarrow$) Let $p \geq 1$.
Then, $\varphi_p(x) = x^p$ and so $\Delta_{\varphi_p, k}(1) = (k+1)^p - k^p$.
Note that $\Delta_{\varphi_p, 0}(1) = 1^p - 0^p = 1$.
Furthermore, $\Delta_{\varphi_p, 1}(1) = 2^p - 1^p \geq 2 - 1 = 1$ since $p \geq 1$.
Hence,  $\Delta_{\varphi_p, k}(1) \leq \Delta_{\varphi_p, k+1}(1)$ when $k = 0$, so $\varphi_p$ does not satisfy {\condintbinary}.

($\Leftarrow$) Let $p < 1$.
We first prove that, for any $x \in \Rpluszero$, the second derivative $\varphi_p''(x)$ is always negative (although not necessarily finite).
Note that $\varphi_p'(x) = C_px^{p-1}$ for some positive constant $C_p$, where $C_p = p$ if $0 < p < 1$, $C_p = 1$ if $p = 0$, and $C_p = -p$ if $p < 0$. 
Hence, $\varphi_p''(x) = C_p(p-1)x^{p-2}$. 
Since $p < 1$, we have $C_p(p-1) < 0$. 
Then, for all $x \in \Rpluszero$, $\varphi_p''(x) < 0$.
This means that $\varphi_p$ is strictly concave in the non-negative real domain.

Let $k \in \Zpluszero$. If $p \leq 0$ and $k = 0$, then $\Delta_{\varphi_p, 0} (1) = \infty > \Delta_{\varphi_p, 1}(1)$.
We now assume that $p > 0$ or $k > 0$, and so $\varphi_p(k), \varphi_p(k+1), \varphi_p(k+2)$ are all finite.
Since $\varphi_p$ is strictly concave in the non-negative real domain, by letting $x = k$, $y = k+2$, and $\alpha = 1/2$ in the definition of strict concavity, we get $\varphi_p(k+1) > \varphi_p(k)/2 + \varphi_p(k+2)/2$. 
Hence, $2\varphi_p(k+1) > \varphi_p(k) + \varphi_p(k+2)$, which implies that $\Delta_{\varphi_p, k}(1) = \varphi_p(k+1) - \varphi_p(k) > \varphi_p(k+2) - \varphi_p(k+1) = \Delta_{\varphi_p, k+1}(1)$.
It follows that $\varphi_p$ satisfies {\condintbinary}.
\end{proof}

\subsection{Two-Value Instances}
\label{subsec:integer_twovalue}

We now consider integer-valued instances which are two-value, a strict generalization of binary instances.
Recall from \Cref{subsec:real_twovalue} that the classes of identical-good instances and two-value instances are not subclasses of each other (except when $n = 2$).
Therefore, the techniques used in \Cref{subsec:integer_idengood} cannot be used in this section.\footnote{In particular, the proof that {\condintidengood}a is a sufficient condition for the identical-good case (\Cref{lem:integer_idengood_cond3a_ef1}) does not apply to the two-value case.}

We show that the characterization for integer-valued two-value instances is given by \emph{{\condinttwovalue}}.
Before we state the characterization, we examine the relationship between Conditions {\condnointidengood} and~{\condnointtwovalue}.

\begin{restatable}{proposition}{propintegertwovaluesufficient}
\label{prop:integer_twovalue_sufficient}
Let $f: \Rpluszero \to \R \cup \{-\infty\}$ be a strictly increasing function.
\begin{enumerate}[label=(\roman*)]
    \item If $f$ satisfies {\condinttwovalue}, then $f$ satisfies {\condintidengood}.
    \item If $f$ satisfies {\condintidengood} and the property that for each $k \in \Zplus$, the function $\Delta_{f, k}$ is either non-decreasing or non-increasing on the domain of positive integers, then $f$ satisfies {\condinttwovalue}.
\end{enumerate}
\end{restatable}

\begin{proof}
(i) Substituting $b = 1$, $\ell = k$, and $r = 0$ into {\condinttwovalue} gives {\condintidengood}b, which is equivalent to {\condintidengood} by \Cref{prop:condition_three_eqv}.

(ii) Let $a, b \in \Zplus$ and $k, \ell, r \in \Zpluszero$ such that $a \geq b$ and $(k+1)b > \ell b + ra$.
Note that $f$ trivially satisfies the right inequality of {\condinttwovalue}; therefore, it suffices to prove the left inequality. Since $f$ satisfies {\condintidengood}, \Cref{prop:condition_three_eqv} implies that $f$ satisfies {\condintidengood}a.

If $\Delta_{f, k+1}$ is non-increasing, then
\begin{align*}
    f((\ell+1)b+ra) - f(\ell b+ra) &= \sum_{t=0}^{b-1} \Delta_{f, \ell b+ra+t}(1) \tag{by telescoping sum} \\
    &> \sum_{t=0}^{b-1} \Delta_{f, (k+1)b+t}(1) \tag{by {\condintidengood}a and \protect{$(k+1)b > \ell b+ra$}} \\
    &= f((k+2)b) - f((k+1)b) \tag{by telescoping sum} \\
    &= \Delta_{f, k+1}(b) \\
    &\geq \Delta_{f, k+1}(a),
\end{align*}
where the last inequality holds because $\Delta_{f, k+1}$ is non-increasing and $a \geq b$.
This proves the left inequality of {\condinttwovalue}.

On the other hand, if $\Delta_{f, k+1}$ is non-decreasing, then
\begin{align*}
    f((\ell+1)b+ra) - f(\ell b+ra) &= \sum_{t=0}^{b-1} \Delta_{f, \ell b+ra+t}(1) \tag{by telescoping sum} \\
    &> \sum_{t=0}^{b-1} \Delta_{f, (k+1)b+t}(a)
    \tag{by {\condintidengood}a and \protect{$(k+1)b > \ell b+ra$}}\\
    &= f((k+2)ba) - f((k+1)ba) \tag{by telescoping sum} \\
    &= \Delta_{f, k+1}(ab) \\
    &\geq \Delta_{f, k+1}(a),
\end{align*}
where the last inequality holds because $\Delta_{f, k+1}$ is non-decreasing and $ab \geq a$.
This proves the left inequality of {\condinttwovalue}.
\end{proof}

We now state the characterization for integer-valued two-value instances.

\begin{restatable}{theorem}{thmintegertwovalue}
\label{thm:integer_twovalue}
Let $n \geq 2$ be given, and let $f: \Rpluszero \to \R \cup \{-\infty\}$ be a strictly increasing function.
Then, the following statements are equivalent:
\begin{enumerate}[label=(\alph*)]
    \item For every positive-admitting integer-valued two-value instance with $n$ agents, every allocation chosen by the additive welfarist rule with $f$ is EF1.
    \item $f$ satisfies {\condinttwovalue}.
\end{enumerate}
\end{restatable}

\begin{proof}
(a) $\Rightarrow$ (b): By a similar proof as that in \Cref{lem:real_idengood_cond1}, $f$ satisfies {\condintidengood}, which implies that $f$ satisfies the right inequality of {\condinttwovalue}.
Therefore, it suffices to show that $f$ satisfies the left inequality of {\condinttwovalue}. 
Note also that since $f$ satisfies {\condintidengood}, by \Cref{prop:condition_three_eqv}, $f$ satisfies {\condintidengood}a as well.

Let $a, b \in \Zplus$ and $k, \ell, r \in \Zpluszero$ be given such that $a \geq b$ and $(k+1)b > \ell b + ra$.
If $a = b$, then by {\condintidengood}a, we have $\Delta_{f, \ell+r}(a) > \Delta_{f, k+1}(a)$, which implies $f((\ell+1)a+ra) - f(\ell a+ra) > \Delta_{f, k+1}(a)$, and thus $f$ satisfies {\condinttwovalue}.
Therefore, we shall consider the case $a > b$.

Assume, for the sake of contradiction, that $f((\ell+1)b+ra) - f(\ell b+ra) \leq \Delta_{f, k+1}(a)$.
Consider an instance with $n$ agents and $(k+1)(n-1) + \ell + r + 1$ goods, and let $M' = \{1, \ldots, r\}$.
The utilities of the goods are as follows.
\begin{itemize}
    \item $u_1(g) = a$ for $g \in M'$, and $u_1(g) = b$ otherwise.
    \item For $i \in \{2, \ldots, n\}$, $u_i(g) = a$ for all $g \in M$.
\end{itemize}
This instance is an integer-valued two-value instance; moreover, it is positive-admitting since there are at least $n$ goods, and an allocation where every agent receives at least one good gives positive utility to every agent.

Among all allocations chosen by the additive welfarist rule with $f$, let $\mathcal{A} = (A_1, \ldots, A_n)$ be an allocation with the smallest value of $|A_1|$.
By assumption, $\mathcal{A}$ is EF1.
We claim that $|A_1| \geq \ell + r + 1$.
If $|A_1| \leq \ell + r$, then there are more than $(k+1)(n-1)$ goods allocated to the agents in $N \setminus \{1\}$, and hence some agent $i \in N \setminus \{1\}$ receives at least $k+2$ goods.
Then, we have $u_1(A_1) \leq \ell b + ra < (k+1)b \leq u_1(A_i \setminus \{g\})$ for all $g \in A_i$, contradicting the assumption that $\mathcal{A}$ is EF1.
Now, since $|A_1| \geq \ell + r + 1 > r = |M'|$, there exists a good $g' \in M \setminus M'$ in agent $1$'s bundle.
Moreover, every good in $M'$ must be in $A_1$.
Indeed, if some good $g \in M'$ is in agent $i$'s bundle for some $i \in N \setminus \{1\}$, then exchanging $g$ with $g'$ increases agent $1$'s utility for her own bundle and does not decrease other agents' utilities for their own bundles, thereby increasing $\sum_{j \in N} u_j(A_j)$ and contradicting the assumption that $\mathcal{A}$ is chosen by the additive welfarist rule with $f$.
This gives $K := u_1(A_1) \geq (\ell+1)b + ra$.

Since agent $1$ has at least $\ell + r + 1$ goods, at most $(k+1)(n-1)$ goods are allocated to the agents in $N \setminus \{1\}$.
In particular, there exists an agent $i \in N \setminus \{1\}$ who has at most $k+1$ goods, and thus $L := |A_i| \leq k+1$.
Let $\mathcal{B} = (B_1, \ldots, B_n)$ be the same allocation as $\mathcal{A}$ except that $g'$ is transferred from agent $1$'s bundle to agent $i$'s bundle, i.e., $B_1 = A_1 \setminus \{g'\}$, $B_i = A_i \cup \{g'\}$, and $B_j = A_j$ for all $j \in N \setminus \{1, i\}$.
Since $|B_1| < |A_1|$, $\mathcal{B}$ cannot be chosen by the additive welfarist rule with $f$, and hence $\sum_{j \in N} f(u_j(A_j)) > \sum_{j \in N} f(u_j(B_j))$.
This simplifies to $f(u_1(A_1)) + f(u_i(A_i)) > f(u_1(B_1)) + f(u_i(B_i))$, or $f(K) - f(K-b) > f((L+1)a) - f(La)$.
On the other hand, we have
\begin{align*}
    f(K) - f(K-b) &= \sum_{t=0}^{b-1} \Delta_{f, K-b+t}(1) \tag{by telescoping sum} \\
    &\leq \sum_{t=0}^{b-1} \Delta_{f, ((\ell+1)b+ra)-b+t}(1) \tag{by {\condintidengood}a and \protect{$K \geq (\ell+1)b+ra$}} \\
    &= \sum_{t=0}^{b-1} \Delta_{f, \ell b+ra+t}(1) \\
    &= f((\ell+1)b+ra) - f(\ell b+ra) \tag{by telescoping sum} \\
    &\leq \Delta_{f, k+1}(a) \tag{by assumption} \\
    &\leq \Delta_{f, L}(a) \tag{by {\condintidengood}a and \protect{$L \leq k+1$}} \\
    &= f((L+1)a) - f(La),
\end{align*}
a contradiction.
Therefore, $f((\ell+1)b+ra) - f(\ell b+ra) > \Delta_{f, k+1}(a)$ holds, and so $f$ satisfies {\condinttwovalue}.

(b) $\Rightarrow$ (a): Let a positive-admitting integer-valued two-value instance be given, and let $\mathcal{A} = (A_1, \ldots, A_n)$ be an allocation chosen by the additive welfarist rule with $f$.
Since the instance is two-value, there exist distinct $a, b \in \Zpluszero$ such that $u_i(g) \in \{a, b\}$ for all $i \in N$ and $g \in M$; assume without loss of generality that $a > b$.
Assume, for the sake of contradiction, that $\mathcal{A}$ is not EF1.
Then, there exist two agents $i, j \in N$ such that agent $i$ envies agent $j$ by more than one good.
Note that $f$ satisfies {\condintidengood} by \Cref{prop:integer_twovalue_sufficient}(a), and so $f$ also satisfies {\condintidengood}a by \Cref{prop:condition_three_eqv}.

We consider the following two cases.

\underline{Case 1:} $u_i(A_j) \leq u_j(A_j)$.
We show that there exists $g' \in A_j$ such that $u_i(g') \geq u_j(g')$.
Suppose on the contrary that $u_i(g) < u_j(g)$ for all $g \in A_j$.
Then, we have $u_i(g) = b$ and $u_j(g) = a$ for all $g \in A_j$.
Since agent $i$ is not EF1 towards agent $j$, we have $|A_j| \geq 2$, and $b|A_j| = u_i(A_j) > u_i(A_i) \geq 0$ implies that $a > b > 0$.
Let $k = |A_j|-2$, and let $\ell$ and $r$ be the number of goods in~$A_i$ with utility $b$ and $a$ respectively to agent $i$.
Note that we have $a, b \in \Zplus$ and $k, \ell, r \in \Zpluszero$.
Moreover, $u_i(A_j \setminus \{g\}) > u_i(A_i)$ for any $g \in A_j$, which implies $(k+1)b > \ell b+ra$.
Since $f$ satisfies {\condinttwovalue}, we have $f((\ell+1)b+ra) - f(\ell b+ra) > \Delta_{f, k+1}(a)$.
Rearranging the terms, we get $f((\ell+1)b+ra) + f((k+1)a) > f(\ell b+ra) + f((k+2)a)$.
Now, 
\begin{align*}
    \sum_{s \in N} f(u_s(A_s)) &=  f(u_i(A_i)) +  f(u_j(A_j)) + \sum_{s \in N \setminus \{i, j\}} f(u_s(A_s)) \\
    &= f(\ell b+ra) + f((k+2)a) + \sum_{s \in N \setminus \{i, j\}} f(u_s(A_s)).
\end{align*}
Let $\mathcal{B} = (B_1, \ldots, B_n)$ be the same allocation as $\mathcal{A}$ except that one good is transferred from agent~$j$'s bundle to agent $i$'s bundle.
Then, we have
\begin{align*}
    \sum_{s \in N} f(u_s(B_s)) = f((\ell+1)b+ra) &+ f((k+1)a) 
    + \sum_{s \in N \setminus \{i, j\}} f(u_s(A_s)).
\end{align*}
Since $\mathcal{A}$ is chosen by the additive welfarist rule with $f$, we have $\sum_{s \in N} f(u_s(B_s)) \leq \sum_{s \in N} f(u_s(A_s))$.
Rearranging the terms, we get $f((\ell+1)b+ra) + f((k+1)a) \leq f(\ell b+ra) + f((k+2)a)$, violating {\condinttwovalue}.
Therefore, there exists $g' \in A_j$ such that $u_i(g') \geq u_j(g')$.

If $u_i(g') > 0$, then let $g'' = g'$.
Otherwise, $u_i(g') = 0 = b$.
Since $u_i(A_j) > u_i(A_i) \geq 0$, there exists a good $g'' \in A_j$ such that $u_i(g'') > 0$, and hence $u_i(g'') = a \geq u_j(g'')$.
In either case, we have $u_i(g'') \geq u_j(g'')$ and $u_i(g'') > 0$.
Let $K = u_i(A_i)$ and $L = u_j(A_j \setminus \{g''\})$, and note that
\begin{align*}
    K &= u_i(A_i) \\
    &< u_i(A_j \setminus \{g''\}) \tag{since agent \protect{$i$} is not EF1 towards agent \protect{$j$}} \\
    &= u_i(A_j) - u_i(g'') \\
    &\leq u_j(A_j) - u_j(g'') \tag{since \protect{$u_i(A_j) \leq u_j(A_j)$} and \protect{$u_i(g'') \geq u_j(g'')$}} \\
    &= u_j(A_j \setminus \{g''\}) 
    = L.
\end{align*}
Hence,
\begin{align*}
    f(K + u_i(g'')) - f(K) &= \sum_{t=0}^{u_i(g'')-1} \Delta_{f, K+t}(1) \tag{by telescoping sum} \\
    &> \sum_{t=0}^{u_i(g'')-1} \Delta_{f, L+t}(1) \tag{by {\condintidengood}a, \protect{$u_i(g'') > 0$}, and \protect{$K < L$}} \\
    &= f(L + u_i(g'')) - f(L) \tag{by telescoping sum} \\
    &\geq f(L + u_j(g'')) - f(L). \tag{since \protect{$u_i(g'') \geq u_j(g'')$} and \protect{$f$} is increasing}
\end{align*}

On the other hand, we have
\begin{align*}
    \sum_{s \in N} f(u_s(A_s)) = f(K) + f(L + u_j(g'')) + \sum_{s \in N \setminus \{i, j\}} f(u_s(A_s)).
\end{align*}
Let $\mathcal{B}' = (B'_1, \ldots, B'_n)$ be the same allocation as $\mathcal{A}$ except that $g''$ is transferred from agent $j$'s bundle to agent $i$'s bundle.
We have 
\begin{align*}
    \sum_{s \in N} f(u_s(B'_s)) = f(K + u_i(g'')) + f(L) + \sum_{s \in N \setminus \{i, j\}} f(u_s(A_s)).
\end{align*}
Since $\mathcal{A}$ is chosen by the additive welfarist rule with $f$, we have $\sum_{s \in N} f(u_s(B'_s)) \leq \sum_{s \in N} f(u_s(A_s))$.
Rearranging the terms, we get $f(K + u_i(g'')) - f(K) \leq f(L + u_j(g'')) - f(L)$, which is a contradiction.

\underline{Case 2:} $u_i(A_j) > u_j(A_j)$.
Then, there exists $g' \in A_j$ such that $u_i(g') > u_j(g')$.
This implies that $u_i(g') = a$ and $u_j(g') = b$.

We show that $|A_i| > 0$.
If $|A_i| = 0$, then $A_i = \emptyset$, and
\begin{align*}
    f(u_i(g')) - f(0) &> f(u_j(g')) - f(0) \tag{since \protect{$u_i(g') > u_j(g')$}} \\
    &= \sum_{t=0}^{u_j(g')-1} \Delta_{f, t}(1) \tag{by telescoping sum} \\
    &\geq \sum_{t=0}^{u_j(g')-1} \Delta_{f, u_j(A_j \setminus \{g'\})+t}(1) \tag{by {\condintidengood}a and \protect{$u_j(A_j \setminus \{g'\}) \geq 0$}} \\
    &= f(u_j(A_j \setminus \{g'\})+u_j(g')) - f(u_j(A_j \setminus \{g'\})) \tag{by telescoping sum} \\
    &= f(u_j(A_j)) - f(u_j(A_j \setminus \{g'\})).
\end{align*}
On the other hand, we have
\begin{align*}
    \sum_{s \in N} f(u_s(A_s)) = f(0) + f(u_j(A_j)) + \sum_{s \in N \setminus \{i, j\}} f(u_s(A_s)).
\end{align*}
Let $\mathcal{B} = (B_1, \ldots, B_n)$ be the same allocation as $\mathcal{A}$ except that $g'$ is transferred from agent $j$'s bundle to agent $i$'s bundle.
We have 
\begin{align*}
    \sum_{s \in N} f(u_s(B_s)) = f(u_i(g')) &+ f(u_j(A_j \setminus \{g'\})) 
    + \sum_{s \in N \setminus \{i, j\}} f(u_s(A_s)).
\end{align*}
Since $\mathcal{A}$ is chosen by the additive welfarist rule with $f$, we have $\sum_{s \in N} f(u_s(B_s)) \leq \sum_{s \in N} f(u_s(A_s))$.
Rearranging the terms, we get $f(u_i(g')) - f(0) \leq f(u_j(A_j)) - f(u_j(A_j \setminus \{g'\}))$, which is a contradiction.
Therefore, $|A_i| > 0$.

Next, we show that $u_i(g) = a$ for all $g \in A_i$.
If there exists a good $g'' \in A_i$ such that $u_i(g'') = b$, then exchanging $g'$ with $g''$ increases agent $i$'s utility for her own bundle and does not decrease other agents' utilities for their own bundles, thereby increasing $\sum_{s \in N} f(u_s(A_s))$ and contradicting the assumption that $\mathcal{A}$ is chosen by the additive welfarist rule with $f$.
This shows that $u_i(g) = a$ for all $g \in A_i$, and hence $u_i(A_i) = a|A_i| > 0$.

Let $K = u_i(A_i) = a|A_i|$ and $L = u_j(A_j \setminus \{g'\})$.
Note that $a|A_i| = u_i(A_i) < u_i(A_j \setminus \{g'\}) \leq a(|A_j|-1)$ and hence $|A_i| < |A_j|-1$.
Moreover, $L = u_j(A_j \setminus \{g'\}) \geq b(|A_j|-1)$ since each good has utility at least $b$.
Then, we have
\begin{align*}
    f(K+a) - f(K) &= f(a(|A_i|+1)) - f(a|A_i|) \\
    &= \Delta_{f, |A_i|}(a) \\
    &> \Delta_{f, |A_j|-1}(b) \tag{by {\condintidengood}a and \protect{$|A_i| < |A_j|-1$}} \\
    &= f(b|A_j|) - f(b(|A_j|-1)) \\
    &= \sum_{t=0}^{b-1} \Delta_{f, b(|A_j|-1)+t}(1) \tag{by telescoping sum} \\
    &\geq \sum_{t=0}^{b-1} \Delta_{f, L+t}(1) \tag{by {\condintidengood}a and \protect{$L \geq b(|A_j|-1)$}} \\
    &= f(L+b) - f(L). \tag{by telescoping sum}
\end{align*}
On the other hand, we have
\begin{align*}
    \sum_{s \in N} f(u_s(A_s)) = f(K) + f(L+b) + \sum_{s \in N \setminus \{i, j\}} f(u_s(A_s)).
\end{align*}
Let $\mathcal{B} = (B_1, \ldots, B_n)$ be the same allocation as $\mathcal{A}$ except that $g'$ is transferred from agent $j$'s bundle to agent $i$'s bundle.
We have 
\begin{align*}
    \sum_{s \in N} f(u_s(B_s)) = f(K+a) + f(L) + \sum_{s \in N \setminus \{i, j\}} f(u_s(A_s)).
\end{align*}
Since $\mathcal{A}$ is chosen by the additive welfarist rule with $f$, we have $\sum_{s \in N} f(u_s(B_s)) \leq \sum_{s \in N} f(u_s(A_s))$.
Rearranging the terms, we get $f(K+a) - f(K) \leq f(L+b) = f(L)$, which is a contradiction.

Since we have reached a contradiction in both cases, we conclude that $\mathcal{A}$ is EF1.
\end{proof}

\Cref{prop:integer_twovalue_sufficient} says that {\condinttwovalue} is stronger than {\condintidengood}.
Therefore, \Cref{thm:integer_twovalue} implies that every additive welfarist rule that guarantees EF1 for every positive-admitting integer-valued two-value instance also guarantees EF1 for every positive-admitting integer-valued identical-good instance, even though these two classes of instances are not subclasses of each other.

Next, we give examples of additive welfarist rules that ensure EF1 for all integer-valued two-value instances.
We use Propositions~\ref{prop:integer_idengood_log}, \ref{prop:integer_idengood_harmonic}, and \ref{prop:integer_twovalue_sufficient} to prove these results.
Recall the definitions of the modified logarithmic function $\lambda_c$ and modified harmonic function $h_c$ from \Cref{subsec:examples}.

\begin{restatable}{proposition}{propintegertwovaluelog}
\label{prop:integer_twovalue_log}
The function $\lambda_c$ satisfies {\condinttwovalue} if and only if $0 \leq c \leq 1$.
\end{restatable}

\begin{proof}
($\Leftarrow$) We show that $\lambda_c$ satisfies {\condinttwovalue} for $0 \leq c \leq 1$.
By \Cref{prop:integer_idengood_log,prop:condition_three_eqv}, $\lambda_c$ satisfies {\condintidengood}.
By \Cref{prop:integer_twovalue_sufficient}, it suffices to prove that $\Delta_{\lambda_c, k}$ is non-decreasing on the domain of positive integers for each $k \in \Zplus$.
For each positive integer $x$, we have
\begin{align*}
    \Delta_{\lambda_c, k}(x) &= \lambda_c((k+1)x) - \lambda_c(kx) \\
    &= \log ((k+1)x + c) - \log (kx + c) 
    = \log \left(\frac{(k+1)x+c}{kx+c}\right) 
    = \log \left(1 + \frac{1}{k+c/x} \right),
\end{align*}
which is non-decreasing in $x$.
Therefore, $\lambda_c$ satisfies {\condinttwovalue}.

($\Rightarrow$) Since $\lambda_c$ does not satisfy {\condintidengood} for $c > 1$ by \Cref{prop:integer_idengood_log}, it does not satisfy {\condinttwovalue} by \Cref{prop:integer_twovalue_sufficient}.
\end{proof}

\begin{restatable}{proposition}{propintegertwovalueharmonic}
\label{prop:integer_twovalue_harmonic}
The function $h_c$ satisfies {\condinttwovalue} if and only if $-1 \leq c \leq 1/\log 2 - 1$ ($\approx 0.443$).
\end{restatable}

The proof of \Cref{prop:integer_twovalue_harmonic} involves complex calculations and is deferred to Appendix~\ref{ap:integer_twovalue_harmonic}.

\subsection{General Integer-Valued Instances}
\label{subsec:integer_general}

Finally, we consider the class of all integer-valued instances.
It turns out to be challenging to find an exact characterization for this class of instances.
Instead, we provide a necessary condition in the form of {\condintgeneral}a and a sufficient condition in the form of {\condintgeneral}b.

\begin{restatable}{proposition}{propintegergeneralnecessary}
\label{prop:integer_general_necessary}
Let $n \geq 2$ be given, and let $f: \Rpluszero \to \R \cup \{-\infty\}$ be a strictly increasing function such that for every positive-admitting integer-valued instance with $n$ agents, every allocation chosen by the additive welfarist rule with $f$ is EF1.
Then, $f$ satisfies {\condintgeneral}a.
\end{restatable}

\begin{proof}
Let $k, a, b \in \Zplus$ be given. 
Consider an instance with $n$ agents and $kn+1$ goods, and let $M' = M \setminus \{g_1\}$.
The utilities of the goods are as follows.
\begin{itemize}
    \item $u_1(g_1) = b-1$ and $u_1(g) = b$ for $g \in M'$.
    \item For $i \in N \setminus \{1\}$, $u_i(g_1) = 0$ and $u_i(g) = a$ for $g \in M'$.
\end{itemize}
This instance is an integer-valued instance; moreover, it is positive-admitting since an allocation where every agent has $k$ goods from $M'$ gives positive utility to every agent.

Let $\mathcal{A} = (A_1, \ldots, A_n)$ be an allocation chosen by the additive welfarist rule with $f$.
We must have $g_1 \in A_1$.
Indeed, if $g_1$ is in $A_i$ for some $i \in N \setminus \{1\}$, then transferring $g_1$ to agent $1$ increases the value of $f(u_1(A_1))$ and does not decrease $f(u_i(A_i))$, thereby increasing $\sum_{j \in N} f(u_j(A_j))$ and contradicting the assumption that $\mathcal{A}$ is chosen by the additive welfarist rule with $f$.

Recall that $\mathcal{A}$ is EF1.
We claim that every agent receives exactly $k$ goods from $M'$.
If not, then some agent receives fewer than $k$ goods from $M'$, and envies another agent who receives more than $k$ goods from $M'$ by more than one good, making $\mathcal{A}$ not EF1.
Therefore, we have $\sum_{j \in N} f(u_j(A_j)) = f((k+1)b-1) + (n-1)f(ka)$.

Let $\mathcal{B} = (B_1, \ldots, B_n)$ be the same allocation as $\mathcal{A}$ except that one good from $M'$ is transferred from agent $1$'s bundle to agent $2$'s bundle.
We have $\sum_{j \in N} f(u_j(B_j)) = f(kb-1) + f((k+1)a) + (n-2)f(ka)$.
Since $\mathcal{B}$ is not EF1, $\mathcal{B}$ cannot be chosen by the additive welfarist rule with $f$.
Therefore, $\sum_{j \in N} f(u_j(A_j)) > \sum_{j \in N} f(u_j(B_j))$.
Rearranging the terms, we get $f((k+1)b-1) - f(kb-1) > \Delta_{f, k}(a)$.
Since $k, a, b$ were arbitrarily chosen, $f((k+1)b-1) - f(kb-1) > \Delta_{f, k}(a)$ holds for all $k, a, b \in \Zplus$.
Therefore, $f$ satisfies {\condintgeneral}a.
\end{proof}

\begin{restatable}{proposition}{propintegergeneralsufficient}
\label{prop:integer_general_sufficient}
Let $n \geq 2$ be given, and let $f: \Rpluszero \to \R \cup \{-\infty\}$ be a strictly increasing function that satisfies {\condintgeneral}b.
Then, for every positive-admitting integer-valued instance with $n$ agents, every allocation chosen by the additive welfarist rule with $f$ is EF1.
\end{restatable}

\begin{proof}
Let a positive-admitting integer-valued instance with $n$ agents be given, and let $\mathcal{A} = (A_1, \ldots, A_n)$ be an allocation chosen by the additive welfarist rule with $f$.
Assume, for the sake of contradiction, that $\mathcal{A}$ is not EF1.
Then, there exist distinct $i, j \in N$ such that agent $i$ envies agent $j$ by more than one good, i.e., $u_i(A_i) < u_i(A_j \setminus \{g\})$ for all $g \in A_j$.
Since the instance is integer-valued, we have $u_i(A_i) \leq u_i(A_j \setminus \{g\}) - 1$ for all $g \in A_j$, or equivalently, $u_i(A_i \cup \{g\}) + 1 \leq u_i(A_j)$.

Since $u_i(A_j) > u_i(A_i) \geq 0$, there exists a good in $A_j$ with positive utility to agent $i$.
Let $g' = \argmin_{g \in A_j,\, u_i(g) > 0} u_j(g) / u_i(g)$.
Note that $u_j(g') > 0$.
Indeed, if $u_j(g') = 0$, then transferring $g'$ from agent $j$'s bundle to agent $i$'s bundle increases $f(u_i(A_i))$ and does not decrease $f(u_j(A_j))$, thereby increasing $\sum_{k \in N} f(u_k(A_k))$ and contradicting the assumption that $\mathcal{A}$ is chosen by the additive welfarist rule with $f$.

Let $a = u_j(g')$, $b = u_i(g')$, $x = u_j(A_j \setminus \{g'\})$, and $y = u_i(A_i)$.
Note that $a, b \in \Zplus$ and $x, y \in \Zpluszero$, and that $\sum_{k \in N} f(u_k(A_k)) = f(y) + f(x+a) + \sum_{k \in N \setminus \{i, j\}} f(u_k(A_k))$.
By the choice of $g'$, we have
\begin{align*}
    \frac{u_j(g')}{u_i(g')} 
    \leq \frac{\sum_{g \in A_j,\, u_i(g) > 0} u_j(g)}{\sum_{g \in A_j,\, u_i(g) > 0} u_i(g)}
    \leq \frac{u_j(A_j)}{u_i(A_j)}
    \leq \frac{u_j(A_j)}{u_i(A_i \cup \{g'\}) + 1};
\end{align*}
this implies that
\begin{align*}
    \frac{x+a}{a} = \frac{u_j(A_j)}{u_j(g')} \geq \frac{u_i(A_i \cup \{g'\}) + 1}{u_i(g')} = \frac{y+b+1}{b},
\end{align*}
or $x/a \geq (y+1)/b$.
By {\condintgeneral}b, we have $f(y+b) - f(y) > f(x+a) - f(x)$.

Let $\mathcal{B} = (B_1, \ldots, B_n)$ be the same allocation as $\mathcal{A}$ except that $g'$ is transferred from agent $j$'s bundle to agent $i$'s bundle.
Then, $\sum_{k \in N} f(u_k(B_k)) = f(y+b) + f(x) + \sum_{k \in N \setminus \{i, j\}} f(u_k(A_k))$.
Since $\mathcal{A}$ is chosen by the additive welfarist rule with $f$, we must have $\sum_{k \in N} f(u_k(B_k)) \leq \sum_{k \in N} f(u_k(A_k))$.
Rearranging the terms, we get $f(y+b) - f(y) \leq f(x+a) - f(x)$, which contradicts {\condintgeneral}b.
Therefore, $\mathcal{A}$ is EF1.
\end{proof}

We demonstrate that the necessary condition given as {\condintgeneral}a is fairly tight, in the sense that it is stronger than {\condinttwovalue}, which characterizes the functions defining the additive welfarist rules that guarantee EF1 for integer-valued two-value instances.

\begin{restatable}{proposition}{propintegergeneralimpliestwovalue}
\label{prop:integer_general_implies_twovalue}
Let $f: \Rpluszero \to \R \cup \{-\infty\}$ be a strictly increasing function.
If $f$ satisfies {\condintgeneral}a, then it satisfies {\condinttwovalue}.
\end{restatable}

\begin{proof}
When $a = b = 1$ in {\condintgeneral}a, we have $\Delta_{f, k-1}(1) = f(k) - f(k-1) > \Delta_{f, k}(1)$ for all $k \in \Zplus$, which is {\condintbinary}.
This implies that $f$ satisfies {\condintbinary}.

Let $a, b \in \Zplus$ and $k, \ell, r \in \Zpluszero$ be given such that $a \geq b$ and $(k+1)b > \ell b + ra$.
We have
\begin{align*}
    &f((\ell+1)b+ra) - f(\ell b+ra) \\
    &= \sum_{t=0}^{b-1} \Delta_{f, \ell b+ra+t}(1) \tag{by telescoping sum} \\
    &\geq \sum_{t=0}^{b-1} \Delta_{f, (k+1)b-1+t}(1) \tag{by {\condintbinary} and \protect{$(k+1)b - 1 \geq \ell b + ra$}} \\
    &= f((k+2)b-1) - f((k+1)b-1) \tag{by telescoping sum} \\
    &> \Delta_{f, k+1}(a), \tag{by {\condintgeneral}a}
\end{align*}
which establishes the left inequality of {\condinttwovalue}.
Now,
\begin{align*}
    \Delta_{f, k+1}(a) &= \sum_{t=0}^{a-1} \Delta_{f, (k+1)a+t}(1) \tag{by telescoping sum} \\
    &\geq \sum_{t=0}^{a-1} \Delta_{f, (k+2)a-1+t}(1) \tag{by {\condintbinary} and \protect{$(k+2)a-1 \geq (k+1)a$}} \\
    &= f((k+3)a-1) - f((k+2)a-1) \tag{by telescoping sum} \\
    &> \Delta_{f, k+2}(1), \tag{by {\condintgeneral}a}
\end{align*}
establishing the right inequality of {\condinttwovalue}.
\end{proof}

We now present some examples.
We show that just like for the subclasses of integer-valued identical-good instances (\Cref{prop:integer_idengood_log}), binary instances (\Cref{prop:integer_idengood_implies_binary}), and integer-valued two-value instances (\Cref{prop:integer_twovalue_log}), the additive welfarist rule with the modified logarithmic function $\lambda_c$ for $0 \leq c \leq 1$ guarantees EF1 for all integer-valued instances.

\begin{restatable}{proposition}{propintegergenerallog}
\label{prop:integer_general_log}
The function $\lambda_c$ satisfies {\condintgeneral}b if and only if $0 \leq c \leq 1$.
\end{restatable}

\begin{proof}
($\Leftarrow$) We show that $\lambda_c$ satisfies {\condintgeneral}b for $0 \leq c \leq 1$.
Let $a, b \in \Zplus$ and $x, y \in \Zpluszero$ such that $x/a \geq (y+1)/b$.
Note that $x \geq a(y+1)/b > 0$.

If $y = c = 0$, then $\lambda_c(y) = -\infty$, and hence $\lambda_c(y+b) - \lambda_c(y) > \lambda_c(x+a) - \lambda_c(x)$ holds since the left-hand side is $\infty$ while the right-hand side is finite, and hence {\condintgeneral}b is satisfied.
Therefore, we shall consider the case $y+c > 0$.

Now, we have 
\begin{align*}
    \lambda_c(y+b) - \lambda_c(y) = \log \left(\frac{y+b+c}{y+c}\right) = \log \left( 1 + \frac{b}{y+c} \right)
\end{align*}
and
\begin{align*}
    \lambda_c(x+a) - \lambda_c(x) = \log \left(\frac{x+a+c}{x+c}\right) = \log \left( 1 + \frac{a}{x+c} \right).
\end{align*}
On the other hand, we have
\begin{align*}
    \frac{x+c}{a} \geq \frac{x}{a} \geq \frac{y+1}{b} \geq \frac{y+c}{b},
\end{align*}
where the first inequality is strict if $c > 0$, and the last inequality is strict if $c < 1$; therefore, it holds that $(x+c)/a > (y+c)/b$.
This implies that $b/(y+c) > a/(x+c)$.
Since $\log$ is an increasing function, we have $\lambda_c(y+b) - \lambda_c(y) > \lambda_c(x+a) - \lambda_c(x)$.
Since $a, b, x, y$ were arbitrarily chosen, $\lambda_c(y+b) - \lambda_c(y) > \lambda_c(x+a) - \lambda_c(x)$ holds for all $a, b \in \Zplus$ and $x, y \in \Zpluszero$ such that $x/a \geq (y+1)/b$.
Therefore, $\lambda_c$ satisfies {\condintgeneral}b.

($\Rightarrow$) We show that $\lambda_c$ does not satisfy {\condintgeneral}b for $c > 1$.
By \Cref{prop:integer_idengood_log,prop:condition_three_eqv}, $\lambda_c$ does not satisfy {\condintidengood}, and by \Cref{thm:integer_idengood}, there exists a positive-admitting integer-valued (identical-good) instance with $n \geq 2$ agents such that not every allocation chosen by the welfarist rule with $f$ is EF1.
Then, by the contraposition of \Cref{prop:integer_general_sufficient}, $\lambda_c$ does not satisfy {\condintgeneral}b.
\end{proof}

On the other hand, leveraging {\condintgeneral}a, we show that the additive welfarist rule with the modified harmonic function $h_c$ for $-1 \leq c < -1/2$ does \emph{not} guarantee EF1 for the class of all integer-valued instances.
This stands in contrast to integer-valued identical-good instances (\Cref{prop:integer_idengood_harmonic}), binary instances (\Cref{prop:integer_idengood_implies_binary}), and integer-valued two-value instances (\Cref{prop:integer_twovalue_harmonic}).

\begin{restatable}{proposition}{propintegergeneralharmonic}
\label{prop:integer_general_harmonic}
The function $h_c$ does not satisfy {\condintgeneral}a when $-1 \leq c < -1/2$.
\end{restatable}

To prove \Cref{prop:integer_general_harmonic}, we make use of the following lemma.

\begin{lemma}
\label{lem:log_limit_negative_half}
Let $f(x) = \frac{1}{\log((x+1)/x)}-(x+1)$.
Then, \[ \lim_{x \to \infty} f(x) = -1/2. \]
\end{lemma}

\begin{proof}
By Taylor's theorem, for all large enough $x$,
\[
    \frac{1}{x} - \frac{1}{2x^2} \; \leq \; \log\left(1 + \frac{1}{x}\right) \; \leq \; \frac{1}{x} - \frac{1}{2x^2} + \frac{1}{3x^3},
\]
and so
\[
    \frac{1}{\frac{1}{x} - \frac{1}{2x^2}} - (x+1) \; \geq \; f(x) \; \geq \; \frac{1}{\frac{1}{x} - \frac{1}{2x^2} + \frac{1}{3x^3}} - (x+1).
\]
Now, we have
\begin{align*}
    \lim_{x\to\infty} \left( \frac{1}{\frac{1}{x} - \frac{1}{2x^2}} - (x+1) \right)
    = \lim_{x\to\infty} \left( \frac{- 1 + \frac{1}{x}}{2 - \frac{1}{x}} \right)
    = -\frac{1}{2}
\end{align*}
and
\begin{align*}
    \lim_{x\to\infty} \left( \frac{1}{\frac{1}{x} - \frac{1}{2x^2} + \frac{1}{3x^3}} - (x+1) \right) 
    = \lim_{x\to\infty} \left( \frac{- 3 + \frac{1}{x} - \frac{2}{x^2}}{6 - \frac{3}{x} + \frac{2}{x^2}} \right)
    = -\frac{1}{2}.
\end{align*}
By the squeeze theorem, we get the required limit.
\end{proof}

\begin{proof}[Proof of \Cref{prop:integer_general_harmonic}]
Let $-1 \leq c < -1/2$, and let $f(x) = \frac{1}{\log((x+1)/x)}-(x+1)$.
By \Cref{lem:log_increasing,lem:log_limit_negative_half}, $f$ is strictly increasing on the integer domain and its limit when $x \to \infty$ is equal to $-1/2$, so there exists $k \in \Zplus$ such that $f(k) > c$.
This implies that $1/(k+1+c) > \log ((k+1)/k)$.
Since $\Delta_{h_c, k}(1) = h_c(k+1) - h_c(k) = 1/(k+1+c)$, it holds that $\Delta_{h_c, k}(1) > \log ((k+1)/k)$.

On the other hand, we have
\begin{align*}
    &\left. \lim_{b \to \infty} \middle( h_c((k+1)b-1) - h_c(kb-1) \right) \\
    &= \left. \lim_{b \to \infty} \middle( \log((k+1)b+c) - \log(kb+c) \right) \tag{by \Cref{lem:limit_harmonic_log_difference}} \\
    &= \log \left(\frac{k+1}{k}\right),
\end{align*}
where the last equality follows from the proof of \Cref{lem:limit_harmonic_scaled_difference}.

Then, by definition of limit, there exists $b \in \Zplus$ such that
\begin{align*}
    h_c((k+1)b-1) - h_c(kb-1) < \Delta_{h_c, k}(1).
\end{align*}
Thus, {\condintgeneral}a is not satisfied.
\end{proof}

\section{Conclusion}
\label{sec:conclusion}

In this paper, we have characterized additive welfarist rules that guarantee EF1 allocations for various classes of instances.
In the real-valued case, we strengthened the result of \citet{Suksompong23} by showing that only the maximum Nash welfare (MNW) rule ensures EF1 even for the most restricted class of identical-good instances.
On the other hand, in the practically important case where all values are integers, we demonstrated that even for the most general setting with no additional restrictions, there is a wide range of additive welfarist rules that always return EF1 allocations.

Since we have established the existence of several alternatives that perform as well as MNW in terms of guaranteeing EF1 allocations for integer-valued instances, a natural follow-up direction would be to compare these rules with respect to other measures, both to MNW and to one another.
As illustrated in \Cref{ex:mnw}, some additive welfarist rules may produce more preferable allocations for certain instances than others.
It would be interesting to formalize this observation, for example, by comparing their prices in terms of social welfare, in the same vein as the ``price of fairness'' \citep{CaragiannisKaKa12,BeiLuMa21,CelineDzKo23,LiLiLu24}.

Other possible future directions include completely characterizing additive welfarist rules for normalized instances with two agents (\Cref{subsec:real_normalized}) and for general integer-valued instances (\Cref{subsec:integer_general}) by tightening their respective necessary and sufficient conditions.
Obtaining characterizations of welfarist rules that are not necessarily additive is also a possible extension of our work.
Beyond EF1, one could consider other fairness notions from the literature, such as proportionality up to one good (PROP1).
Finally, another intriguing avenue is to study the problem of computing (or approximating) an allocation produced by various additive welfarist rules, as has been done extensively for the MNW rule \citep{Lee17,ColeGk18,AmanatidisBiFi21,AkramiChHo22}.

%%%%%%%%%%%%%%%%%%%%%%%

\subsection*{Acknowledgments}

This work was partially supported by the Singapore Ministry of Education under grant number MOE-T2EP20221-0001 and by an NUS Start-up Grant.
A preliminary version of the paper appeared in Proceedings of the 24th International Conference on Autonomous Agents and Multiagent Systems, May 2025.
We would like to thank the anonymous reviewers of AAMAS 2025 and ACM Transactions on Economics and Computation for their valuable comments, as well as Sanjay Jain and Frank Stephan for their help with the proof of \Cref{lem:log_increasing}.

% Bibliography
\bibliographystyle{ACM-Reference-Format}
\bibliography{main}

\appendix

\section{Omitted Proofs}
\label{ap:proofs}

\subsection{Proof of \Cref{prop:integer_idengood_harmonic}}
\label{ap:integer_idengood_harmonic}

We shall prove \Cref{prop:integer_idengood_harmonic} via a series of intermediate results.

\begin{lemma}
\label{lem:log_increasing}
Let $x, y \in \Zplus$ be such that $x < y$.
Then, 
\begin{align}
\label{eq:log_increasing}
    \frac{1}{\log((x+1)/x)}-x < \frac{1}{\log((y+1)/y)}-y.
\end{align}
\end{lemma}

\begin{proof}
Let $f(z) = \frac{1}{\log((z+1)/z)}-z$ be defined on $z \in \Rplus$.
We have
\begin{align*}
    f'(z) = \frac{1}{z(z+1)\left( \log \left(\frac{z+1}{z}\right) \right)^2} - 1.
\end{align*}
We first show that $f'(z) > 0$ for all $z \geq 7$.

Let $z \geq 7$. By Taylor's theorem, we have
\begin{align*}
    \log\left(1 + \frac{1}{z}\right) &= \sum_{t=1}^\infty {\frac{(-1)^{t+1}}{tz^{t}}} \\
    &= \frac{1}{z} - \frac{1}{2z^2} + \frac{1}{3z^3} - \sum_{t=2}^\infty \left(\frac{1}{2tz^{2t}} - \frac{1}{(2t+1)z^{2t+1}}\right) \\
    &< \frac{1}{z} - \frac{1}{2z^2} + \frac{1}{3z^3} \tag{since $z \geq 1$}.
\end{align*}
Then,
\begin{align*}
    f'(z) &> \frac{1}{(z^2 + z)\left(\frac{1}{z} - \frac{1}{2z^2} + \frac{1}{3z^3}\right)^2} - 1 
    = \frac{36z^5}{(z + 1)\left(6z^2 - 3z + 2\right)^2} - 1.
\end{align*}
Expanding the denominator of the above fraction gives
\[
    (z + 1)\left(6z^2 - 3z + 2\right)^2 = 36z^5 - 3z^3 + 21z^2 - 8z + 4.
\]
It can be shown that the cubic equation $-3x^3 + 21x^2 - 8x + 4 = 0$ has exactly one real root at $x \approx 6.62802$.
Moreover, for all $x > 6.62802$, we have $-3x^3 + 21x^2 - 8x + 4 < 0$, and hence $(x + 1)\left(6x^2 - 3x + 2\right)^2 < 36x^5$.
Since $z \geq 7$, it follows that
\begin{align*}
    f'(z)
    > \frac{36z^5}{(z + 1)\left(6z^2 - 3z + 2\right)^2} - 1
    > \frac{36z^5}{36z^5} - 1
    > 0.
\end{align*}

Now, since $f$ is differentiable on $\Rplus$, we have shown that \eqref{eq:log_increasing} holds for all $x, y \in \Zplus$ with $7 \leq x < y$.
For the remaining cases, one can verify that $f(1) < f(2) < \cdots < f(7)$.
Hence, \eqref{eq:log_increasing} holds for all $x, y \in \Zplus$ with $x < y$.
\end{proof}

The limits in the following three lemmas (\ref{lem:euler_mascheroni_constant}--\ref{lem:limit_harmonic_scaled_difference}) are taken over only integers (as opposed to real numbers).

\begin{lemma}[\citet{Euler1740,Mascheroni1790}]
\label{lem:euler_mascheroni_constant}
Let $\gamma \approx 0.57721$ be the Euler--Mascheroni constant.
Then, $\lim_{x \to \infty} (h_0(x) - \log x) = \gamma$.
\end{lemma}

\begin{lemma}
\label{lem:limit_harmonic_log_difference}
For any $c \geq -1$, $\lim_{x \to \infty} (h_c(x) - \log(x+c+1))$ exists and is finite.
\end{lemma}

\begin{proof}
We first consider the case $c = -1$. Observe that $h_{-1}(x) = h_0(x-1)$ for any integer $x > 0$. Furthermore, we have 
\begin{align*}
    \lim_{x \to \infty} (\log x - \log(x-1)) &= \lim_{x \to \infty} \log(\frac{x}{x-1}) 
    = \lim_{x \to \infty} \log\left(1 + \frac{1}{x-1}\right) 
    = 0.
\end{align*}
It follows that
\begin{align*}
    \lim_{x \to \infty} (h_{-1}(x) - \log x) &= \lim_{x \to \infty} (h_{-1}(x) - \log(x-1)) \\
    &= \lim_{x \to \infty} (h_0(x-1) - \log(x-1)) \\
    &= \gamma, \tag{by \Cref{lem:euler_mascheroni_constant}}
\end{align*}
where $\gamma$ is the Euler--Mascheroni constant, thus proving the result for $c = -1$.

For $c > -1$, we have, for any integer $x > 2$,
\begin{align*}
    h_c(x) &= \sum_{k=1}^x \frac{1}{k+c} 
    = \sum_{k=1}^x \int_k^{k+1} \frac{1}{\lfloor t\rfloor+c} \, \dd t 
    = \int_1^{x+1} \frac{1}{\lfloor t\rfloor+c} \, \dd t
\end{align*}
and 
\begin{align*}
    \log(x+c+1) &= \int_1^{x+c+1} \frac{1}{s} \, \dd s \\
    &= \int_{1-c}^{x+1} \frac{1}{t+c} \, \dd t \tag{by substituting $s = t+c$} \\
    &= \int_{1}^{x+1} \frac{1}{t+c} \, \dd t + \int_{1-c}^{1} \frac{1}{t+c} \, \dd t.
\end{align*}
Hence,
\begin{align*}
    h_c(x) - \log(x+c+1) 
    &= \int_{1}^{x+1} \left(\frac{1}{\lfloor t\rfloor+c} - \frac{1}{t+c}\right) \, \dd t - \int_{1-c}^{1} \frac{1}{t+c} \, \dd t \\
    &= \int_{2}^{x+1} \left(\frac{1}{\lfloor t\rfloor+c} - \frac{1}{t+c}\right) \, \dd t
    + \int_{1}^{2} \frac{1}{\lfloor t\rfloor+c} \, \dd t 
    - \int_{1-c}^{2} \frac{1}{t+c} \, \dd t.
\end{align*}
Then, we have
\begin{align*}
    \lim_{x \to \infty} (h_c(x) - \log(x+c+1)) 
    &= \int_{2}^{\infty} \left(\frac{1}{\lfloor t\rfloor+c} - \frac{1}{t+c}\right) \, \dd t + D_1(c) - D_2(c), 
\end{align*}
where $D_1(c) = \int_{1}^{2} \frac{1}{\lfloor t\rfloor+c} \, \dd t$ and $D_2(c) = \int_{1-c}^{2} \frac{1}{t+c} \, \dd t$ are finite.
Therefore, it suffices to prove that $\int_{2}^{\infty} \left(\frac{1}{\lfloor t\rfloor+c} - \frac{1}{t+c}\right) \, \dd t$ is finite.
When $c = 0$, we have
\begin{align*}
    \int_{2}^{\infty} \left(\frac{1}{\lfloor t\rfloor} - \frac{1}{t}\right) \, \dd t + D_1(0) - D_2(0)
    &= \lim_{x \to \infty} (h_0(x) - \log(x+1)) \\
    &= \lim_{x \to \infty} (h_{-1}(x+1) - \log(x+1)) \\
    &= \lim_{y \to \infty} (h_{-1}(y) - \log y) \\
    &= \gamma
\end{align*}
as shown in the case $c = -1$; hence, $\int_{2}^{\infty} \left(\frac{1}{\lfloor t\rfloor} - \frac{1}{t}\right) \, \dd t$ is finite.
Now, for $t \geq 2$, we have $\lfloor t \rfloor + c > \lfloor t \rfloor / 2$ and $t+c > t/2$, and so
\begin{align*}
    \frac{1}{\lfloor t\rfloor+c} - \frac{1}{t+c} = \frac{t - \lfloor t \rfloor}{(\lfloor t\rfloor+c)(t+c)}
    \le \frac{t - \lfloor t \rfloor}{\frac{\lfloor t\rfloor}{2} \cdot \frac{t}{2}}
    = 4 \left(\frac{1}{\lfloor t\rfloor} - \frac{1}{t}\right).
\end{align*}
Since $\int_{2}^{\infty} \left(\frac{1}{\lfloor t\rfloor} - \frac{1}{t}\right) \, \dd t$ is finite, we get that $\int_{2}^{\infty} \left(\frac{1}{\lfloor t\rfloor+c} - \frac{1}{t+c}\right) \, \dd t$ is finite by the comparison test for integrals.
This shows that $\lim_{x \to \infty} (h_c(x) - \log(x+c+1))$ exists for $c > -1$.
\end{proof}

Using \Cref{lem:limit_harmonic_log_difference}, we can now prove the following result.

\begin{lemma}
\label{lem:limit_harmonic_scaled_difference}
Let $k \in \Zpluszero$ and $c \geq -1$. Then,
\[
    \lim_{x \to \infty} \Delta_{h_c, k}(x) = \lim_{x \to \infty} \Delta_{\lambda_{c+1},k}(x) = \log\left(\frac{k+1}{k}\right).
\]
\end{lemma}

\begin{proof}
For the right equality, we have
\begin{align*}
    \lim_{x \to \infty} \Delta_{\lambda_{c+1},k}(x)
    &= \lim_{x \to \infty} \left(\lambda_{c+1}((k+1)x) - \lambda_{c+1}(kx)\right) \\
    &= \lim_{x \to \infty} \left(\log((k+1)x+c+1) - \log(kx+c+1)\right) \\
    &= \lim_{x \to \infty} \log\left(\frac{(k+1)x+c+1}{kx+c+1}\right) 
    = \lim_{x \to \infty} \log\left(\frac{k+1+\frac{c+1}{x}}{k+\frac{c+1}{x}}\right) 
    = \log\left(\frac{k+1}{k}\right).
\end{align*}
For the left equality, we have
\begin{align*}
    & \left. \lim_{x \to \infty} \middle(\Delta_{h_c, k}(x) - \Delta_{\lambda_{c+1},k}(x)\right) \\
    &= \left. \lim_{x \to \infty} \middle( (h_c((k+1)x) - h_c(kx) ) 
    -  (\lambda_{c+1}((k+1)x) - \lambda_{c+1}(kx) ) \right) \\
    &= \left. \lim_{x \to \infty} \middle( h_c((k+1)x) - \lambda_{c+1}((k+1)x) \right) 
    - \left. \lim_{x \to \infty} \middle(h_c(kx) - \lambda_{c+1}(kx)\right)
    = 0,
\end{align*}
where the last transition follows from \Cref{lem:limit_harmonic_log_difference}.
\end{proof}

\begin{lemma}
\label{lem:approx_sum_with_int}
Let $a, b \in \Zplus$ with $a \leq b$. Then, 
\[
    \sum_{k=a}^b \frac{1}{k} < \int_{a-1/2}^{b+1/2} \frac{1}{x} \; \dd x.
\]
\end{lemma}

\begin{proof}
We have
\begin{align*}
    \int_{a-\frac{1}{2}}^{b+\frac{1}{2}} \frac{1}{x} \, \dd x - \sum_{k=a}^b \frac{1}{k} 
    &= \sum_{k=a}^b \left(\int_{k-\frac{1}{2}}^{k+\frac{1}{2}} \left(\frac{1}{x} - \frac{1}{k}\right) \, \dd x \right) \\
    &= \sum_{k=a}^b \left(\int_{k-\frac{1}{2}}^{k} \left(\frac{1}{x} - \frac{1}{k}\right) \, \dd x + \int_{k}^{k+\frac{1}{2}} \left(\frac{1}{x} - \frac{1}{k}\right) \, \dd x \right) \\
    &= \sum_{k=a}^b \left(\int_{k-\frac{1}{2}}^{k} \left(\frac{1}{x} - \frac{1}{k}\right) \, \dd x - \int_{k}^{k-\frac{1}{2}} \left(\frac{1}{2k-x} - \frac{1}{k}\right) \, \dd x \right) \tag{by the substitution \protect{$x \leftarrow 2k-x$}} \\
    &= \sum_{k=a}^b \left(\int_{k-\frac{1}{2}}^{k} \left(\frac{1}{x} - \frac{1}{k}\right) \, \dd x + \int_{k-\frac{1}{2}}^{k} \left(\frac{1}{2k-x} - \frac{1}{k}\right) \, \dd x \right) \\
    &= \sum_{k=a}^b \left(\int_{k-\frac{1}{2}}^{k} \left(\frac{1}{x} + \frac{1}{2k-x} - \frac{2}{k}\right) \, \dd x \right) \\
    &> \sum_{k=a}^b \left(\int_{k-\frac{1}{2}}^{k} \left(\frac{4}{2k} - \frac{2}{k}\right) \, \dd x \right) \\
    &= 0,
\end{align*}
where the inequality follows from the inequality of arithmetic and harmonic means, and this inequality is strict for a positive interval of $x$.
\end{proof}

\begin{lemma}
\label{lem:mod_harmonic_scaled_diff_increasing}
Let $c \geq 0$, and let $a, b, k \in \Zplus$ such that $a < b$.
Then, $\Delta_{h_c, k}(a) < \Delta_{h_c, k}(b)$.
\end{lemma}

\begin{proof}
Note that
\begin{equation}\label{eq:diff_scaled_by_b_as_sum}
    \Delta_{h_c, k}(b) = h_c((k+1)b) - h_c(kb) = \sum_{i=1}^b \frac{1}{kb+i+c}
\end{equation}
and 
\begin{equation}\label{eq:diff_scaled_by_a_as_sum}
    \Delta_{h_c, k}(a) = h_c((k+1)a) - h_c(ka) = \sum_{i=1}^a \frac{1}{ka+i+c}.
\end{equation}
Since the two sums have different numbers of terms,
we cannot compare each term in the sum in~\eqref{eq:diff_scaled_by_b_as_sum} 
with the corresponding term in the sum in \eqref{eq:diff_scaled_by_a_as_sum}.
Instead, we approximate the sum in \eqref{eq:diff_scaled_by_a_as_sum} using an integral. By \Cref{lem:approx_sum_with_int},
\begin{equation}\label{eq:diff_scaled_by_a_as_integral}
    \sum_{i=1}^a \frac{1}{ka+i+c} < \int_{\frac{1}{2}}^{\frac{1}{2}+a} \frac{1}{ka+x+c} \, \dd x. 
\end{equation}
We can now write this integral as a sum of $b$ terms, by dividing the interval into $b$ equal parts:
\begin{equation}\label{eq:diff_scaled_by_a_as_sum_of_b_terms}
    \int_{\frac{1}{2}}^{\frac{1}{2}+a} \frac{1}{ka+x+c} \, \dd x
    = \sum_{i=1}^b 
        \left(
            \int_{\frac{1}{2}+\frac{(i-1)a}{b}}^{\frac{1}{2}+\frac{ia}{b}} \frac{1}{ka+x+c} \, \dd x
        \right).
\end{equation}
Now each of the sums in \eqref{eq:diff_scaled_by_b_as_sum} and \eqref{eq:diff_scaled_by_a_as_sum_of_b_terms} has exactly $b$ terms,
and we can compare each term with the corresponding term in the other sum.
We show that for each $i \in \{1, \ldots, b\}$,
\begin{equation}\label{eq:compare_each_summed_term}
    \frac{1}{kb+i+c} \geq \int_{\frac{1}{2}+\frac{(i-1)a}{b}}^{\frac{1}{2}+\frac{ia}{b}} \frac{1}{ka+x+c} \, \dd x.
\end{equation}
To compare the two terms in \eqref{eq:compare_each_summed_term}, we first write the left-hand side as an integral with the same interval as the right-hand side:
\begin{equation}\label{eq:diff_scaled_by_b_as_integral}
\begin{aligned}
    \frac{1}{kb+i+c} = \frac{a}{b} \cdot \frac{1}{ka+\frac{ia}{b}+\frac{ca}{b}} 
    = \int_{\frac{1}{2}+\frac{(i-1)a}{b}}^{\frac{1}{2}+\frac{ia}{b}}
        \frac{1}{ka+\frac{ia}{b}+\frac{ca}{b}} \, \dd x 
    \geq \int_{\frac{1}{2}+\frac{(i-1)a}{b}}^{\frac{1}{2}+\frac{ia}{b}}
        \frac{1}{ka+\frac{ia}{b}+c} \, \dd x,
\end{aligned}
\end{equation}
where the last inequality follows from the fact that $\inlinefrac{ca}{b} \leq c$ as $a < b$ and $c \geq 0$.
Hence, it suffices to show that
\begin{align*}
    \int_{\frac{1}{2}+\frac{(i-1)a}{b}}^{\frac{1}{2}+\frac{ia}{b}}
        \frac{1}{ka+\frac{ia}{b}+c} \, \dd x
    &\geq \int_{\frac{1}{2}+\frac{(i-1)a}{b}}^{\frac{1}{2}+\frac{ia}{b}} 
        \frac{1}{ka+x+c} \, \dd x,
\end{align*}
which by substituting $x = \frac{ia}{b} + y$ is equivalent to
\begin{align*}
    \int_{\frac{1}{2}-\frac{a}{b}}^{\frac{1}{2}} \frac{1}{ka+\frac{ia}{b}+c} \, \dd y
        &\geq \int_{\frac{1}{2}-\frac{a}{b}}^{\frac{1}{2}} \frac{1}{ka+\frac{ia}{b}+y+c} \, \dd y.
\end{align*}
Let $t = ka+\inlinefrac{ia}{b}+c$. Then, the above inequality is equivalent to 
\begin{equation}\label{eq:integral_diff}
\begin{aligned}
    \int_{\frac{1}{2}-\frac{a}{b}}^{\frac{1}{2}} \left(\frac{1}{t} - \frac{1}{t+y}\right) \, \dd y \geq 0.
\end{aligned}
\end{equation}

Note that if $\inlinefrac{a}{b} \leq \inlinefrac{1}{2}$,
then $\inlinefrac{1}{2} - \inlinefrac{a}{b} \geq 0$.
In this case, for all $y$ in the interval, we have $y \geq 0$ and so $\inlinefrac{1}{t} \geq \inlinefrac{1}{(t+y)}$. 
Subsequently, the integrand in \eqref{eq:integral_diff} is never negative,
and is positive when $y = \inlinefrac{1}{2}$.
As a result, inequality \eqref{eq:integral_diff} holds.

It remains to prove that \eqref{eq:integral_diff} holds when $\inlinefrac{a}{b} > \inlinefrac{1}{2}$, that is, $\inlinefrac{1}{2}-\inlinefrac{a}{b} < 0$.
We first divide the integral in \eqref{eq:integral_diff} into two parts, noting that $\inlinefrac{1}{t} - \inlinefrac{1}{(t+y)} > 0$ if and only if $y > 0$:
\begin{align}
    \int_{\frac{1}{2}-\frac{a}{b}}^{\frac{1}{2}} \left(\frac{1}{t} - \frac{1}{t+y}\right) \, \dd y 
        &= \int_{0}^{\frac{1}{2}} \left(\frac{1}{t} - \frac{1}{t+y}\right) \, \dd y
        + \int_{\frac{1}{2}-\frac{a}{b}}^{0} \left(\frac{1}{t} - \frac{1}{t+y}\right) \, \dd y. \label{eq:integral_pos_neg}
\end{align}
We then rewrite the second integral on the right-hand side so that it has the same interval as the first integral, 
using a series of variable substitutions:
\begin{align*}
    \int_{\frac{1}{2}-\frac{a}{b}}^0 \left(\frac{1}{t} - \frac{1}{t+y}\right) \, \dd y
    &= - \int_{\frac{a}{b}-\frac{1}{2}}^0 \left(\frac{1}{t} - \frac{1}{t-x}\right) \, \dd x 
        \tag{by substituting $y = -x$} \\
    &= \int_0^{\frac{a}{b}-\frac{1}{2}} \left(\frac{1}{t} - \frac{1}{t-x}\right) \, \dd x \\
    &= \int_0^{\frac{1}{2}} \left(\frac{2a}{b}-1\right) \left( 
        \frac{1}{t} - \frac{1}{t-\left(\frac{2a}{b}-1\right)z} \right) \, \dd z. 
        \tag{by substituting $x = \left(\frac{2a}{b}-1\right)z$}
\end{align*}
Substituting back into \eqref{eq:integral_pos_neg} and replacing the variable $z$ with $y$, 
we get
\begin{align*}
   \int_{\frac{1}{2}-\frac{a}{b}}^{\frac{1}{2}}& \left(\frac{1}{t} - \frac{1}{t+y}\right) \, \dd y 
   = \int_0^{\frac{1}{2}} \left(\left(\frac{1}{t} - \frac{1}{t+y}\right) + \left(\frac{2a}{b}-1\right) \left( 
        \frac{1}{t} - \frac{1}{t-\left(\frac{2a}{b}-1\right)y} \right)\right) \, \dd y.
\end{align*}
Next, we use algebraic manipulation to show that
the integrand is non-negative for all $y$ in the interval.
Writing $q = \frac{2a}{b} - 1$, the integrand becomes
\begin{align}
    \frac{1}{t} - \frac{1}{t+y} + q \left( \frac{1}{t} - \frac{1}{t-qy} \right) 
    &= \frac{y}{t(t+y)} - \frac{q^2 y}{t(t-qy)} \nonumber \\
    &= \frac{y(t-qy-q^2t-q^2y)}{t(t+y)(t-qy)} \nonumber \\
    &= \frac{y(t(1-q^2) - qy(1+q))}{t(t+y)(t-qy)} \nonumber \\
    &= \frac{y(t(1-q)(1+q) - qy(1+q))}{t(t+y)(t-qy)} \nonumber \\
    &= \frac{y(1+q)(t(1-q) - qy)}{t(t+y)(t-qy)}. \label{eq:integrand_pos}
\end{align}    
We now show that the expression in the last line is non-negative for all $y \in [0, \inlinefrac{1}{2}]$,
by showing that all the factors in the numerator are non-negative and all the factors in the denominator are positive.
Clearly, $y \geq 0$ and $1+q = \inlinefrac{2a}{b} \geq 0$. 
Moreover, since $t = ka + \inlinefrac{ia}{b} + c$, we have $t+y \geq t > 0$.
Hence, it remains to prove that $t(1-q) - qy \geq 0$ and $t-qy > 0$.

To show that $t(1-q) - qy \geq 0$ for all $y \in [0, \inlinefrac{1}{2}]$,
it suffices to show that $t \geq q/(2-2q) \geq qy/(1-q)$.
Note that $1-q = 2 - 2a/b > 0$ as we assume that $a < b$.
By algebraic manipulation, we can see that
\begin{align*}
    \frac{q}{2-2q} = \frac{\frac{2a}{b}-1}{4-\frac{4a}{b}} 
    = \frac{2a-b}{4b-4a} 
    &= \frac{1}{4}\left(\frac{a}{b - a} - 1\right) 
    < a 
    < t,
\end{align*}
where the first inequality holds because $b-a \ge 1$.
Hence, we have $t(1-q) - qy \geq 0$ for all $y \in [0, 1/2]$.
Furthermore, by the assumption that $a/b > 1/2$, we have $q = 2a/b - 1 > 0$, and so $1 - q < 1$. 
Then,
$t - qy > t(1-q) - qy \geq 0$.
As a result, the expression in \eqref{eq:integrand_pos} is non-negative for all $y \in [0, 1/2]$. 
It follows that the inequality \eqref{eq:integral_diff} also holds for the case where $a/b > 1/2$.

Putting everything together, we have
\begin{align*}
    \Delta_{h_c, k}(b) - \Delta_{h_c, k}(a) &= \sum_{i=1}^b \frac{1}{kb+i+c} - \sum_{i=1}^a \frac{1}{ka+i+c} 
        \tag{by \eqref{eq:diff_scaled_by_b_as_sum} and \eqref{eq:diff_scaled_by_a_as_sum}}\\
    &> \sum_{i=1}^b \frac{1}{kb+i+c}
        - \sum_{i=1}^b 
        \left(
            \int_{\frac{1}{2}+\frac{(i-1)a}{b}}^{\frac{1}{2}+\frac{ia}{b}} \frac{1}{ka+x+c} \, \dd x
        \right)
        \tag{by \eqref{eq:diff_scaled_by_a_as_integral} 
            and \eqref{eq:diff_scaled_by_a_as_sum_of_b_terms}} \\
    &= \sum_{i=1}^b \left(
        \frac{1}{kb+i+c}
        - \int_{\frac{1}{2}+\frac{(i-1)a}{b}}^{\frac{1}{2}+\frac{ia}{b}} 
            \frac{1}{ka+x+c} \, \dd x
        \right) \\
    &\geq \sum_{i=1}^b \left(
        \int_{\frac{1}{2}+\frac{(i-1)a}{b}}^{\frac{1}{2}+\frac{ia}{b}}
            \frac{1}{ka+\frac{ia}{b}+c} \, \dd x
        - \int_{\frac{1}{2}+\frac{(i-1)a}{b}}^{\frac{1}{2}+\frac{ia}{b}} 
                \frac{1}{ka+x+c} \, \dd x
        \right)
        \tag{by \eqref{eq:diff_scaled_by_b_as_integral}} \\
    &\geq 0, \tag{by \eqref{eq:integral_diff}}
\end{align*}
which yields the desired inequality.
\end{proof}

We are now ready to prove \Cref{prop:integer_idengood_harmonic}.

\propintegeridengoodharmonic*

\begin{proof}
($\Leftarrow$) We show that $h_c$ satisfies {\condintidengood}b for $-1 \leq c \leq 1/\log 2 - 1$.
We first show this for $-1 \leq c < 0$.
Let $k \in \Zpluszero$ and $a \in \Zplus$.
Consider the left inequality of {\condintidengood}b.
If $c = -1$ and $k = 0$, then $\Delta_{h_c, k}(1) = \infty > \Delta_{h_c, k+1}(a)$, so the left inequality holds.
Otherwise, $k + 1 + c > 0$.
Then, we have
\begin{align*}
    \Delta_{h_c, k+1}(a) &= h_c((k+2)a) - h_c((k+1)a) \\
    &= \sum_{t=1}^{a} \frac{1}{(k+1)a+t+c} \\
    &\leq \sum_{t=1}^{a} \frac{1}{(k+1)a+1+c} \\
    &= \frac{a}{(k+1)a+1+c} 
    = \frac{1}{(k+1)+(1+c)/a} 
    \le \frac{1}{k+1} 
    < \frac{1}{(k+1)+c} 
    = \Delta_{h_c, k}(1),
\end{align*}
proving the left inequality of {\condintidengood}b.
Now, consider the right inequality of {\condintidengood}b.
By a similar approach, we have
\begin{align*}
    \Delta_{h_c, k+1}(a) &= h_c((k+2)a) - h_c((k+1)a) \\
    &= \sum_{t=1}^{a} \frac{1}{(k+1)a+t+c} \\
    &\geq \sum_{t=1}^{a} \frac{1}{(k+1)a+a+c} \\
    &= \frac{a}{(k+2)a+c} 
    = \frac{1}{(k+2)+c/a} 
    > \frac{1}{k+2} 
    \ge \frac{1}{(k+2)+1+c} 
    = \Delta_{h_c, k+2}(1),
\end{align*}
proving the right inequality of {\condintidengood}b.

Next, we show that $h_c$ satisfies {\condintidengood}b for $0 \leq c \leq 1/\log 2 - 1$.
Let $k \in \Zpluszero$.
Consider the left inequality of {\condintidengood}b.
By \Cref{lem:log_increasing}, we have
\begin{align*}
    c \leq \frac{1}{\log 2}-1 \leq \frac{1}{\log \left(\frac{k+2}{k+1}\right)}-(k+1),
\end{align*}
and therefore,
\begin{align*}
    \Delta_{h_c, k}(1) = h_c(k+1)-h_c(k) = \frac{1}{k+1+c} \geq \log \left(\frac{k+2}{k+1}\right).
\end{align*}
On the other hand, from \Cref{lem:limit_harmonic_scaled_difference}, we have
\[
    \lim_{x \to \infty} \Delta_{h_c, k+1}(x) = \log \left(\frac{k+2}{k+1}\right).
\]
By \Cref{lem:mod_harmonic_scaled_diff_increasing}, $\Delta_{h_c, k+1}$ is strictly increasing on the domain of positive integers. So, we have 
\[
    \Delta_{h_c, k+1}(a) < \log \left(\frac{k+2}{k+1}\right) \leq \Delta_{h_c, k}(1)
\]
for all $a \in \Zplus$, proving the left inequality of {\condintidengood}b.
Now, consider the right inequality of {\condintidengood}b.
By \Cref{lem:mod_harmonic_scaled_diff_increasing}, $\Delta_{h_c, k+1}$ is increasing on the domain of positive integers. 
It follows that
\begin{align*}
    \Delta_{h_c, k+1}(a) \geq \Delta_{h_c, k+1}(1) = \frac{1}{k+2+c} > \frac{1}{k+3+c} = \Delta_{h_c, k+2}(1)
\end{align*}
for all $a \in \Zplus$, proving the right inequality of {\condintidengood}b.

($\Rightarrow$) We show that $h_c$ does not satisfy {\condintidengood}b for $c > 1/\log 2 - 1$.
To this end, we prove that $\Delta_{h_c, 0}(1) \leq \Delta_{h_c, 1}(a)$ for some $a \in \Zplus$.
The expression on the left-hand side of the inequality is $\Delta_{h_c, 0}(1) = h_c(1) - h_c(0) = 1/(1+c) < \log 2$ since $c > 1/\log 2 - 1$.
For the expression on the right-hand side of the inequality,
by \Cref{lem:limit_harmonic_scaled_difference}, we have $\lim_{x \to \infty} \Delta_{h_c, 1}(x) = \log 2$.
Since $\Delta_{h_c, 0}(1) < \log 2$, by definition of limit, there exists $a \in \Zplus$ such that $\Delta_{h_c, 1}(a) \geq \Delta_{h_c, 0}(1)$.
Hence, $h_c$ does not satisfy {\condintidengood}b.
\end{proof}

\subsection{Proof of \Cref{prop:integer_twovalue_harmonic}}
\label{ap:integer_twovalue_harmonic}

To prove \Cref{prop:integer_twovalue_harmonic}, we need the following intermediate result concerning $h_c$.

\begin{lemma}\label{lem:h_neg_scaled_diff_minus_1_decreasing}
Let $-1 \leq c < 0$, and let $a, b, k \in \Zplus$ such that $a > b$. Then, 
\begin{equation}\label{eq:h_neg_scaled_diff_minus_1_decreasing}
    h_c((k+1)b-1) - h_c(kb-1) > h_c((k+1)a-1) - h_c(ka-1).
\end{equation}
\end{lemma}

\begin{proof}
First, if $kb = 1$ and $c = -1$, then $h_c((k+1)b-1) - h_c(kb-1) = \infty > h_c((k+1)a-1) - h_c(ka-1)$.
We can now assume that $kb + c > 0$.
Similarly to the proof of \Cref{lem:mod_harmonic_scaled_diff_increasing}, we can approximate the right-hand side of \eqref{eq:h_neg_scaled_diff_minus_1_decreasing} as a sum of $b$ terms:
\begin{align*}
    h_c((k+1)a-1) - h_c(ka-1) &= \sum_{i=1}^a \frac{1}{ka+i-1+c} \\
    &< \int_{\frac{1}{2}}^{a+\frac{1}{2}} \frac{1}{ka+x-1+c} \, \dd x
        \tag{by \Cref{lem:approx_sum_with_int}} \\
    &= \sum_{i=1}^b 
        \int_{\frac{(i-1)a}{b}+\frac{1}{2}}^{\frac{ia}{b}+\frac{1}{2}}
            \frac{1}{ka+x-1+c} 
        \, \dd x \\
    &= \sum_{i=1}^b
        \int_{-\frac{a}{b}}^{0}
            \frac{1}{ka+\frac{ia}{b}+y-\frac{1}{2}+c}
        \, \dd y.
        \tag{by substituting \protect{$x = y+\frac{ia}{b}+\frac{1}{2}$}}
\end{align*}
Hence, we have
\begin{align*}
    &\Big( h_c((k+1)b-1) - h_c(kb-1) \Big) - \Big( h_c((k+1)a-1) - h_c(ka-1) \Big) \\
    &> \sum_{i=1}^b \frac{1}{kb+i-1+c} 
        - \sum_{i=1}^b
            \int_{-\frac{a}{b}}^{0}
                \frac{1}{ka+\frac{ia}{b}+y-\frac{1}{2}+c}
            \, \dd y \\
    &= \sum_{i=1}^b 
        \left(
            \frac{1}{kb+i-1+c} 
            - \int_{-\frac{a}{b}}^{0}
                \frac{1}{ka+\frac{ia}{b}+y-\frac{1}{2}+c}
            \, \dd y
        \right).
\end{align*}
We now show that for each $i \in \{1, \ldots, b\}$, the summand is non-negative, that is,
\begin{equation}\label{eq:h_neg_compare_summands}
    \frac{1}{kb+i-1+c} 
    \geq \int_{-\frac{a}{b}}^{0}
        \frac{1}{ka+\frac{ia}{b}+y-\frac{1}{2}+c}
    \, \dd y.
\end{equation}
We first split the integral on the right-hand side into two parts with intervals $[-\frac{a}{b}, -\frac{a}{b}+\frac{1}{2}]$ and $[-\frac{a}{b}+\frac{1}{2}, 0]$ respectively.
We then rewrite the second part so that both parts have the same interval of integration, by substituting
$
    y = \left(-\frac{2a}{b}+1\right)\left(v+\frac{a}{b}\right).
$
Let
$
    q(y) = \inlinefrac{1}{(ka+\frac{ia}{b}+y-\frac{1}{2}+c)}.
$
Then, the second part of the integral becomes
\begin{align}
    \int_{-\frac{a}{b}+\frac{1}{2}}^{0} q(y) \, \dd y 
    &= \int_{-\frac{a}{b}+\frac{1}{2}}^{0}
        \frac{1}{ka+\frac{ia}{b}+y-\frac{1}{2}+c} \, \dd y \notag \\
    &=  \int_{-\frac{a}{b}+\frac{1}{2}}^{-\frac{a}{b}} 
            \left(
                \frac{1}{ka+\frac{ia}{b}+\left(-\frac{2a}{b}+1\right)\left(v+\frac{a}{b}\right)-\frac{1}{2}+c}   
            \right)
            \left(-\frac{2a}{b}+1\right)
        \, \dd v \notag \\
    &=  \int_{-\frac{a}{b}}^{-\frac{a}{b}+\frac{1}{2}}
            \left(\frac{2a}{b}-1\right)
            \left(
                \frac{1}{ka+\frac{ia}{b}-\left(\frac{2a}{b}-1\right)\left(v+\frac{a}{b}\right)-\frac{1}{2}+c}
            \right)
        \, \dd v.
\label{eq:h_neg_substitute_integral}
\end{align}
Now we can combine back the two parts into one integral under the same interval of integration:
\begin{align*}
    \int_{-\frac{a}{b}}^{0} q(y) \, \dd y 
    &= \int_{-\frac{a}{b}}^{-\frac{a}{b}+\frac{1}{2}} q(y) \, \dd y
    + \int_{-\frac{a}{b}+\frac{1}{2}}^{0} q(y) \, \dd y \\
    &= \int_{-\frac{a}{b}}^{-\frac{a}{b}+\frac{1}{2}} 
            \frac{1}{ka+\frac{ia}{b}+y-\frac{1}{2}+c} \, \dd y \\
    &\qquad +\int_{-\frac{a}{b}}^{-\frac{a}{b}+\frac{1}{2}} 
            \left(\frac{2a}{b}-1\right)
            \cdot \frac{1}{
                ka+\frac{ia}{b}
                -\left(\frac{2a}{b}-1\right)\left(v+\frac{a}{b}\right)
                -\frac{1}{2}+c}
        \, \dd v
    \tag{by \eqref{eq:h_neg_substitute_integral}} \\
    &=\int_{-\frac{a}{b}}^{-\frac{a}{b}+\frac{1}{2}} 
            \Bigg(\frac{1}{x+\left(ka+\frac{ia}{b}-\frac{1}{2}+c\right)} 
        - \frac{1}{x+\frac{a}{b}-\frac{1}{\frac{2a}{b}-1}
            \left(ka+\frac{ia}{b}-\frac{1}{2}+c\right)} \Bigg)
            \, \dd x.
    \tag{by replacing variables \protect{$v, y$ with $x$}}
\end{align*}
Define
$
    f(x) = \frac{1}{x+r} - \frac{1}{x+s},
$
where
\[
    r = ka+\frac{ia}{b}-\frac{1}{2}+c, \qquad
    s = \frac{a}{b}-\frac{1}{\frac{2a}{b}-1}\left(ka+\frac{ia}{b}-\frac{1}{2}+c\right)
\]
are terms independent of $x$.
We can then rewrite inequality \eqref{eq:h_neg_compare_summands} that we are trying to prove:
\begin{align*}
    \int_{-\frac{a}{b}}^{-\frac{a}{b}+\frac{1}{2}} 
        f(x)
    \, \dd x
    &= \int_{-\frac{a}{b}}^{0} q(y) \, \dd y 
    \leq \frac{1}{kb+i-1+c} 
    = \int_{-\frac{a}{b}}^{-\frac{a}{b}+\frac{1}{2}} \frac{2}{kb+i-1+c} \, \dd x.
\end{align*}
Note that the first and last expressions are integrals with the same interval of integration $[-\frac{a}{b}, -\frac{a}{b}+\frac{1}{2}]$.
Hence, it suffices to compare the integrand in this common interval of integration. That is, it suffices to prove that for each $x \in [-\frac{a}{b}, -\frac{a}{b}+\frac{1}{2}]$, we have
\begin{equation}\label{eq:h_neg_compare_integrands}
    f(x) \leq \frac{2}{kb+i-1+c}.
\end{equation}
To prove this, we examine the second derivative
\[
    f''(x) = \frac{2}{(x+r)^3} - \frac{2}{(x+s)^3}.
\]
Note that $r \geq 2 + \frac{a}{b} - \frac{1}{2} - 1 > \frac{a}{b}$. Then, $s = \frac{a}{b} - \inlinefrac{r}{(\frac{2a}{b}-1)} < \frac{a}{b} - \inlinefrac{(\frac{a}{b})}{(\frac{2a}{b} - 1)} < \frac{a}{b} - \frac{1}{2}$. So, for any $x \in [-\frac{a}{b}, -\frac{a}{b}+\frac{1}{2}]$, we have $x + r > 0$ and $x + s < 0$. Hence, we have $(x+r)^3 > 0$ and $(x+s)^3 < 0$. Therefore, the second derivative $f''(x)$ is always positive for any $x \in [-\frac{a}{b}, -\frac{a}{b}+\frac{1}{2}]$.
Subsequently, the first derivative $f'$ is always increasing in the interval.
So, we can conclude that the maximum of $f$ in the interval $x \in [-\frac{a}{b}, -\frac{a}{b}+\frac{1}{2}]$ must be attained at one of the endpoints;
any local extremum in the interior of the interval must be a minimum.
We name the endpoints $x_0 = -\frac{a}{b}$ and $x_1 = -\frac{a}{b}+\frac{1}{2}$.
It suffices to prove that $f$ satisfies the inequality \eqref{eq:h_neg_compare_integrands} at $x = x_0$ and $x = x_1$.
Note that
\begin{align}
    x_1 + r &= ka+(i-1)\frac{a}{b}+c; 
        \label{eq:h_neg_x1p} \\
    x_1 + s
        &= \frac{1}{2}
            -\frac{1}{
                \frac{2a}{b}-1}\left(ka+\frac{ia}{b}-\frac{1}{2}+c\right) 
        = -\frac{1}{\frac{2a}{b}-1}\left(ka+(i-1)\frac{a}{b}+c\right).
        \label{eq:h_neg_x1q}
\end{align}
We can now show that $f$ satisfies \eqref{eq:h_neg_compare_integrands} at $x = x_1$.
We have
\begin{align*}
    f(x_1) &=\frac{1}{x_1+r} - \frac{1}{x_1+s} \\
    &= \frac{1}{ka+(i-1)\frac{a}{b}+c}
        + \left(\frac{2a}{b}-1\right)\frac{1}{ka+(i-1)\frac{a}{b}+c}
        \tag{by \eqref{eq:h_neg_x1p} and \eqref{eq:h_neg_x1q}} \\
    &= \frac{2a}{b} \cdot \frac{1}{ka+(i-1)\frac{a}{b}+c} \\
    &= \frac{2}{kb+i-1+\frac{cb}{a}} \\
    &< \frac{2}{kb+i-1+c}.
        \tag{since $\frac{b}{a} < 1$ and $c < 0$}
\end{align*}
We next show that $f(x_0) \leq f(x_1)$, which means that $f$ also satisfies \eqref{eq:h_neg_compare_integrands} at $x = x_0$.
Note that $x_0 = x_1 - \frac{1}{2}$. Then,
\begin{align}
    x_0 + r &= x_1+r-\frac{1}{2}; 
        \label{eq:h_neg_x0p}\\
    x_1 + s &= -\frac{1}{\frac{2a}{b}-1}(x_1+r); 
        \label{eq:h_neg_x1q_with_p}\\
    x_0 + s &= -\frac{1}{\frac{2a}{b}-1}(x_1+r) - \frac{1}{2} 
    = -\frac{1}{\frac{2a}{b}-1}\left(x_1+r+\frac{a}{b}-\frac{1}{2}\right).
        \label{eq:h_neg_x0q}
\end{align}
It follows that
\begin{align*}
    &f(x_1) - f(x_0) \\
    &= \left(\frac{1}{x_1+r} - \frac{1}{x_1+s}\right) 
        - \left(\frac{1}{x_0+r} - \frac{1}{x_0+s}\right) \\
    &= \frac{s-r}{(x_1+r)(x_1+s)} - \frac{s-r}{(x_0+r)(x_0+s)} \\
    &= (r-s)\left(\frac{2a}{b}-1\right)\left(\frac{1}{(x_1+r)^2}-\frac{1}{\left(x_1+r-\frac{1}{2}\right)\left(x_1+r+\frac{a}{b}-\frac{1}{2}\right)}\right).
    \tag{by \eqref{eq:h_neg_x0p}, \eqref{eq:h_neg_x1q_with_p}, and \eqref{eq:h_neg_x0q}}
\end{align*}
We now show that $f(x_1) - f(x_0) \geq 0$. Clearly, $\frac{2a}{b} - 1 > 0$. Moreover, recall that $r > \frac{a}{b}$ and $s < \frac{a}{b} - \frac{1}{2}$, so $r-s > \frac{1}{2} > 0$.
Furthermore,
\begin{align*}
    &\left(x_1+r-\frac{1}{2}\right)\left(x_1+r+\frac{a}{b}-\frac{1}{2}\right) - \left(x_1+r\right)^2 \\
    &= (x_1+r)\left(\frac{a}{b}-1\right)-\frac{1}{2}\left(\frac{a}{b}-\frac{1}{2}\right) \\
    &= \left(ka+(i-1)\frac{a}{b}+c\right)\left(\frac{a}{b}-1\right)-\frac{1}{2}\left(\frac{a}{b}-\frac{1}{2}\right) 
        \tag{by \eqref{eq:h_neg_x1p}}\\
    &= \left(ka+(i-1)\frac{a}{b}+c\right)\left(\frac{a}{b}-1\right)-\frac{1}{2}\left(\frac{a}{b}-1+\frac{1}{2}\right) \\
    &= \left(ka+(i-1)\frac{a}{b}+c-\frac{1}{2}\right)\left(\frac{a}{b}-1\right)-\frac{1}{4} \\
    &\geq \left(ka-\frac{3}{2}\right)\left(\frac{a}{b}-1\right)-\frac{1}{4} 
        \tag{since \protect{$i-1 \geq 0$, $c \geq -1$, and $\frac{a}{b}-1 > 0$}}\\
    &= \frac{a}{b}\left(k-\frac{3}{2a}\right)\left(a-b\right)-\frac{1}{4} \\
    &\geq \left(1-\frac{3}{4}\right)-\frac{1}{4} 
        \tag{since \protect{$a \geq b + 1 \geq 2$ and $k \geq 1$}}\\
    &= 0.
\end{align*}
Hence,
\[
    \frac{1}{(x_1+r)^2}-\frac{1}{\left(x_1+r-\frac{1}{2}\right)\left(x_1+r+\frac{a}{b}-\frac{1}{2}\right)} \geq 0,
\]
and so
\[
    f(x_0) \leq f(x_1) \leq \frac{2}{kb+i-1+c},
\]
as desired.
\end{proof}

\propintegertwovalueharmonic*

\begin{proof}
($\Leftarrow$) 
We first show this for $-1 \leq c < 0$.
Let $a, b \in \Zplus$ and $k, \ell, r \in \Zpluszero$ be given such that $a \geq b$ and $(k+1)b > \ell b + ra$.
Note that $h_c$ satisfies {\condintidengood}b by \Cref{prop:integer_idengood_harmonic} and hence {\condintidengood}a by \Cref{prop:condition_three_eqv}.
Then,
\begin{align*}
    &h_c((\ell+1)b+ra) - h_c(\ell b+ra) \\
    &= \sum_{t=0}^{b-1} \Delta_{h_c, \ell b+ra+t}(1) \tag{by telescoping sum} \\
    &\geq \sum_{t=0}^{b-1} \Delta_{h_c, (k+1)b-1+t}(1) \tag{by {\condintidengood}a and \protect{$(k+1)b - 1 \geq \ell b + ra$}} \\
    &= h_c((k+2)b-1) - h_c((k+1)b-1)  \tag{by telescoping sum} \\
    &\geq h_c((k+2)a-1) - h_c((k+1)a-1)
        \tag{by \Cref{lem:h_neg_scaled_diff_minus_1_decreasing} and \protect{$a \geq b$}} \\
    &= \sum_{t=0}^{a-1} \Delta_{h_c, (k+1)a-1+t} (1) \tag*{(by telescoping sum)} \\
    &> \sum_{t=0}^{a-1} \Delta_{h_c, (k+1)a+t} (1) \tag*{(by {\condintidengood}a)} \\
    &= h_c((k+2)a) - h_c((k+1)a) \tag*{(by telescoping sum)} \\
    &= \Delta_{h_c, k+1}(a) \\
    &> \Delta_{h_c, k+2}(1), \tag{by {\condintidengood}a}
\end{align*}
so $h_c$ satisfies {\condinttwovalue}.

Next, we show that $h_c$ satisfies {\condinttwovalue} for $0 \leq c \leq 1/\log 2 - 1$.
Note that $h_c$ satisfies {\condintidengood}b by \Cref{prop:integer_idengood_harmonic} and hence {\condintidengood}, and that $\Delta_{h_c, k}$ is non-decreasing on the domain of positive integers for each $k \in \Zplus$ by \Cref{lem:mod_harmonic_scaled_diff_increasing}.
By \Cref{prop:integer_twovalue_sufficient}, $h_c$ satisfies {\condinttwovalue}.

($\Rightarrow$) Since $h_c$ does not satisfy {\condintidengood} for $c > 1/\log 2 - 1$ by \Cref{prop:integer_idengood_harmonic}, it does not satisfy {\condinttwovalue} by \Cref{prop:integer_twovalue_sufficient}.
\end{proof}

\section{Welfare Function Must Be Strictly Increasing to Guarantee EF1}
\label{ap:welfare_function_strictly_increasing}

In this section, we justify why we require the function defining the additive welfarist rules to be strictly increasing.

Let $f: \Rpluszero \to \R \cup \{-\infty\}$ be a function continuous on $\Rplus$.
Suppose that $f$ is non-decreasing but not strictly increasing.
We show that there exists a simple instance such that the additive welfarist rule with $f$ does not always choose an EF1 allocation.

Since $f$ is not strictly increasing, there exist $a, b \in \Rplus$ with $a < b$ such that for all $x \in (a, b)$, we have $f(a) = f(x) = f(b)$.
Let $d$ be a positive integer larger than $3/(b-a)$.
We claim that there exists $c \in \Zplus$ such that $(c-1)/d, c/d, (c+1)/d \in (a, b)$.
Let $W = \{w \in \Zplus : w/d > a\}$. 
Since $\lceil ad+1 \rceil \in W$, $W$ is nonempty.
Let $c'$ be the smallest element in $W$.
We show that $c = c'+1$ satisfies the claim.
Since $c'-1 \notin W$, we have $(c'-1)/d \leq a$.
Then, $(c'+2)/d = (c'-1)/d + 3/d < a + (b-a) = b$.
This implies that $a < c'/d < (c'+1)/d < (c'+2)/d < b$, proving the claim.

Now, consider an instance with $n \geq 2$ agents and $m = nc$ goods such that $u_i(g) = 1/d$ for all $i \in N$, $g \in M$.
This instance is positive-admitting, identical-good, and normalized.
The only EF1 allocation $\mathcal{A} = (A_1, \ldots, A_n)$ gives $c$ goods to each agent, thereby achieving $\sum_{i \in N} f(u_i(A_i)) = nf(c/d)$.
Consider another allocation $\mathcal{B} = (B_1, \ldots, B_n)$ where agent $1$ receives $c+1$ goods, agent $2$ receives $c-1$ goods, and every other agent receives $c$ goods; this allocation is not EF1.
We have $\sum_{i \in N} f(u_i(B_i)) = f((c+1)/d) + f((c-1)/d) + (n-2)f(c/d) = nf(c/d) = \sum_{i \in N} f(u_i(A_i))$.
This means that if the additive welfarist rule with $f$ can choose an EF1 allocation (i.e., $\mathcal{A}$), it can also choose $\mathcal{B}$, which is not EF1.
Hence, this additive welfarist rule cannot always guarantee EF1.

\section{Values of the Modified Harmonic Function on Integers}
\label{ap:harmonic_extension}

Let $c > -1$ and $x \in \Zpluszero$.
Then,
\begin{align*}
    h_c(x) &= \int_{0}^{1} \frac{t^c-t^{x+c}}{1-t} \, \dd t \\
    &= \int_{0}^{1} \left( (t^c-t^{x+c}) \sum_{r=0}^\infty t^r \right) \, \dd t \tag{expansion of \protect{$(1-t)^{-1}$}} \\
    &= \int_{0}^{1} \sum_{r=0}^\infty (t^{r+c}-t^{r+c+x}) \, \dd t \\
    &= \sum_{r=0}^\infty \int_{0}^{1} (t^{r+c}-t^{r+c+x}) \, \dd t \tag{since \protect{$t^{r+c}-t^{r+c+x} \geq 0$}} \\
    &= \sum_{r=0}^\infty \left( \frac{1}{r+1+c} - \frac{1}{r+1+c+x} \right) \tag{\protect{$r+c, r+c+x \neq -1$}} \\
    &= \sum_{r=1}^\infty \left( \frac{1}{r+c} - \frac{1}{r+c+x} \right).
\end{align*}
When $x = 0$, this sum is equal to $0$, and hence $h_c(0) = 0$.
Otherwise, $x \in \Zplus$, and so
\begin{align*}
    h_c(x) &= \lim_{s \to \infty} \sum_{r=1}^s \left( \frac{1}{r+c} - \frac{1}{r+c+x} \right) \\
    &= \lim_{s \to \infty} \left( \sum_{r=1}^s \frac{1}{r+c} - \sum_{r=1}^s \frac{1}{r+c+x} \right) \\
    &= \lim_{s \to \infty} \left( \sum_{r=1}^x \frac{1}{r+c} + \sum_{r=x+1}^s \frac{1}{r+c}
    - \sum_{r=1}^{s-x} \frac{1}{r+c+x} - \sum_{r=s-x+1}^{s} \frac{1}{r+c+x} \right) \\
    &= \sum_{r=1}^x \frac{1}{r+c} - \lim_{s \to \infty} \sum_{r=s-x+1}^{s} \frac{1}{r+c+x} \\
    &= \sum_{r=1}^x \frac{1}{r+c} - \lim_{s \to \infty} \sum_{r=1}^{x} \frac{1}{s+c+r} \\
    &= \sum_{r=1}^x \frac{1}{r+c} - \sum_{r=1}^{x} \lim_{s \to \infty} \frac{1}{s+c+r} \\
    &= \sum_{r=1}^x \frac{1}{r+c}.
\end{align*}
This shows that
\[
    h_c(x) =
    \begin{cases}
        \frac{1}{1+c} + \frac{1}{2+c} + \cdots + \frac{1}{x+c} & \text{if $x \geq 1$}; \\
        0                                                      & \text{if $x = 0$}
    \end{cases}
\]
holds for $c > -1$.

Now, we consider the case $c = -1$.
We have
\begin{align*}
    h_{-1}(0) = \int_{0}^{1}{\frac{1-t^{-1}}{1-t}}\, \dd t = \int_{0}^{1} -\frac{1}{t} \, \dd t = -\infty,
\end{align*}
and for $x \in \Zplus$, we have
\begin{align*}
    h_{-1}(x) = \int_{0}^{1}{\frac{1-t^{x-1}}{1-t}}\, \dd t = \int_{0}^{1}{\frac{t^0-t^{(x-1)+0}}{1-t}}\, \dd t = h_0(x-1).
\end{align*}
This shows that
\[
    h_{-1}(x) =
    \begin{cases}
        1 + \frac{1}{2} + \cdots + \frac{1}{x-1} & \text{if $x \geq 2$}; \\
        0                                        & \text{if $x = 1$}; \\
        -\infty                                  & \text{if $x = 0$}
    \end{cases}
\]
holds.

\section{Other Properties of Additive Welfarist Rules}

\subsection{Normalized Instances with Two Agents}
\label{ap:real_normalized_two}

We continue the discussion from \Cref{subsec:real_normalized}.

Recall the definition of $\varphi_p$ from \Cref{subsec:examples}.
\citet{EckartPsVe24} proved that for $p \leq 0$, the additive welfarist rule with $\varphi_p$ guarantees EF1 for the class of normalized instances with two agents.
We shall show that there exist other additive welfarist rules that also ensure EF1 for this class of instances.

We rewrite Lemma 3 of \citet{EckartPsVe24} as follows.

\begin{proposition}[{\citet[Lemma 3]{EckartPsVe24}}]
\label{prop:real_normalized_two_pmean}
$\varphi_p$ satisfies {\condrealnormtwo} for all $p \leq 0$.
\end{proposition}

Note that \citet[Lemma 3]{EckartPsVe24} used a restricted version of {\condrealnormtwo} where $a, b, c, d \in \Rplus$ instead of $\Rpluszero$, but this can be easily extended to our version of {\condrealnormtwo}.
Indeed, if at least one of $a, b, c, d$ is zero, then $\min\{a, b\} \leq \min\{c, d\}$ and $ab < cd$ implies that at least one of $a$ and $b$ is zero while $c$ and $d$ are both positive.
Hence, for $p \leq 0$ we have $\varphi_p(a) + \varphi_p(b) = -\infty < \varphi_p(c) + \varphi_p(d)$, and so $\varphi_p$ satisfies {\condrealnormtwo}.

We show that any linear combination of such $\varphi_p$ also satisfies {\condrealnormtwo}.
More formally, let $\Psi$ be the set of all functions $\psi : \Rpluszero \to \R \cup \{-\infty\}$ for which there exist $k \in \Zplus$, $a_1, \ldots, a_k \in \Rplus$, and $p_1, \ldots, p_k \leq 0$ such that $\psi = \sum_{t=1}^k a_t \varphi_{p_t}$.
Then, the following proposition follows directly from \Cref{prop:real_normalized_two_sufficient,prop:real_normalized_two_pmean}.

\begin{proposition}
\label{prop:real_normalized_two_pmeanlincomb}
Let $\psi \in \Psi$.
Then, $\psi$ satisfies {\condrealnormtwo}.
Therefore, the additive welfarist rule with~$\psi$ guarantees EF1 for normalized instances with two agents.
\end{proposition}

Clearly, $\Psi$ contains functions other than those defining the $p$-mean rules.
However, it is possible that additive welfarist rules defined by two distinct functions are equivalent to each other.
For example, the MNW rule is equivalent to the additive welfarist rule defined by $\alpha\log x + \beta$ for any $\alpha \in \Rplus$ and $\beta\in\R$.
We prove via the following proposition that there exists an additive welfarist rule that is not equivalent to a $p$-mean rule but nevertheless guarantees EF1 for the class of normalized instances with two agents.

\begin{proposition}
\label{prop:real_normalized_two_notpmean}
Let $\psi = \varphi_0 + 40\varphi_{-1}$.
Then, the additive welfarist rule with $\psi$ is not equivalent to the additive welfarist rule with $\varphi_p$ for any $p \leq 0$.
\end{proposition}

\begin{proof}
For each $k \in \Zplus$, define the instance $\mathcal{I}_k$ with two agents and $m = 2k+1$ goods such that the utilities of the goods are as follows.
\begin{itemize}
    \item $u_1(g_j) = 4$ for $j \in \{1, \ldots, k\}$, $u_1(g_{k+1}) = 1$, and $u_1(g_j) = 0$ for $j \in \{k+2, \ldots, m\}$.
    \item $u_2(g_j) = 2$ for $j \in \{1, \ldots, 2k\}$, and $u_2(g_m) = 1$.
\end{itemize}
This instance is a normalized instance; moreover, it is positive-admitting since the allocation where agent $1$ receives $g_1$ and agent~$2$ receives the remaining goods gives positive utility to every agent.

Let $f$ be a function that satisfies {\condrealnormtwo}, and let $\mathcal{A} = (A_1, A_2)$ be an allocation chosen by the additive welfarist rule with $f$.
We have $\{g_{k+2}, \ldots, g_m\} \subseteq A_2$; otherwise, if some of these goods are with agent $1$, then transferring these goods to agent $2$ increases $u_2(A_2)$ and does not decrease $u_1(A_1)$, thereby increasing $\sum_{i=1}^2 f(u_i(A_i))$ and contradicting the assumption that $\mathcal{A}$ is chosen by the additive welfarist rule with $f$.
By \Cref{prop:real_normalized_two_sufficient}, $\mathcal{A}$ is EF1.
We have $|A_1| \geq k/2$; otherwise, if $|A_1| < k/2$, then we have $u_1(A_1) < 4(k/2) = 2k$ and $u_1(A_2 \setminus \{g\}) > 4(k/2) = 2k$ for all $g \in A_2$, making the allocation not EF1.
Moreover, we have $|A_1| \leq k$; otherwise, if $|A_1| > k$, then $A_1 = \{g_1, \ldots, g_{k+1}\}$ and $A_2 = \{g_{k+2}, \ldots, g_m\}$, and we have $u_2(A_2) = 2k-1 < 2k = u_2(A_1 \setminus \{g\})$, making the allocation not EF1.
Finally, we have $g_{k+1} \in A_2$; otherwise, exchanging $g_{k+1} \in A_1$ with $g \in A_2$ for any $g \in \{g_1, \ldots, g_k\}$ increases $u_1(A_1)$ and does not decrease $u_2(A_2)$, thereby increasing $\sum_{i=1}^2 f(u_i(A_i))$ and contradicting the assumption that $\mathcal{A}$ is chosen by the additive welfarist rule with $f$.
Overall, we have agent~$1$ receiving $x$ goods from $\{g_1, \ldots, g_k\}$ for some $k/2 \leq x \leq k$ and agent~$2$ receiving the remaining goods.
This gives $u_1(A_1) = 4x$ and $u_2(A_2) = 2(2k-x)+1 = 4k+1-2x$.

Note that $\varphi_p$ and $\psi$ both satisfy {\condrealnormtwo} by \Cref{prop:real_normalized_two_pmean,prop:real_normalized_two_pmeanlincomb}.

For $p = 0$, it can be verified (e.g., by exhaustion) that in the instance $\mathcal{I}_{20}$, we have $x = 20$ for~$\varphi_0$ and $x = 18$ for $\psi$, so the additive welfarist rules defined by the two respective functions are different.

Consider the case $p < 0$.
For each $k \in \Zplus$, let $F_k(y) = -(4y)^p - (4k+1-2y)^p$.
We have $\sum_{i=1}^2 \varphi_p(u_i(A_i)) = F_k(x)$.
Then,
\begin{align*}
    x = \argmax_{y \in [k/2, k] \cap \Z} F_k(y).
\end{align*}
Note that
\begin{align*}
    \frac{\partial F_k}{\partial y} = 2p\left((4k+1-2y)^{p-1} - 2(4y)^{p-1}\right)
\end{align*}
and
\begin{align*}
    \frac{\partial^2 F_k}{\partial y^2} = -4p(p-1)\left((4k+1-2y)^{p-2} + 4(4y)^{p-2}\right) < 0,
\end{align*}
so $F_k(y)$ (when viewed as a real function) achieves the maximum at $\partial F_k/\partial y = 0$.
Let $y^*$ be the value of $y$ at $\partial F_k/\partial y = 0$.
Then, we have $y^* = (4k+1)/(2+2^{(2p-1)/(p-1)})$, and $x = \lfloor y^* \rfloor$ or $x = \lceil y^* \rceil$.

Suppose that the additive welfarist rule with $\varphi_p$ chooses the same allocation as that with $\psi$ in $\mathcal{I}_{20}$, i.e., $x = 18$.
Then, we have $y^* < 19$, and with $k = 20$ it follows that $1/(2+2^{(2p-1)/(p-1)}) < 19/81$.
Now, consider the instance $\mathcal{I}_{100}$.
We have $y^* = (4k+1)/(2+2^{(2p-1)/(p-1)}) < 401(19/81) < 95$, so the additive welfarist rule with $\varphi_p$ chooses an allocation that gives fewer than $96$ goods to agent~$1$.
On the other hand, it can be verified that the additive welfarist rule with $\psi$ chooses an allocation that gives $96$ goods to agent~$1$.
Therefore, the two additive welfarist rules are not equivalent.
\end{proof}

Since the function $\psi$ in the statement of \Cref{prop:real_normalized_two_notpmean} satisfies $\psi \in \Psi$, \Cref{prop:real_normalized_two_pmeanlincomb,prop:real_normalized_two_notpmean} imply the following corollary.

\begin{corollary}
There exists an additive welfarist rule that is not equivalent to the $p$-mean rule for any $p\le 0$ but guarantees EF1 for every normalized instance with two agents.
\end{corollary}

By a similar argument as in \Cref{prop:real_normalized_two_notpmean}, it can be shown that for any $p_4 < p_3 \leq p_2 < p_1 \leq 0$, there exists $a \in \Rplus$ such that the additive welfarist rule with $\psi = \varphi_{p_1} + a\varphi_{p_2}$ is not a $p$-mean rule for any $p \leq 0$, and there exists $b \in \Rplus$ such that the additive welfarist rule with $\varphi_{p_1} + b\varphi_{p_2}$ is different from the additive welfarist rule with $\varphi_{p_3} + b\varphi_{p_4}$.

We now show that the additive welfarist rules defined by the modified logarithmic functions~$\lambda_c$ and modified harmonic functions $h_c$ do not guarantee EF1 for normalized instances with two agents.

\begin{proposition}
\label{prop:real_normalized_two_log}
For each $c > 0$, there exists a positive-admitting normalized instance with two agents such that the additive welfarist rule with $\lambda_c$ does not choose an EF1 allocation.
\end{proposition}

\begin{proof}
Let $c > 0$.
Suppose, for the sake of contradiction, that the additive welfarist rule with function $\lambda_c$ guarantees EF1 for all positive-admitting normalized instances with two agents.

Define the instance $\mathcal{I}_z$ for any $z \in \Rplus$ to be the normalized instance with $2$ agents, $9$ goods, and the following utilities:
\begin{itemize}
    \item $u_1(g_j) = 3z$ for $j \leq 8$ and $u_1(g_9) = z$;
    \item $u_2(g_j) = 5z$ for $j \leq 5$ and $u_2(g_j) = 0$ for $j \geq 6$.
\end{itemize}

Take an arbitrary $z \in \Rplus$ and consider the allocation $(A_1, A_2)$ chosen by the additive welfarist rule with function $\lambda_c$.
We must have $\{g_6, \ldots, g_9\} \subseteq A_1$; otherwise, the welfare can be increased by moving any of $g_6, \ldots, g_9$ to $A_1$.
Furthermore, by our assumption, this allocation must be EF1. Hence, we must have $\{g_6, \ldots, g_9\} \subsetneq A_1$; otherwise, agent~$1$ is not EF1 towards agent~$2$. Then, $u_1(A_1) \geq 13z$ and $u_2(A_2) \leq 20z$.

Fix some good $g_j \in A_1$ with $j \leq 5$. Note that
\begin{equation}
\label{eq:real_normalized_two_log_eq1}
\begin{aligned}
    \lambda_c (u_1(A_1)) - \lambda_c (u_1(A_1) - u_1(g_j)) 
    &= \log(u_1(A_1) + c) - \log(u_1(A_1) - u_1(g_j) + c) \\
    &= \log\left(\frac{u_1(A_1) + c}{u_1(A_1) - u_1(g_j) + c}\right) \\
    &\leq \log\left(\frac{13z + c}{10z + c}\right),
\end{aligned}
\end{equation}
and similarly,
\begin{equation}
\label{eq:real_normalized_two_log_eq2}
\begin{aligned}
    \lambda_c(u_2(A_2) + u_2(g_j)) - \lambda_c(u_2(A_2)) 
    &= \log(u_2(A_2) + u_2(g_j) + c) - \log(u_2(A_2) + c) \\
    &= \log\left(\frac{u_2(A_2) + u_2(g_j) + c}{u_2(A_2) + c}\right) \\
    &\geq \log\left(\frac{25z + c}{20z + c}\right).
\end{aligned}
\end{equation}
Observe that
\begin{align*}
    &\lim_{z \to 0^+} \left(\dv{z} \left( \log\left(\frac{13z + c}{10z + c}\right) - \log\left(\frac{25z + c}{20z + c}\right) \right) \right)\\
    &= \lim_{z \to 0^+} \left(\left(\frac{13}{13z + c} - \frac{10}{10z + c}\right) - \left(\frac{25}{25z + c} - \frac{20}{20z + c}\right)\right) 
    = \frac{3}{c} - \frac{5}{c} 
    < 0.
\end{align*}
When $z = 0$ we have
\[
    \log\left(\frac{13z + c}{10z + c}\right) = \log\left(\frac{25z + c}{20z + c}\right) = \log\left(\frac{c}{c}\right) = 0,
\]
so for small enough $z > 0$, 
\begin{equation}
\label{eq:real_normalized_two_log_eq3}
    \log\left(\frac{13z + c}{10z + c}\right) < \log\left(\frac{25z + c}{20z + c}\right).
\end{equation}
Fix such a small enough $z$, and consider the instance $\mathcal{I}_z$ and the allocation $(A_1, A_2)$ chosen by the additive welfarist rule with function $\lambda_c$.
Combining \eqref{eq:real_normalized_two_log_eq1}, \eqref{eq:real_normalized_two_log_eq2}, and \eqref{eq:real_normalized_two_log_eq3}, we get for this $z$ that
\begin{align*}
    \lambda_c (u_1(A_1)) - \lambda_c (u_1(A_1) - u_1(g_j)) < \lambda_c(u_2(A_2) + u_2(g_j)) - \lambda_c(u_2(A_2)).
\end{align*}
In other words, moving some $g_j \in A_1$ with $j \leq 5$ to $A_2$ strictly increases the welfare. 
This contradicts the assumption that the allocation $(A_1, A_2)$ is chosen by the additive welfarist rule with function~$\lambda_c$.
\end{proof}

\begin{proposition}
For each $c > -1$, there exists a positive-admitting normalized instance with two agents such that the additive welfarist rule with $h_c$ does not choose an EF1 allocation.
\end{proposition}

\begin{proof}
Let $c > -1$.
Suppose, for the sake of contradiction, that the additive welfarist rule with function $h_c$ guarantees EF1 for all positive-admitting normalized instances with two agents. Note that we use the extension of $h_c$ to the non-negative real domain, since the agents' utilities can take any non-negative real value.

For an arbitrary $z \in \Rplus$, consider the instance $\mathcal{I}_z$ defined in the proof of \Cref{prop:real_normalized_two_log} and the allocation $(A_1, A_2)$ chosen by the additive welfarist rule with function $h_c$.
As shown in the proof of \Cref{prop:real_normalized_two_log}, there exists $g_j \in A_1$ with $j \leq 5$, and moreover $u_1(A_1) \geq 13z$ and $u_2(A_2) \leq 20z$.

Fix some good $g_j \in A_1$ with $j \leq 5$. Note that
\begin{equation}
\label{eq:real_normalized_two_harmonic_eq1}
\begin{aligned}
    h_c (u_1(A_1)) - h_c (u_1(A_1) - u_1(g_j)) 
    &= \int_{0}^{1} \frac{t^c-t^{u_1(A_1)+c}}{1-t} \, \dd t 
    - \int_{0}^{1} \frac{t^c-t^{u_1(A_1)-u_1(g_j)+c}}{1-t} \, \dd t \\
    &= \int_{0}^{1}{\frac{t^{u_1(A_1)-u_1(g_j)+c}-t^{u_1(A_1)+c}}{1-t}}\, \dd t \\
    &\leq \int_{0}^{1}{\frac{t^{10z+c}-t^{13z+c}}{1-t}}\, \dd t,
\end{aligned}
\end{equation}
and similarly,
\begin{equation}
\label{eq:real_normalized_two_harmonic_eq2}
\begin{aligned}
    h_c(u_2(A_2) + u_2(g_j)) - h_c(u_2(A_2)) 
    &= \int_{0}^{1} \frac{t^c-t^{u_2(A_2)+u_2(g_j)+c}}{1-t} \, \dd t 
    - \int_{0}^{1} \frac{t^c-t^{u_2(A_2)+c}}{1-t} \, \dd t \\
    &= \int_{0}^{1}{\frac{t^{u_2(A_2)+c}-t^{u_2(A_2)+u_2(g_j)+c}}{1-t}}\, \dd t \\
    &\geq \int_{0}^{1}{\frac{t^{20z+c}-t^{25z+c}}{1-t}}\, \dd t.
\end{aligned}
\end{equation}
Observe that
\begin{align*}
    &\lim_{z \to 0^+} \left(\dv{z} \left( \int_{0}^{1}{\frac{t^{10z+c}-t^{13z+c}}{1-t}}\, \dd t - \int_{0}^{1}{\frac{t^{20z+c}-t^{25z+c}}{1-t}}\, \dd t \right) \right)\\
    &= \lim_{z \to 0^+} \left( \int_{0}^{1}{\frac{\left(10t^{10z+c}-13t^{13z+c} - 20t^{20z+c}+25t^{25z+c}\right)\log t}{1-t}}\, \dd t \right) 
    = \int_{0}^{1}{\frac{2t^{c}\log t}{1-t}}\, \dd t 
    < 0.
\end{align*}
When $z = 0$ we have
\[
    \int_{0}^{1}{\frac{t^{10z+c}-t^{13z+c}}{1-t}}\, \dd t = \int_{0}^{1}{\frac{t^{20z+c}-t^{25z+c}}{1-t}}\, \dd t = 0,
\]
so for small enough $z > 0$, 
\begin{equation}
\label{eq:real_normalized_two_harmonic_eq3}
    \int_{0}^{1}{\frac{t^{10z+c}-t^{13z+c}}{1-t}}\, \dd t < \int_{0}^{1}{\frac{t^{20z+c}-t^{25z+c}}{1-t}}\, \dd t.
\end{equation}
Fix such a small enough $z$, and consider the instance $\mathcal{I}_z$ and the allocation $(A_1, A_2)$ chosen by the additive welfarist rule with function $h_c$.
Combining \eqref{eq:real_normalized_two_harmonic_eq1}, \eqref{eq:real_normalized_two_harmonic_eq2}, and \eqref{eq:real_normalized_two_harmonic_eq3}, we get for this $z$ that
\begin{align*}
    h_c (u_1(A_1)) - h_c (u_1(A_1) - u_1(g_j)) 
    < h_c(u_2(A_2) + u_2(g_j)) - h_c(u_2(A_2)).
\end{align*}
In other words, moving some $g_j \in A_1$ with $j \leq 5$ to $A_2$ strictly increases the welfare. 
This contradicts the assumption that the allocation $(A_1, A_2)$ is chosen by the additive welfarist rule with function~$h_c$.
\end{proof}

\subsection{Integer-Valued Identical-Good Instances}
\label{ap:integer_idengood}

We continue the discussion from \Cref{subsec:integer_idengood}.

We state an interesting result regarding functions $f$ satisfying {\condintidengood}: the marginal increase $f(x+1)-f(x)$ grows as $\Theta(1/x)$ on the positive integers, which means that the function is similar to the logarithmic function in terms of growth rate.

\begin{proposition}
\label{prop:integer_idengood_marginal}
Let $f: \Rpluszero \to \R \cup \{-\infty\}$ be a strictly increasing function that satisfies {\condintidengood}.
Then, there exist $C_1, C_2 \in \Rplus$ such that for all $x \in \Zplus$, we have $C_1/x \leq f(x+1) - f(x) \leq C_2/x$.
\end{proposition}

\begin{proof}
We first prove the left inequality.
For all $x \in \Zplus$, we have $\Delta_{f, x}(1) \geq \Delta_{f, x+i}(1)$ for all $i \in \Zpluszero$ by {\condintidengood}a.
Therefore,
\begin{align*}
    f(x+1) - f(x) &= \Delta_{f, x}(1) \\
    &\geq \frac{1}{x} \sum_{i=0}^{x-1} \Delta_{f, x+i}(1) \\
    &= \frac{1}{x} \sum_{i=0}^{x-1} (f(x+i+1) - f(x+i)) \\
    &= \frac{1}{x} (f(2x) - f(x)) \tag{by telescoping sum} \\
    &= \frac{1}{x} \Delta_{f, 1}(x) \\
    &> \frac{1}{x} \Delta_{f, 2}(1) \tag{by {\condintidengood}} \\
    &= \frac{f(3) - f(2)}{x}.
\end{align*}
Noting that $C_1 := f(3) - f(2) \in \Rplus$ since $f$ is strictly increasing, this establishes the left inequality.

We now turn to the right inequality.
Let $x \geq 3$ be a positive integer, and let $y = \lfloor (x+1)/3 \rfloor$.
Note that $y \in \Zplus$ and $3y-1 \leq x$, so $\Delta_{f, 3y-1}(1) \geq \Delta_{f, x}(1)$ by {\condintidengood}a.
Moreover, we have $\Delta_{f, 3y-1}(1) \leq \Delta_{f, 2y+i}(1)$ for all $0 \leq i \leq y-1$ by {\condintidengood}a.
Therefore,
\begin{align*}
    f(x+1) - f(x) &= \Delta_{f, x}(1) \\
    &\leq \Delta_{f, 3y-1}(1) \\
    &\le \frac{1}{y} \sum_{i=0}^{y-1} \Delta_{f, 2y+i}(1) \\
    &= \frac{1}{y} \sum_{i=0}^{y-1} (f(2y+i+1) - f(2y+i)) \\
    &= \frac{1}{y} (f(3y) - f(2y)) \tag{by telescoping sum} \\
    &= \frac{1}{y} \Delta_{f, 2}(y) \\
    &< \frac{1}{y} \Delta_{f, 1}(1) \tag{by {\condintidengood}} \\
    &= \frac{f(2) - f(1)}{y}.
\end{align*}
Now, for $x \geq 3$, we have
\begin{align*}
    y = \left\lfloor \frac{x+1}{3} \right\rfloor
    \geq \frac{x+1}{3} - 1
    = \frac{x}{3} - \frac{2}{3}
    \geq \frac{x}{3} - \frac{2}{3} \cdot \frac{x}{3}
    = \frac{x}{9},
\end{align*}
and hence
\begin{align*}
    f(x+1) - f(x) < \frac{f(2) - f(1)}{y} \leq \frac{9(f(2) - f(1))}{x}.
\end{align*}
This shows that $f(x+1) - f(x) \leq C_2/x$ for all integers $x \geq 3$, where $C_2 := 9(f(2) - f(1))$ is positive since $f$ is strictly increasing.

Note that $f(x+1) - f(x) \leq C_2/x$ holds trivially for $x=1$.
For $x=2$, we have $f(3)-f(2) = \Delta_{f,2}(1) < \Delta_{f,1}(1) = f(2)-f(1) = (2/9)(C_2/2) \leq C_2/2$, and so the inequality $f(x+1) - f(x) \leq C_2/x$ holds as well.
\end{proof}

Next, we show that within the class of $p$-mean rules, the only rule that guarantees EF1 for integer-valued identical-good instances is the $0$-mean rule (i.e., MNW).

\begin{proposition}
For each $n\ge 2$, it holds that $p = 0$ if and only if for every positive-admitting integer-valued identical-good instance with $n$ agents, every allocation chosen by the $p$-mean rule is EF1.
\end{proposition}

\begin{proof}
($\Leftarrow$) 
By \Cref{thm:integer_idengood}, it suffices to show that $\varphi_p$ does not satisfy {\condintidengood} whenever $p\ne 0$.

First, we show that $\varphi_p$ does not satisfy {\condintidengood} for $p \geq 1$.
Since $2^p \geq 2$, we have $\Delta_{\varphi_p, 0}(1) = \varphi_p(1) - \varphi_p(0) = 1 \leq 2^p - 1 = \varphi_p(2) - \varphi_p(1) = \Delta_{\varphi_p, 1}(1)$, so {\condintidengood} is violated.

Next, we show that $\varphi_p$ does not satisfy {\condintidengood} for $0 < p < 1$.
For any $x \in \Zplus$, we have
\begin{align*}
    x(\varphi_p(x+1) - \varphi_p(x)) &= x\int_x^{x+1} \varphi'_p(t) \; \dd t \\
    &= x\int_x^{x+1} pt^{p-1} \; \dd t \\
    &\geq x\int_x^{x+1} p(x+1)^{p-1} \; \dd t \tag{since \protect{$p > 0$} and \protect{$p-1 < 0$}} \\
    &= xp(x+1)^{p-1}.
\end{align*}
Note that $\lim_{x \to \infty} xp(x+1)^{p-1} = \lim_{x \to \infty} p(x+1)^p = \infty$, so there does not exist a constant $C \in \Rplus$ such that $x(\varphi_p(x+1) - \varphi_p(x)) \leq C$ for all $x \in \Zplus$.
By \Cref{prop:integer_idengood_marginal}, $\varphi_p$ does not satisfy {\condintidengood}.

Finally, we show that $\varphi_p$ does not satisfy {\condintidengood} for $p < 0$.
For any $x \in \Zplus$, we have
\begin{align*}
    x(\varphi_p(x+1) - \varphi_p(x)) &= x\int_x^{x+1} \varphi'_p(t) \; \dd t \\
    &= x\int_x^{x+1} - pt^{p-1} \; \dd t \\
    &\leq x\int_x^{x+1} - px^{p-1} \; \dd t \tag{since \protect{$p, p-1 < 0$}} \\
    &= -px^p.
\end{align*}
Note that $\lim_{x \to \infty} -px^p = 0$, so there does not exist a constant $C \in \Rplus$ such that $x(\varphi_p(x+1) - \varphi_p(x)) \geq C$ for all $x \in \Zplus$.
By \Cref{prop:integer_idengood_marginal}, $\varphi_p$ does not satisfy {\condintidengood}.

($\Rightarrow$) This follows from the result of \citet{CaragiannisKuMo19} that MNW guarantees EF1.
\end{proof}

\end{document}